\documentclass[12pt]{article}
\usepackage[usenames,dvipsnames,svgnames,table]{xcolor} 
\usepackage{kotex}
\usepackage{graphicx}
\usepackage{amssymb, amsmath, mathtools}
\usepackage{verbatim}
\usepackage[round]{natbib}
\usepackage[colorlinks=true,allcolors=black]{hyperref}
\usepackage[affil-it]{authblk}
\usepackage{mathrsfs}
\usepackage{here}
\usepackage{makecell}
\usepackage{tikz}
\usepackage[shortlabels]{enumitem}
\usepackage{yfonts}
\usepackage{comment}
\usepackage{longtable}
\usepackage{bm, color, amsthm}
\usetikzlibrary{circuits.logic.US,circuits.logic.IEC,fit}
\usepackage{multirow}
\usepackage{float} 

\usepackage{bibunits}
\defaultbibliographystyle{abbrvnat} 
\defaultbibliography{refs}

\usepackage[noabbrev,capitalise]{cleveref}
\creflabelformat{equation}{#2\textup{#1}#3}

\usepackage{color, colortbl}
\usepackage{booktabs}
\usepackage{geometry}
\usepackage{array}
\usepackage{caption}
\usepackage{subcaption}
\usepackage{dsfont}
\definecolor{Yellow}{rgb}{1,1,0.6}
\usepackage{setspace}
\usepackage{aliascnt}

\usepackage{algorithm}
\usepackage{algpseudocode}
\usepackage{adjustbox}
\usepackage{lineno} 

\algrenewcommand\algorithmicrequire{\textbf{Input:}}
\algrenewcommand\algorithmicensure{\textbf{Output:}}

\newcommand{\jc}[1]{\textcolor{violet}{#1}}

\newcommand{\jeff}[1]{\textcolor{magenta}{#1}}

\hypersetup{
 colorlinks,
 citecolor=Blue,
 linkcolor=Blue,
 urlcolor=Blue}

\newtheorem{theorem}{Theorem}[section]
\newtheorem{lemma}[theorem]{Lemma}
\newtheorem{proposition}[theorem]{Proposition}
\newtheorem{corollary}[theorem]{Corollary}

\newaliascnt{condition}{theorem}
\newtheorem{condition}[condition]{Condition}
\aliascntresetthe{condition}
\crefname{condition}{Condition}{Conditions}
\Crefname{condition}{Condition}{Conditions}

\theoremstyle{definition}
\newtheorem{definition}{Definition}

\DeclareMathOperator*{\argmin}{arg\,min}
\newcommand{\E}{\mathbb{E}}

\newcommand{\V}{{\mathrm{Var}}}
\newcommand{\C}{{\mathrm{Cov}}}

\newcommand{\muhat}{{\hat{\mu}}}
\newcommand{\diag}{\mbox{diag}}

\newcommand{\bea}{\begin{eqnarray*}}
	\newcommand{\eea}{\end{eqnarray*}}
\newcommand{\bean}{\begin{eqnarray}}
	\newcommand{\eean}{\end{eqnarray}}
\newcommand{\benu}{\begin{enumerate}}
	\newcommand{\eenu}{\end{enumerate}}

\newcommand{\bbR}{\mathbb{R}}
\newcommand{\bbN}{\mathbb{N}}

\newcommand{\cF}{\mathcal{F}}

\newcommand{\cM}{\mathcal{M}}
\newcommand{\cN}{\mathcal{N}}

\newcommand{\cX}{\mathcal{X}}

\providecommand{\keywords}[1]{\textbf{\textit{Keywords---}} #1}

\newcommand*\oline[1]{%
  \,\vbox{%
    \hrule height 0.5pt
    \kern0.25ex
    \hbox{%
      \kern-0.1em
      \ifmmode#1\else\ensuremath{#1}\fi
      \kern0em
    }
  }
}

\newcommand{\branch}[4]{
\left\{
	\begin{array}{ll}
		#1  & \mbox{if } #2 \\
		#3 & \mbox{if } #4
	\end{array}
\right.
}

\title{Robust model selection using likelihood as data}

\author[1]{Jongwoo Choi\footnote{Corresponding author, E-mail: jongwoo.choi@uconn.edu}}
\author[1]{Neil A. Spencer}
\author[2]{Jeffrey W. Miller}
\affil[1]{Department of Statistics, University of Connecticut}
\affil[2]{Department of Biostatistics, Harvard T.H. Chan School of Public Health}

\date{} 

\begin{document}
\begin{bibunit}

\maketitle

\begin{abstract}
Model selection is a central task in statistics, but standard methods are not robust in misspecified settings where the true data-generating process (DGP) is not in the set of candidate models.
The key limitation is that existing methods---including information criteria and Bayesian posteriors---do not quantify uncertainty about how well each candidate model approximates the true DGP.
In this paper, we introduce a novel approach to model selection based on modeling the likelihood values themselves.
Specifically, given $K$ candidate models and $n$ observations, we view the $n\times K$ matrix of negative log-likelihood values as a random data matrix and observe that the expectation of each row is equal to the vector of Kullback--Leibler divergences between the $K$ models and the true DGP, up to an additive constant.  We use a multivariate normal model to estimate and quantify uncertainty in this expectation, providing calibrated inferences for robust model selection under misspecification.
The procedure is easy to compute, interpretable, and comes with theoretical guarantees, including consistency.
\end{abstract}


\newpage

\section{Introduction}\label{sec:intro}

A good statistical model should be flexible enough to capture the essential structure of the data, yet simple enough to interpret, communicate, and use in practice. Model selection helps balance these goals by choosing among competing models of varying complexity and explanatory power \citep{claeskens2008model}. In most applications, the true data-generating process (DGP) is unknown and all of the candidate models are misspecified \citep{bernardo2000bayesian}. 
In such cases, it is desirable to find a model that is as close as possible to the true DGP while remaining sufficiently simple to interpret and use.

Selection criteria such as the Akaike Information Criterion (AIC; \citealp{akaike1974new}) and Bayesian Information Criterion (BIC; \citealp{schwarz1978estimating}) are based on the log-likelihood, which can be viewed as estimating the relative closeness of each model to the true DGP.  However, they do not quantify uncertainty in these estimates or in the choice of model.  Meanwhile, Bayesian posteriors quantify uncertainty in the choice of model \citep{wasserman2000bayesian}, but posterior probabilities do not measure how far each model is from the truth; a model may receive a high probability simply because the alternatives are worse.   

When none of the candidate models are correct, these information criteria and Bayesian posteriors asymptotically concentrate on the model with the smallest Kullback--Leibler (KL) divergence from the true DGP, even if there is a much simpler model that is nearly as good \citep{berk1966limiting}.  Further, these methods tend to be highly unstable when more than one model attains the minimal KL divergence, even approximately \citep{huggins2023reproducible}.  Due to these issues, the standard leading methods do not provide well-calibrated inferences about the utility of each model in a way that is robust to misspecification.

\begin{figure}
    \centering
    \includegraphics[width=0.99\linewidth]{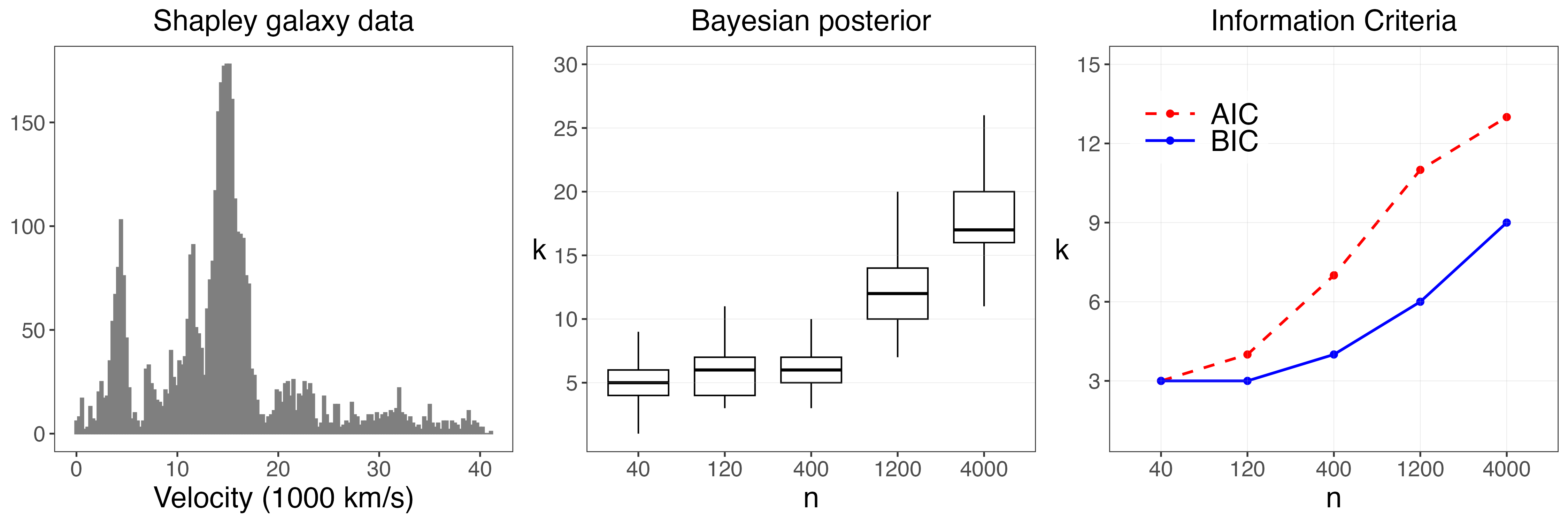}
    \captionsetup{font=small}
    \caption{Limitations of existing model selection methods. 
    \emph{(Left)} Histogram of the Shapley galaxy velocities ($\times 10^3$ km/s). We trim the top $0.5\%$ of data points to remove the extremely long right tail extending beyond $45{,}000$ km/s.
    \emph{(Middle)} Boxplots summarizing draws from the Bayesian posterior on $k$ for a Gaussian mixture model with $k$ components, for sample sizes $n \in \{40, 120, 400, 1200, 4000\}$. As $n$ grows, the posterior shifts toward larger $k$, a typical phenomenon under misspecification.
    \emph{(Right)} Values of $k$ selected by AIC (red) and BIC (blue) for each $n$. 
    These results illustrate lack of robustness and inability to quantify uncertainty in how well a $k$-component Gaussian mixture can approximate the true DGP.}
    \label{fig:intro_example}
\end{figure}

To illustrate these limitations, consider the problem of choosing the number of components $k$ in a finite mixture of Gaussians. 
Figure \ref{fig:intro_example} (left) shows a histogram of radial velocity measurements of $4215$ galaxies in the Shapley supercluster, a large concentration of gravitationally interacting galaxies \citep{drinkwater2004large}.  
We apply AIC, BIC, and a Bayesian model with a prior on the number of components to increasing subsets of this data set.
In Figure \ref{fig:intro_example} (middle), we see that the posterior distribution concentrates on larger and larger values of $k$ as the sample size $n$ increases, and does not provide any insight into how close each model is to the true DGP. Similarly,  Figure \ref{fig:intro_example} (right) shows that AIC and BIC select increasing values of $k$ as $n$ grows, and furthermore, they do not provide any uncertainty quantification in the choice of $k$.
See \cref{supp:mix_details} for details.



In this paper, we introduce a new framework for model selection that overcomes these issues by jointly quantifying uncertainty in the expected log-likelihoods of the candidate models.
More precisely, let $Z_{i k}$ denote the negative log-likelihood for model $k$ on observation $i$, assuming for the moment that there are no unknown parameters.
Then $\mu_k = E(Z_{i k})$ equals the KL divergence between the true DGP and model $k$, plus the (unknown) entropy of the true DGP.
We propose viewing $Z_{i k}$ as data---hence the term ``Likelihood as Data''---and using a Bayesian multivariate normal model to compute the posterior distribution of $\mu = (\mu_1,\ldots,\mu_K)$.
This makes it straightforward to perform Bayesian inference for a range of questions, including not only which model is closest to the true DGP, but how much closer a given model is than another, and which is the simplest model within a specified tolerance of being closest to the true DGP.


A key feature of our approach is that it is fully robust to misspecification of the candidate models since it is based on inferring the $\mu_k$'s, which are well defined and interpretable without requiring any assumptions of model correctness. 
We prove that the method is consistent and remains stable under ties -- that is, when multiple models are the same distance from the true DGP.
We propose a procedure for finding the simplest class of models within a specified tolerance of being closest to the true DGP, and choosing the best-fitting model (or models) within this class.
In a range of simulations and real-data applications, we find that the proposed approach yields reliable and informative inferences for model selection in the presence of misspecification.

The paper is organized as follows. \cref{sec:method} introduces our proposed methodology, and \cref{sec:computation} provides algorithms for computing the method. In \cref{sec:theory}, we provide theoretical results. \cref{sec:examples} contains empirical demonstrations on Gaussian mixture models, sparse multivariate normal models, microbial ecology models, and population structure admixture models in genetics. \cref{sec:conclusion} concludes with a brief discussion.

\section{Methodology}\label{sec:method}

In this section, we introduce the \emph{Likelihood as Data} framework (\cref{sec:method:lad}), describe our robust model selection criterion (\cref{sec:method:robust}), discuss interpretation of the selection threshold (\cref{sec:method:delta}), extend the method to parametrized models (\cref{sec:method:param}), and provide an adjustment for overfitting bias in parametrized models (\cref{sec:method:param:overfitting}).

Suppose $X_1,\ldots, X_n \in \cX$ are independent and identically distributed (i.i.d.)\ observations from some unknown distribution $P_0$, where $\cX$ denotes the sample space. Assume $P_0$ admits a density $f_0$ with respect to a sigma-finite measure $\lambda$, for example, Lebesgue measure for continuous observations or counting measure for discrete observations. 

Consider a finite collection of candidate models $P_1, \ldots, P_K$ with associated densities $f_1, \ldots, f_K$ (also with respect to $\lambda$), respectively. To clearly explain the main idea of our proposed method, we assume for now that there are no unknown parameters to be estimated in any of the candidate models;  however, in \cref{sec:method:param}, we extend the method to handle unknown parameters.  
We do not assume $P_0$ coincides with one of the candidate models, that is, all $K$ models may be misspecified.

\subsection{Likelihood as Data (LaD) framework}
\label{sec:method:lad}

Let $Z_{i k} := -\log(f_k(X_i))$ be the negative log-likelihood value for each observation $i\in\{1,\ldots,n\}$ and each model $k\in\{1,\ldots,K\}$, and collect these values into a matrix $Z = [Z_{i k}]\in\mathbb{R}^{n\times K}$. 
The key idea is to treat each vector $Z_i := (Z_{i 1},\dots,Z_{i K})^\texttt{T} \in \mathbb{R}^K$ as a data point, and analyze the resulting data set, $Z_1,\ldots,Z_n$. In this way, we regard the likelihood values as data in their own right, which is the foundation of the LaD approach.

As with any data set, it is natural to estimate the mean $\mu^0 := \mathbb{E}(Z_1) \in\mathbb{R}^K$ and quantify uncertainty in it. 
In the LaD setting, $\mu^0$ has a special interpretation since
\[
    \mu^0_k = \mathbb{E}(Z_{i k}) = \mathbb{E}(-\log(f_k(X_i))) = D(f_0 \| f_k) + H(f_0)
\]
where $H(f_0) := -\int_{\cX} f_0 (x) \log f_0(x) \lambda (dx)$ is the entropy of the true DGP and $D(f_0 \| f_k)$ is the Kullback--Leibler (KL) divergence between the true DGP and model $k$,
\begin{equation}\label{eq:kl}
    D(f_0 \| f_k) := \int_\cX f_0(x) \Big(\log \frac{f_0(x)}{f_k(x)}\Big) \lambda(dx).
\end{equation}
Thus, $\mu^0 = (\mu^0_1,\ldots,\mu^0_K)^\mathtt{T}$ is a vector containing the KL divergences between the true DGP and the $K$ models, up to a common additive constant $H(f_0)$.

Estimation and inference for $\mu^0$ are straightforward.
Let $\oline{Z}_n := \frac{1}{n}\sum_{i=1}^n Z_i$ and observe that since $Z_1,\ldots,Z_n$ are i.i.d., we have
\begin{equation}
\label{eq:lln_clt}
\begin{split}
     \oline{Z}_n &\xrightarrow[n\to\infty]{\mathrm{a.s.}} \mu^0, \\
    \sqrt{n}(\oline{Z}_n - \mu^0) &\xrightarrow[n\to\infty]{\mathrm{d}} \mathcal{N}(0, \Sigma^0)
\end{split}
\end{equation}
where $\Sigma^0 = \mathrm{\C}(Z_1) \in \mathbb{R}^{K \times K}$, 
by the strong law of large numbers and the central limit theorem (CLT), assuming $\Sigma^0$ exists and is positive definite (\cref{cond1}).
Thus, one could estimate $\mu^0$ by the sample mean $\oline{Z}_n$ and construct Wald-type confidence sets for $\mu^0$ based on an estimate of $\Sigma^0$ such as the sample covariance matrix.
Instead, however, we use a Bayesian multivariate normal model since it facilitates inference for many questions of interest and tends to be more stable at small sample sizes.

Specifically, we employ a Bayesian model where
\begin{equation}
\begin{split}
\label{eq:bayesian-model}
    Z_1,\ldots,Z_n\mid \mu,\Sigma\; &\overset{\mathrm{iid}}{\sim}\; \mathcal{N}(\mu,\Sigma),\\
    (\mu,\Sigma) &\sim \pi(\mu,\Sigma),
\end{split}
\end{equation}
and $\pi(\mu,\Sigma)$ is a prior on the unknown mean $\mu$ and covariance matrix $\Sigma$; see \cref{sec:method:bayesinf}.
We write $\mu$, $\Sigma$ for the parameters of this Bayesian model, and $\mu^0$, $\Sigma^0$ for the true values.
Uncertainty quantification for $\mu$ enables statistical inference for the relative distances from the true DGP. Note that looking only at marginal uncertainties in each entry $\mu_k$ alone is not particularly useful, for at least two reasons: first, $H(f_0)$ is unknown and can be very difficult to estimate; second, when the log-likelihood values for models $j$ and $k$ are strongly correlated, it can happen that there is very little uncertainty in the difference $\mu_j - \mu_k$ even if there is high uncertainty in $\mu_j$ and $\mu_k$. Hence, multivariate inference for $\mu = (\mu_1,\ldots,\mu_K)^\texttt{T}$ is far more informative than univariate inference for each $\mu_k$.

The LaD approach does not require any assumptions regarding the form of the models or the true DGP, except for the moment conditions needed for the CLT to apply. 
While the Bayesian model does assume the negative log-likelihoods $Z_1,\ldots,Z_n$ are Gaussian, which is unlikely to hold, this discrepancy has minimal impact on inferences for $\mu$ since the posterior only depends on  $Z_1,\ldots,Z_n$ through the sufficient statistics, namely the sample mean and sample covariance of the $Z_i$'s.
Consequently, we find that LaD provides correctly calibrated inferences for model selection, despite the incorrectness of the Gaussian model.


\subsection{Robust model selection using LaD}
\label{sec:method:robust}


Using the LaD framework in \cref{sec:method:lad}, one can perform frequentist or Bayesian inference for various hypotheses, such as whether model $k$ minimizes the KL divergence to the true DGP (that is, $\mu^0_k = \min_{j} \mu^0_{j}$) or whether model $k$ is closer than model $j$ to the true DGP (that is, $\mu^0_k < \mu^0_{j}$).
To perform robust model selection, we aim to identify the simplest model that is ``close enough'' to having minimal KL divergence.
Specifically, we consider any model whose KL divergence is within a tolerance  $\delta > 0$ of the minimum to be sufficiently close.
This enables more parsimonious models to be favored when their fit is roughly comparable to that of more complex models. 

\begin{definition}
\label{def:robust}
Given $\delta > 0$, under the setup above, we say model $k$ is \emph{$\delta$-optimal} if 
\begin{equation}
\label{eq:delta-optimal}
    D(f_0 \| f_k) \leq \min_j D(f_0 \| f_j) + \delta.
\end{equation}
\end{definition}


Note that \cref{eq:delta-optimal} is equivalent to $\mu^0_k \leq \min_j \mu^0_j + \delta$.
While our focus in this paper is on KL divergence, \cref{def:robust} extends naturally to other measures of discrepancy between distributions, such as Wasserstein distance or Bregman divergence. 

\subsubsection{Best minimal-complexity $\delta$-optimal model}

Let $M_\delta(\mu)$ denote the set of $\delta$-optimal models as a function of $\mu$, that is,
\begin{equation}
\label{eq:M}
    M_\delta(\mu) := \Big\{k: \mu_k \leq \min_j \mu_j + \delta\Big\}.
\end{equation}
Suppose $c:\{1,\ldots,K\}\to [0,\infty)$ assigns a complexity $c(k)$ to model $k$, with smaller values corresponding to simpler models.
We refer to a set of models with the same $c(k)$ as a complexity class.
Among $\delta$-optimal models, the minimal complexity attained is
\begin{equation}
\label{eq:min_c}
    c_\delta^*(\mu) := \min \{c(k): k \in M_\delta(\mu) \}.
\end{equation}
Since there may be multiple $\delta$-optimal models that attain this minimal complexity $c_\delta^*(\mu)$, among these we prefer the ones with best fit, that is, with minimal KL from the true DGP.
Thus, we define the \emph{best minimal-complexity $\delta$-optimal model(s)} as 
\begin{equation}
\label{eq:W}
    M^*_\delta(\mu) := \argmin_{k \,:\, c(k) = c_{\delta}^*(\mu)} \mu_k.
\end{equation}

The set $M^*_\delta(\mu)$ contains those models that (i) have the lowest possible complexity  attainable by a $\delta$-optimal model, and (ii) have the smallest KL divergence within that complexity class. Typically, $M^*_\delta(\mu)$ will consist of a single model, but it can contain multiple models if there are ties $\mu_j = \mu_k$. For robust model selection, our goal is to quantify uncertainty about which model or models belong to $M^*_\delta(\mu^0)$ for the true $\mu^0$.




\subsubsection{Reproducible inference for $M^*_\delta(\mu)$}
\label{sec:method:sol}





A natural way to infer the true set $M^*_\delta(\mu^0)$ would be to consider the posterior probability that each model $k$ is in $M^*_\delta(\mu)$, that is, $P(k\in M^*_\delta(\mu) \mid Z_{1:n})$, where $\mu\mid Z_{1:n}$ is distributed according to the posterior under a Bayesian model as in \cref{eq:bayesian-model}.
However, if we use a continuous prior on $\mu$ and the true $M^*_\delta(\mu^0)$ contains multiple models, then $P(k\in M^*_\delta(\mu) \mid Z_{1:n})$ is not an accurate representation of uncertainty. For example, when there is correlation among models, some models that belong to $M^*_\delta(\mu^0)$ may receive arbitrarily small posterior probability; see \cref{sec:theory}.  
One solution would be to use a prior that allows $\mu_j = \mu_k$ with positive probability for any $j,k$, but this would make posterior inference more computationally burdensome and complicate theoretical analysis.  




Instead, we propose a computationally and theoretically attractive alternative based on a smooth selection function that relaxes the hard minimum in \cref{eq:W}. 
For each model $k$ and sample size $n$, we define a within-class soft-selection score,
\begin{align}
\label{eq:softmin}
    r_{n k}(\mu) := \exp\!\big(-\alpha_n (\mu_k - \mu_{\min,c(k)}) \big),
\end{align}
where $\mu_{\min,c(k)} = \min_{j \,:\, c(j) = c(k)} \mu_j$ and $\alpha_n >0$ is a temperature parameter such that $\alpha_n\to\infty$ and $\alpha_n = o(\sqrt{n})$.
In all of the examples in this paper, we use $\alpha_n = n^{0.45}$.
Note that $r_{n k}(\mu)\in (0,1]$, with $r_{n k}(\mu) = 1$ for models $k$ having minimal $\mu_k$ within their complexity class $c(k)$.
We now introduce our proposed model selection score.

\begin{definition}
\label{def:score}
The \emph{smooth LaD criterion (SLC)} score is 
\begin{equation}\label{eq:pkn2}
    w_{\delta}(k \mid Z_{1:n}) = P\big(c_\delta^*(\mu) = c(k) \mid Z_{1:n} \big) \, 
         \E\big( r_{n k}(\mu) \mid Z_{1:n} \big)
\end{equation}
where $\mu\mid Z_{1:n}$ is distributed according to posterior of a model as in \cref{eq:bayesian-model}.
\end{definition}

An estimator of $M_\delta^*(\mu^0)$ can be constructed as $\hat{M}^*_\delta = \{k : w_{\delta}(k \mid Z_{1:n}) > \omega\}$, where $\omega$ is a threshold between $0$ and $1$. 
In \cref{sec:theory}, we justify the SLC score by showing that it is asymptotically consistent, that is, $w_{\delta}(k \mid Z_{1:n})$ concentrates on the true target set $M^*_\delta(\mu^0)$ as $n$ grows. 
The SLC score $w_{\delta}(k \mid Z_{1:n})$ is defined as the product of (i) a between-class selection factor and (ii) a within-class selection factor. The first factor represents the probability that model $k$ attains the minimal complexity among all $\delta$-optimal models. The second factor quantifies model $k$'s relative performance within its complexity class. 
We compute each of the two factors with a Monte Carlo approximation using posterior samples from our Bayesian model; see \cref{alg:proc} for details.

For interpretation, if there is only one model per complexity class, then $r_{n k}(\mu) = 1$ for all $k$, so \cref{eq:pkn2} simplifies to $w_{\delta}(k \mid Z_{1:n}) = P\big(\argmin_{j \in M_\delta(\mu)} c(j) = \{k\} \mid Z_{1:n} \big)$, because $c_\delta^*(\mu) = c(k)$ if and only if $\argmin_{j \in M_\delta(\mu)} c(j) = \{k\}$.
More generally, whenever $|M^*_\delta(\mu^0)| = 1$, the SLC score $w_{\delta}(k \mid Z_{1:n})$ behaves similarly to the posterior probability $P(k \in M^*_\delta(\mu)\mid Z_{1:n})$.
Meanwhile, when $|M^*_\delta(\mu^0)| > 1$, the SLC score $w_{\delta}(k \mid Z_{1:n})$ remains stable since it does not concentrate on a single model.  This is because, under the posterior, $\mu_k - \mu_{\min,c(k)} = O_p(n^{-1/2})$ for the best fitting models $k$ with complexity $c(k)$, and we choose $\alpha_n = o(\sqrt{n})$, which makes $\E\big(r_{n k}(\mu) \mid Z_{1:n} \big) \to 1$ as $n\to\infty$.

\subsection{Interpreting the tolerance $\delta$}
\label{sec:method:delta}

The tolerance $\delta$ makes the selection procedure robust to model misspecification, since it allows one to select models that have a KL divergence within $\delta$ of being minimal.  Smaller values of $\delta$ enforce stricter optimality but may only be satisfied by complex models, whereas larger values of $\delta$ encourage parsimony at the cost of fit.
Rather than choosing one value of $\delta$, we generally recommend computing the SLC score for a range of $\delta$ values, as demonstrated in the examples in \cref{sec:examples}.
Since the units of $\delta$ may not be intuitively clear, in this section we provide a technique for putting $\delta$ on an interpretable scale.

Motivated by an idea presented by \citet{pescador2025adjusting} (Section 4.3), we augment the set of candidate models with  a deliberately misspecified ``noise'' model with density $f_{\mathrm{noise}}$, chosen to represent a baseline that does not capture any meaningful signal in the data. 
Under the noise model, the expected negative log-likelihood is 
\[
    \mu^0_{\mathrm{noise}} := \E (-\log f_{\mathrm{noise}}(X_i)) = D(f_0 \| f_{\mathrm{noise}}) + H(f_0).
\]
This noise model should satisfy: (i) full support over $\cX$, so that $D(f_0 \| f_{\mathrm{noise}}) < \infty$, (ii) stability under estimation, so that $\mu^0_{\mathrm{noise}}$ can be estimated with little uncertainty, and (iii) minimal signal (for example, an intercept-only or null model).
Now, we define
\begin{equation}\label{eq:delta-calibration}
    \tau := \delta / (\mu^0_{\mathrm{noise}} - \mu^0_\mathrm{min})
\end{equation}
where $\mu^0_\mathrm{min} = \min_k \mu^0_k$.
Note that, equivalently, $\tau = \delta / (D_{\mathrm{noise}} - D_{\mathrm{min}})$, where $D_{\mathrm{noise}} := D(f_0 \| f_{\mathrm{noise}})$ and $D_{\mathrm{min}} := \min_k D (f_0 \| f_k )$.
Thus, $\tau$ is a rescaled version of $\delta$ that puts it in interpretable units: Specifically, $\tau$ represents the proportion of the explainable information that we are willing to sacrifice in return for model simplicity.  A value of $\tau>1$ signifies that $\delta$ is excessively large, such that even it tolerates the noise model.

By the definition of the $\delta$-optimal set $M_{\delta}(\mu^0)$ (\cref{eq:M}), for any model $k \in M_{\delta}(\mu^0)$, 
\[
D(f_0 \| f_k) \le D_{\mathrm{min}} + \tau \, (D_{\mathrm{noise}} - D_{\mathrm{min}}).
\]
Rearranging this expression, we obtain
\[
\frac{D_{\mathrm{noise}} - D(f_0 \| f_k)}{D_{\mathrm{noise}} - D_{\mathrm{min}}} \ge 1 - \tau.
\]
This yields a natural interpretation: A model
is $\delta$-optimal if it recovers at least $100(1-\tau)\%$ of the total possible improvement from the noise model to the best available model, in terms of KL divergence.
Thus, the choice of $\delta$ becomes interpretable as a relative tolerance expressed on this information-theoretic scale.

Suitable choices for the noise model depend on the context. For example, for multivariate Gaussian models, we use $f_{\mathrm{noise}} = \cN(0, I)$. For generalized linear models, we use an intercept-only model (no covariates). For mixture models, a suitable noise baseline could be a uniform density over the observed data range or a one-component mixture.

\subsection{Extending to unknown parameters and other loss functions}
\label{sec:method:param}

Up to now, we have assumed that each model consists of a single distribution with known parameters, but in practice, a model will usually have unknown parameters that need to be estimated.  We now show how to handle such models in our framework.  

Suppose model $k$ is a parametric family with densities $\{f_k(x;\, \theta_k): \, \theta_k \in \Theta_k\}$, where $\Theta_k \subseteq \mathbb{R}^{d_k}$ and $d_k$ is the dimension of the parameter vector $\theta_k$ for model $k$. 
Furthermore, we generalize from the negative log-likelihood to an arbitrary loss function $\ell_k$  for model $k$.  As before, suppose $X,X_1,\ldots,X_n\in\mathcal{X}$ i.i.d.\ $\sim P_0$, where $P_0$ has density $f_0$ with respect to $\lambda$.  
For each $k \in \{1, \ldots, K\}$, let $\ell_k: \mathcal{X} \times \Theta_k \to \mathbb{R}$ be measurable, and define 
\begin{align}
\label{eq:thetakstar}
    \theta_k^* \in \argmin_{\theta_k \in \Theta_k} \E\big(\ell_k(X; \theta_k)\big).
\end{align}
Define the random vector $\ell(X; \theta^*) := [\ell_1(X; \theta_1^*), \ldots, \ell_K(X; \theta_K^*) ]^\texttt{T} \in \mathbb{R}^K$, and set 
\begin{equation}
\label{eq:mu_cov}
\begin{split}
    \mu^0 &:= \E\big( \ell(X; \theta^*) \big) \in \mathbb{R}^K, \\ 
    \Sigma^0 &:= \C\big( \ell(X; \theta^*) \big) \in \mathbb{R}^{K \times K}.
\end{split}
\end{equation}
In \cref{eq:thetakstar,eq:mu_cov}, the expectations are with respect to $X\sim P_0$.

\begin{condition}
\label{cond1} 
Assume:
    \begin{itemize}
        \item[\textup{(i)}] $\E\big(\ell_k(X; \theta_k^*)^2\big) < \infty$, that is, $\ell(X; \theta^*)$ has finite second moments, and
        \item[\textup{(ii)}] $\Sigma^0$ is a positive definite matrix.  
    \end{itemize}
\end{condition}


To estimate the parameter vector $\theta_k$ for each model $k$, we use the M-estimator, 
\[
    \hat{\theta}_k := \argmin_{\theta_k \in \Theta_k} \Big( \frac{1}{n} \sum_{i=1}^n \ell_k(X_i; \theta_k) \Big).
\]
In the examples in this article, we focus on the negative log-likelihood loss, in which case the M-estimator is simply the maximum likelihood estimator (MLE), but the method extends to any loss leading to a consistent M-estimator.
Plugging in the estimators $\hat{\theta}_1, \ldots, \hat{\theta}_K$ for models $1,\ldots,K$, respectively, the vector of loss values for $X_i$ becomes
\[
    Z_i(\hat\theta) := \ell(X_i; \hat{\theta}) = [\ell_1(X_i; \hat{\theta}_1), \ldots, \ell_K(X_i; \hat{\theta}_K)]^{\texttt{T}} \in \mathbb{R}^K.
\]
Then, the sample mean and covariance based on $Z_1(\hat\theta),\ldots,Z_n(\hat\theta)\in\mathbb{R}^K$ are
\begin{equation}
\label{eq:sample_mean_covariance}
\begin{split}
    \oline{Z}_n(\hat\theta) &:= \frac{1}{n} \sum_{i=1}^n \ell(X_i; \hat{\theta}),\\
    S_n(\hat\theta) &:= \frac{1}{n} \sum_{i=1}^n \big(\ell(X_i; \hat{\theta}) - \oline{Z}_n(\hat\theta)\big)\big(\ell(X_i; \hat{\theta}) - \oline{Z}_n(\hat\theta)\big)^\texttt{T}.
\end{split}
\end{equation}
Under regularity conditions (\cref{cond2}), including that $\hat{\theta}_k$ is a consistent estimator of $\theta_k^*$ for each $k$, the limits in \cref{eq:lln_clt} also hold for $\oline{Z}_n(\hat\theta)$ (see \cref{supp:theorem:param})
generalizing from the no-parameter case to the setting of estimated parameters.
Thus, frequentist inference for $\mu^0$ could again be used in this more general setting; but as mentioned in \cref{sec:method:lad}, we instead recommend a Bayesian approach (\cref{sec:computation}).

\subsection{Adjustment for overfitting bias}\label{sec:method:param:overfitting}

For large sample sizes $n$, plugging in a consistent estimator $\hat{\theta}_k$ as a proxy for $\theta^*$ leads to negligible error.  However, for smaller sample sizes, $\oline{Z}_n(\hat\theta) = \big(\frac{1}{n}\sum_{i=1}^n \ell_k(X_i; \hat\theta_k)\big)_{k=1}^K$ tends to be downward biased as an estimator of $\mu^0$, because the parameters are fit on the same data that is used to estimate $\mu^0$. To adjust for this overfitting, we propose a simple bias correction based on the asymptotic distribution of the likelihood-ratio test statistic. 

To derive the correction, assume the loss is the negative log-likelihood $\ell_k(x; \theta_k) = -\log f_k(x; \theta_k)$. Then the likelihood-ratio test statistic is
\[
    \Lambda_k = 2 \Big( \sum_{i=1}^n \ell_k(X_i; \theta_k^*) - \sum_{i=1}^n \ell_k(X_i; \hat\theta_k) \Big).
\] 
Rearranging and taking the expectation under $P_0$ yields that
\[
    \mu^0_k = \E \Big( \frac{1}{n} \sum_{i=1}^n \ell_k(X_i; \theta_k^*) \Big) = \E\Big(\frac{1}{n}\sum_{i=1}^n \ell_k(X_i; \hat\theta_k)\Big)  + \frac{1}{2n} \E\big(\Lambda_k \big).
\]
Hence, the bias of $\frac{1}{n}\sum_{i=1}^n \ell_k(X_i; \hat\theta_k)$  as an estimator of $\mu^0_k$ is $-(2 n)^{-1} \E(\Lambda_k)$. 

By Wilks' theorem, if model $k$ were correctly specified, then $\Lambda_k$ would be approximately $\chi^2$-distributed with $d_k = \mathrm{dim}(\Theta_k)$ degrees of freedom,
under regularity conditions.
Since the mean of $\chi^2(d_k)$ is $d_k$, this suggests using $\hat{\mu}_k^{\mathrm{bc}} := \frac{1}{n}\sum_{i=1}^n \ell_k(X_i; \hat\theta_k) + d_k/(2 n)$ as a bias-corrected estimator of $\mu^0_k$.  Since LaD operates on vectors $Z_1,\ldots,Z_n\in\mathbb{R}^K$, we apply this adjustment entrywise to define bias-corrected LaD vectors 
\begin{align}
\label{eq:bias-adjustment}
    Z^{\mathrm{bc}}_i(\hat\theta) := Z_i(\hat\theta) + \frac{d}{2n}
\end{align}
where $d = (d_1,\ldots,d_K)^\texttt{T} \in \mathbb{R}^K$.
In practice, we expect all of the models to be misspecified, so Wilks' theorem does not directly apply; nonetheless, empirically we find that the approximation is sufficiently good that the bias correction in \cref{eq:bias-adjustment} works well.

\section{Computation}
\label{sec:computation}

In this section, we describe our Bayesian model for the LaD framework (\cref{sec:method:bayesinf}) and provide a step-by-step algorithm for our proposed methodology (\cref{sec:method:algorithm}).

\subsection{Bayesian inference in the LaD framework}
\label{sec:method:bayesinf}

Let $Z_1,\ldots,Z_n\in\mathbb{R}^K$ denote the LaD vectors, $Z_i = (Z_{i 1},\ldots,Z_{i K})^\texttt{T}$; these may comprise negative log-likelihood values $Z_{i k} = -\log(f_k(X_i))$ (as in \cref{sec:method:lad}), loss values $Z_{i k} = \ell_k(X_i; \hat{\theta}_k)$ evaluated at estimated parameters $\hat{\theta}_k$ (as in \cref{sec:method:param}), or bias-corrected versions of such values $Z_{i k} = \ell_k(X_i; \hat{\theta}_k) + d_k/(2 n)$ (as in \cref{sec:method:param:overfitting}).
To perform Bayesian inference, we model $Z_1,\ldots,Z_n$ as i.i.d.\ $\sim\mathcal{N}(\mu,\Sigma)$ and place a Normal-Inverse-Wishart (NIW) prior on $(\mu,\Sigma)$. Specifically, the prior density of $(\mu, \Sigma)$ is
\[
\mathrm{NIW}(\mu, \Sigma \mid \mu_0, \lambda_0, \Psi_0, \nu_0) =
\mathcal{N}(\mu \mid \mu_0, \Sigma/\lambda_0) \, \mathrm{InverseWishart}(\Sigma \mid \Psi_0, \nu_0)
\]
where $\mu_0 \in \mathbb{R}^K$ is the prior mean, $\lambda_0 > 0$ is a scaling factor, $\nu_0 > K-1$ is the degrees of freedom, and $\Psi_0$ is the (positive definite) inverse scale matrix.
The posterior distribution of $\mu, \Sigma\mid Z_{1:n}$ is NIW with updated parameters, specifically,
\begin{align}
\label{eq:NIW-posterior}
\mu,\Sigma \mid Z_{1:n} \; \sim \; \mathrm{NIW}(\mu_n,\lambda_n,\Psi_n,\nu_n)
\end{align}
where $\lambda_n = \lambda_0 + n$, $\mu_n = (\lambda_0\mu_0 + n \oline{Z}_n)/\lambda_n$, $\Psi_n = \Psi_0 + nS_n + (\lambda_0 n/\lambda_n) (\oline{Z}_n - \mu_0) (\oline{Z}_n - \mu_0)^\texttt{T}$, and $\nu_n = \nu_0 + n$,
where $\oline{Z}_n = \frac{1}{n} \sum_{i=1}^n Z_i$ and $S_n = \frac{1}{n}\sum_{i=1}^n \big(Z_i - \oline{Z}_n \big) \big(Z_i - \oline{Z}_n \big)^\texttt{T}$. 

To generate a posterior sample of $(\mu', \Sigma') \sim p(\mu,\Sigma\mid Z_{1:n})$, all that needs to be done is to draw $\Sigma' \sim \mathrm{InverseWishart}(\Psi_n, \nu_n)$ and then draw $\mu' \sim \cN(\mu_n, \Sigma' / \lambda_n)$. 
To perform posterior inference, we generate multiple posterior samples and use Monte Carlo approximations.
For the hyperparameter settings, we recommend $\mu_0 = 0$, $\nu_0 = K + 2$, $\lambda_0 = 0.01$, and $\Psi_0 = I_K$ (the $K\times K$ identity matrix), since these yield a weakly informative prior that allows the data to mostly dominate.

\subsection{Algorithm for LaD model selection}
\label{sec:method:algorithm}

\cref{alg:workflow} provides a step-by-step description of the workflow for our proposed method.

\begin{algorithm}
\small
\caption{~~ Workflow for LaD model selection}
\label{alg:workflow}

\textbf{Input: } Observations $X_1,\ldots,X_n$, parametrized models $k=1,\ldots,K$ and their corresponding complexities $c(1),\ldots,c(K) \geq 0$, number of Monte Carlo samples $T > 0$, and tolerance $\delta \geq 0$.

\begin{enumerate}
    \item Compute parameter estimates $\hat\theta_k$ for each model $k$.
    \item Compute the bias-corrected LaD values $Z_{i k} = \ell_k(X_i; \hat\theta_k) + d_k/(2 n)$ for all $i$ and $k$.
    \item Sample $\mu^{(t)}, \Sigma^{(t)}$ from the NIW posterior (\cref{eq:NIW-posterior}) for $t = 1,\ldots,T$.
    \item Compute the SLC scores $\hat{w}_{\delta}(k)$ for each $k$ using \cref{alg:proc}.
    \item (Optional) Compute the rescaled tolerance $\hat{\tau} = \delta / (\hat{\mu}_{\mathrm{noise}} - \hat{\mu}_{\mathrm{min}})$ where $\hat{\mu}_{\mathrm{noise}}$ is an estimate of $\mu^0_{\mathrm{noise}}$ (\cref{sec:method:delta}), $\hat{\mu}_{\mathrm{min}} = \min_k \hat{\mu}_k$, and $\hat{\mu}_k = \frac{1}{n}\sum_{i=1}^n Z_{i k}$.
\end{enumerate}
\textbf{Output: } Return the samples $\mu^{(t)}$, SLC scores $\hat{w}_{\delta}(k)$, and rescaled tolerance $\hat{\tau}$.
\end{algorithm}


Recall that our SLC score $w_{\delta}(k \mid Z_{1:n})$ (\cref{eq:pkn2}) is defined as the product of a between-class selection factor $P\big(c_\delta^*(\mu) = c(k) \mid Z_{1:n} \big)$ and a within-class selection factor $\E\big(r_{n k}(\mu) \mid Z_{1:n} \big)$.
To compute $w_{\delta}(k \mid Z_{1:n})$, we generate posterior draws $\mu^{(1)},\ldots,\mu^{(T)}$ given $Z_{1:n}$ as described in \cref{sec:method:bayesinf}, then compute each of these two factors using a simple Monte Carlo approximation, and take their product.  See \cref{alg:proc} for details.

\begin{algorithm}
\small
\caption{~~ Computing the smooth LaD criterion (SLC) scores}
\label{alg:proc}

\textbf{Input: } Posterior samples $\mu^{(1)},\ldots,\mu^{(T)}\in\mathbb{R}^K$, complexities $c(1), \dots, c(K)\geq 0$, tolerance $\delta \ge 0$, and temperature $\alpha_n > 0$.

\vspace{1em}
\textbf{for $t=1, \dots, T$ do}
\begin{enumerate}
    \item  \textbf{Between-class selection:} 
        Compute the set of $\delta$-optimal models,
        \[
            M_\delta(\mu^{(t)}) = \big\{k: \mu_k^{(t)} \le \min_j \mu_j^{(t)} + \delta \big\}.
        \]
        Then, compute the minimal complexity among $\delta$-optimal models,
        \[
            c_\delta^*(\mu^{(t)}) = \min \big\{ c(k): k \in M_\delta(\mu^{(t)}) \big\}.
        \]
    \item  \textbf{Within-class selection:} 
        For each model $k = 1,\ldots,K$, compute the within-class minimum $\mu_{\min,c(k)}^{(t)} = \min\big\{\mu_j^{(t)}: c(j) = c(k) \big\}$, and compute the soft-selection score
        \[
            r_{n k}(\mu^{(t)}) = \exp\big(-\alpha_n (\mu_k^{(t)} - \mu_{\min,c(k)}^{(t)}) \big).
        \]
\end{enumerate}
\textbf{for $k=1, \dots, K$ do}
        \[
        \hat{p}_{\delta}(k) := \frac{1}{T} \sum_{t=1}^T \mathds{1} \big(c_\delta^*(\mu^{(t)}) = c(k) \big) \quad \text{and} \quad \hat{r}(k) := \frac{1}{T} \sum_{t=1}^T r_{n k}(\mu^{(t)}).
        \]

\textbf{Output: } Return the estimated SLC scores $\hat{w}_{\delta}(k) := \hat{p}_{\delta}(k) \, \hat{r}(k)$ for $k = 1,\ldots,K$.
\end{algorithm}



\section{Theory}\label{sec:theory}

In this section, we prove that our proposed method is asymptotically consistent (\cref{sec:theory:consistency}) and we show that alternative approaches based on a plug-in posterior or using a hard minimum exhibit instability (\cref{sec:theory:instability}).

\subsection{Consistency of the SLC score}
\label{sec:theory:consistency}

Our first result establishes that the SLC score $w_{\delta}(k \mid Z_{1:n})$ concentrates on the true target set $M_\delta^*(\mu^0)$, asymptotically.

\begin{theorem} \label{sec:theory:thm:consistency_score}
Assume Conditions \ref{cond1} and \ref{cond2}, and suppose $\delta \neq \mu^0_k - \mu^0_{\mathrm{min}}$ for all $k = 1,\ldots,K$.
Let $\alpha_n \geq 0$ such that $\alpha_n \to \infty$ with $\alpha_n = o(\sqrt{n})$.
Define $w_{\delta}(k \mid Z_{1:n})$ as in \cref{eq:pkn2} based on the NIW posterior in \cref{sec:method:bayesinf}.
Then, for all $k \in \{1,\dots,K\}$,
$$ w_{\delta}(k \mid Z_{1:n}) \xrightarrow[n \to \infty]{\mathrm{p}} \mathds{1}\big(k \in M^*_\delta(\mu^0)\big).$$
\end{theorem}

All proofs are collected in \cref{supp:proofs}.
The next result shows that the Monte Carlo estimates $\hat{w}_{\delta}(k)$ produced by \cref{alg:proc} also concentrate on $M_\delta^*(\mu^0)$.

\begin{theorem} \label{sec:theory:thm:consistency}
Under the same assumptions as \cref{sec:theory:thm:consistency_score},
for all $k \in \{1,\dots,K\}$,
$$ \hat{w}_{\delta}(k) \xrightarrow[n \to \infty]{\mathrm{p}} \mathds{1}\big(k \in M^*_\delta(\mu^0)\big).$$
\end{theorem}

\begin{theorem} \label{sec:theory:thm:consistency_M}
Under the same assumptions as \cref{sec:theory:thm:consistency_score}, for any fixed $\omega \in (0,1)$,
\[
    P \big( \hat{M}^*_\delta = M^*_\delta(\mu^0) \big) \xrightarrow[n\to\infty]{} 1,
\]
where $\hat{M}^*_\delta := \{k : w_{\delta}(k \mid Z_{1:n}) > \omega\}$.
\end{theorem}
\cref{sec:theory:thm:consistency_M} also holds when using $\hat{w}_{\delta}(k)$ instead of $w_{\delta}(k \mid Z_{1:n})$.

\subsection{Instability of alternative approaches at ties}
\label{sec:theory:instability}

The results in this section provide the rationale for using our soft-selection score for within-class selection (\cref{eq:softmin}), rather than the Bayesian plug-in posterior $P(k\in M^*_\delta(\mu) \mid Z_{1:n})$ or a hard-minimum selection score. 
When multiple models of equal complexity are tied for having minimal KL divergence, these naive alternatives randomly allocate probability mass among tied models, and not necessarily in a uniform way.

To understand why the issue arises, suppose there are $K = 3$ models of equal complexity and all three are equally close to the true DGP, that is, $\mu^0 = (\mu^0_1,\mu^0_2,\mu^0_3)^\texttt{T}$ where $\mu^0_1 = \mu^0_2 = \mu^0_3$.
Furthermore, suppose models $1$ and $2$ are strongly negatively correlated and have much higher variance than model $3$, for example,
$$ \Sigma^0 = \begin{bmatrix} 1 & -0.99 & 0 \\ -0.99 & 1 & 0 \\ 0 & 0 & 0.01 \end{bmatrix}. $$
Under the posterior, suppose for simplicity that $\mu \sim \mathcal{N}(\mu^0, \Sigma^0/n)$ given $Z_{1:n}$, which is the expected asymptotic behavior.
Then with fairly high probability, $\min\{\mu_1,\mu_2\} < \mu_3$ because $\mu_3$ has much smaller variance than the other two, and the strong negative correlation makes $\mu_1 \approx -\mu_2$.
Specifically, $P(\mu_k = \min_j \mu_j \mid Z_{1:n})$ is $\approx 0.48$ for $k\in\{1,2\}$ and $\approx 0.04$ for $k = 3$.
Thus, model $3$ is rarely selected in this scenario, even though all three models are equally good. \cref{fig:instability_example} visualizes this pattern. 

\begin{figure}
    \centering
    \includegraphics[width=0.99\linewidth]{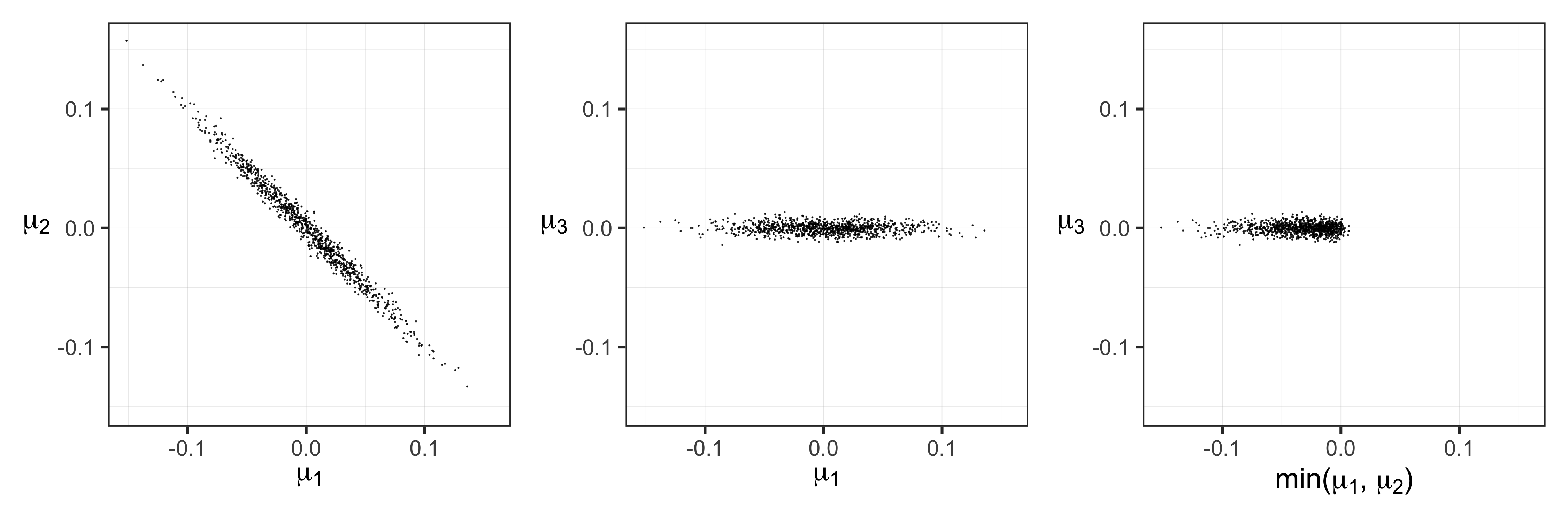}
    \captionsetup{font=small}
    \caption{Instability under ties with $K=3$. The figure shows 1,000 posterior draws from $\mu \sim \mathcal{N}(\mu^0, \Sigma^0/n)$ with $n=500$, $\mu^0 = (0,0,0)^\mathtt{T}$, and $\Sigma^0$ as specified in the text. The left panel shows the strong negative correlation between $\mu_1$ and $\mu_2$, the middle panel shows that $\mu_3$ is much more concentrated than $\mu_1$ and $\mu_2$, and the right panel makes clear that $\mu_3$ is rarely the minimum. 
    }
    \label{fig:instability_example}
\end{figure}

We make this phenomenon precise in the following results.
For $x\in\mathbb{R}^K$, define 
\begin{align}
\label{eq:pi_k}
\pi_k(x) = \branch{P\Big(W_j - W_k \leq x_j  - x_k \text{ for all } j\in M_\delta^*(\mu^0)\setminus\{k\}\Big)}{k\in M_\delta^*(\mu^0)}{0}{k\not\in M_\delta^*(\mu^0)}
\end{align}
where $W\sim\mathcal{N}(0,\Sigma^0)$.

In \cref{sec:theory:thm:post_instability}, parts (a) and (b) show the instability of the plug-in posterior and a hard-minimum rule, respectively. When $|M^*_\delta(\mu^0)| > 1$, the choice among tied models is randomly distributed in a way that depends on the covariance of the limiting Gaussian. 


\begin{theorem} 
\label{sec:theory:thm:post_instability}
Assume Conditions \ref{cond1} and \ref{cond2}, let $\mu, \Sigma \mid Z_{1:n}$ be distributed according to \cref{eq:NIW-posterior}, and suppose $\delta \neq \mu^0_k - \mu^0_{\mathrm{min}}$ for all $k = 1,\ldots,K$. Let $W\sim\mathcal{N}(0,\Sigma^0)$.
\begin{enumerate}[label=(\alph*)]
\item 
$\displaystyle\big( P(k\in M^*_\delta(\mu) \mid Z_{1:n}) \big)_{k=1}^K \xrightarrow[n \to \infty]{\mathrm{d}} \big( \pi_k(W) \big)_{k=1}^K$.
\item 
Consider a modified score $\tilde{w}_{\delta}(k \mid Z_{1:n})$ defined as in \cref{eq:pkn2} but with $\tilde{r}_{n k}(\mu) = \mathds{1}(\mu_k = \mu_{\min,c(k)})$ in place of $r_{n k}(\mu)$. Then
$\displaystyle\big(\tilde{w}_{\delta}(k \mid Z_{1:n})\big)_{k=1}^K \xrightarrow[n\to\infty]{\mathrm{d}} \big(\pi_k(W)\big)_{k=1}^K$.
\end{enumerate}
\end{theorem}



Interestingly, the limiting behavior of $\tilde{w}_{\delta}(k \mid Z_{1:n})$ coincides with the asymptotic BayesBag posterior with $M=N$ \citep[Theorem 3.2]{huggins2023reproducible}. 
Corollary \ref{sec:theory:cor:instability_two2} shows that in the case of two tied models, the scores of the two models are asymptotically distributed as $(U, 1-U)$ where $U \sim \mathrm{Uniform}(0,1)$.

\begin{corollary} \label{sec:theory:cor:instability_two2}
    Assume the conditions of \cref{sec:theory:thm:post_instability}. If $M_\delta^*(\mu^0) = \{k_1,k_2\}$ then
    $$ \big(\tilde{w}_\delta(k_1 \mid Z_{1:n}), \tilde{w}_\delta(k_2 \mid Z_{1:n})\big) \xrightarrow[n\to\infty]{\mathrm{d}} (U, 1-U) $$
    where $U \sim \mathrm{Uniform}(0,1)$.
\end{corollary}

\section{Examples}\label{sec:examples}

\subsection{Gaussian mixture models}\label{sec:examples:gmm}

We revisit the motivating example from \cref{sec:intro} using finite Gaussian mixture models on the Shapley galaxy data. 
We implement the LaD workflow in \cref{alg:workflow} as follows.
First, we fit a $k$-component Gaussian mixture for each $k = 1, \ldots, 10$ using the expectation-maximization (EM) algorithm to estimate the \emph{maximum a posteriori} (MAP) solution \citep{fraley2007bayesian}. For this,  we employ the \texttt{mclust} package in R \citep{scrucca2016mclust, r2016r} with random initialization, using a moderately informative normal-inverse-gamma prior on the mean and variance of each mixture component, independently.  For each $k$, we run the algorithm on $50$ independent restarts and keep the run with the highest observed log-likelihood.
See \cref{supp:mix_details} for model details.

We then compute the bias-corrected LaD values $Z_{i k} = \ell_k(X_i; \hat\theta_k) + d_k/(2 n)$ (\cref{eq:bias-adjustment}) where $\hat\theta_k$ is the MAP estimate and $d_k = 3 k - 1$ since a univariate $k$-component mixture with unconstrained means, variances, and mixture weights has a total of $k + k + (k-1) = 3 k - 1$ parameters.
While mixtures do not perfectly satisfy the LRT conditions, we find that this bias correction is an effective approximation in practice.
We draw $T = 1000$ samples of $\mu \mid Z_{1:n}$ from the NIW posterior (\cref{eq:NIW-posterior}), and then run \cref{alg:proc} to obtain the SLC scores $\hat{w}_{\delta}(k)$ for a range of values of the tolerance $\delta>0$.
We use $\alpha_n = n^{0.45}$ for the temperature and $c(k) = k$ for the model complexities.
To compute the rescaled tolerance $\hat{\tau}$ as described in \cref{alg:workflow}, we let $\muhat_{k} = \frac{1}{n} \sum_{i=1}^n Z_{i k}$ be the mean of the bias-adjusted LaD values, and we define the noise baseline to be the uniform distribution over the observed range $[\min_i x_i, \max_i x_i]$, which yields $\muhat_{\mathrm{noise}} = \log( \max_i x_i - \min_i x_i)$.  This puts $\delta$ on an interpretable scale, namely, $\hat{\tau} \in [0,1]$ is the fraction of explainable KL improvement one is willing to give up for model simplicity.

\begin{figure}
    \centering
    \includegraphics[trim=0 0 0 0, clip, width=0.95\linewidth]{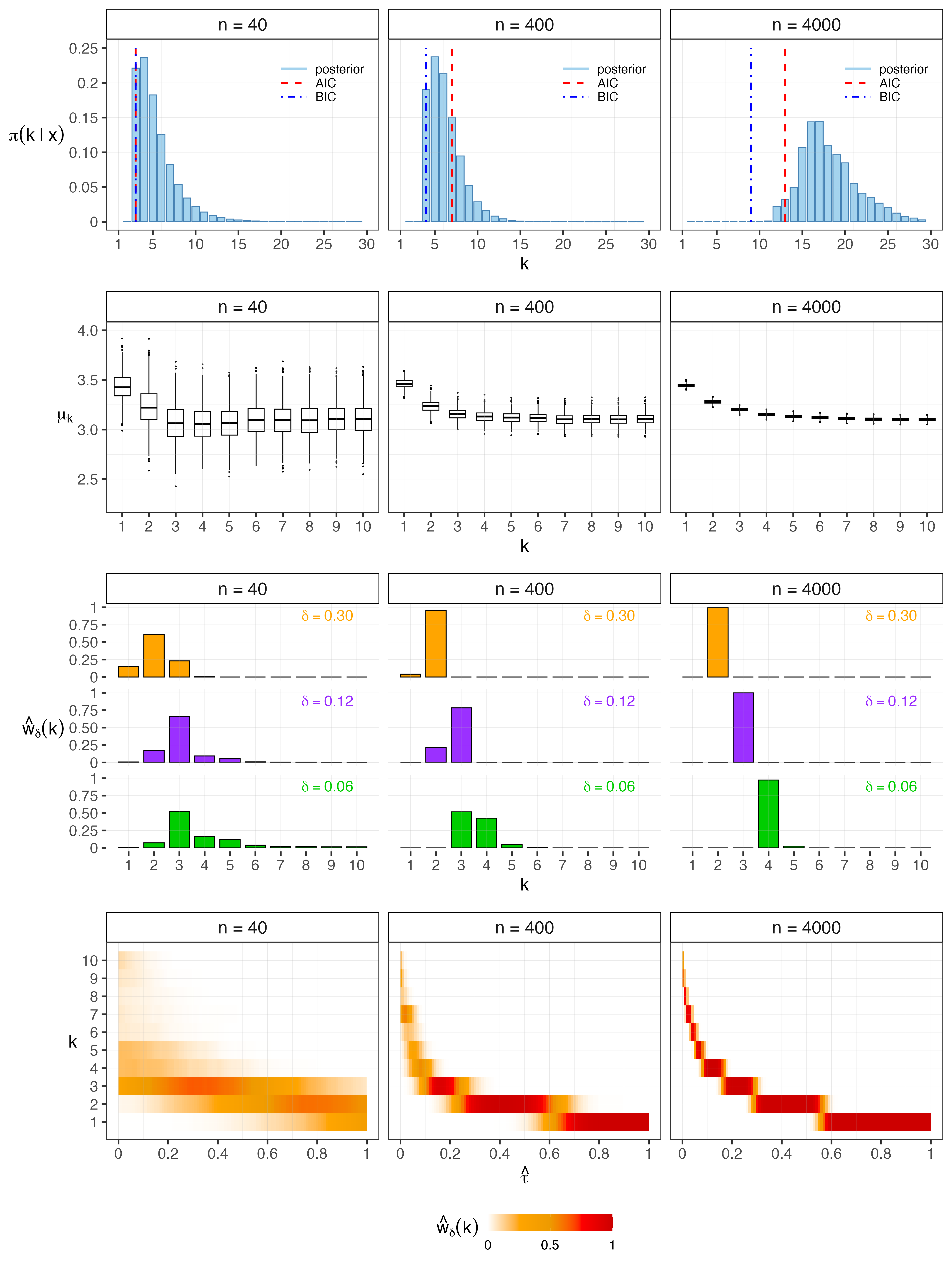}
    \captionsetup{font=small}
    \caption{Gaussian mixture model results on the Shapley galaxy dataset. 
    \emph{(Top row)} Usual Bayesian posterior on $k$ (bars) with AIC (red dashed) and BIC (blue dot-dash) overlaid; as $n$ increases, all three trend toward larger $k$.
    \emph{(Second row)} LaD posteriors on $\mu_k$ for each $k$ show the fit to the true DGP as a function $k$, and concentrate as $n$ grows.
    \emph{(Third row)} SLC score $\hat{w}_{\delta}(k)$ for tolerances $\delta \in \{0.3, 0.12, 0.06\}$. Larger $\delta$ selects simpler models; increasing $n$ concentrates the scores onto a single $k$ for each $\delta$.
    \emph{(Bottom row)} Heatmap of ``posterior path'' of SLC score $\hat{w}_{\delta}(k)$ shows the transition from selecting more complex to simpler models as a function of the rescaled tolerance $\hat{\tau} = \delta/(\hat{\mu}_{\mathrm{noise}} - \hat{\mu}_{\mathrm{min}}) \in [0,1]$.
    }
    \label{fig:gmm_results}
\end{figure}

\cref{fig:gmm_results} displays the results. For comparison, the top row shows the usual Bayesian posterior on $k$, with AIC and BIC overlaid; all three increasingly favor larger $k$ as $n$ grows, exhibiting unstable drift toward complexity under misspecification. 
The second row shows boxplots of the LaD posterior on $\mu_k$ for each $k$, which provide estimates (with uncertainty quantification) of the KL divergence between the true DGP and each model $k$, up to a common additive constant. 
As $n$ increases, we see that the uncertainty in these estimates decreases. As $k$ increases, the estimated KL decreases and gradually levels off, indicating a diminishing improvement in fit as the model complexity grows.
The third row plots the SLC scores $\hat{w}_{\delta}(k)$ for $\delta \in \{0.30, 0.12, 0.06\}$. Larger $\delta$ values are more tolerant, resulting in selection of simpler models; meanwhile, smaller $\delta$ values enforce a tighter tolerance, requiring more complex models in order to achieve a KL within $\delta$ of being minimal.
As $n$ increases, the $\hat{w}_{\delta}(k)$ values concentrate on a single $k$ for each $\delta$, illustrating the large-sample concentration guaranteed by our theory when no two models are tied in terms of both complexity and KL.
Notably, in contrast with AIC, BIC, and the usual Bayesian posterior results in \cref{fig:intro_example}, the LaD results stabilize as $n$ increases.  

This is illustrated further in the bottom row, which shows the \textit{posterior path} of $\hat{w}_{\delta}(k)$ as $\hat{\tau}$ varies from $0$ to $1$.
Here, we see a continuous transition from selecting more complex to simpler models as the proportion of explainable KL that we are willing to sacrifice increases from $0$ to $1$.
For instance, while the usual Bayesian posterior selects values of $k$ between 10 and 25 when $n = 4000$ (\cref{fig:intro_example}), 
the LaD posterior path shows that $k > 8$ is only needed for very tight tolerances of $\hat{\tau} < 0.01$, that is, $<1\%$ of explainable KL sacrificed.
This plot encapsulates the key features of LaD: It operates without assuming any candidate model is correctly specified, quantifies uncertainty in their relative KL divergences from the true DGP, and places these divergences on an interpretable scale.

\subsection{Sparse multivariate normal mean models}\label{sec:examples:mvn}
In this simulation, we consider multivariate normal (MVN) models with sparse mean vectors. We generate $n$ i.i.d.\ samples $x_1, \ldots, x_n \sim \cN(\theta_0, I)$, where the true mean vector is $\theta_0 = (1, 1, 0.5, 0.5, 0.4, 0)^\texttt{T}$, and $I$ is the $6 \times 6$ identity matrix. We compare $K=7$ candidate models, each of which is $\mathcal{N}(\theta,I)$ but with a different sparsity pattern imposed on the mean vector $\theta$, and we define the complexity $c(k)$ to be the number of nonzero entries of $\theta$.
The models vary in complexity and KL divergence from the truth, $\min_{\theta_k} D\big(\mathcal{N}(\theta_0,I)\,\|\,\mathcal{N}(\theta_k,I)\big)$ subject to the sparsity pattern in \cref{tab:sparse_mvn_models}; see \cref{supp:mvn_details}.

\begin{table}[H]
\centering
\caption{Candidate models for sparse multivariate normal simulation example.}
\begin{tabular}{|c|l|c|c|}
\hline
$k$ & \multicolumn{1}{c|}{$\theta_k$} & $c(k)$ & $\min D$ \\
\hline
$1$ & $(\theta_1, 0, 0, \theta_4, 0, 0)^\texttt{T}$ & $2$ & $0.705$ \\
$2$ & $(\theta_1, \theta_2, 0, 0, 0, 0)^\texttt{T}$ & $2$ & $0.33$ \\
$3$ & $(\theta_1, \theta_2, 0, 0, \theta_5, 0)^\texttt{T}$ & $3$ & $0.25$ \\
$4$ & $(\theta_1, \theta_2, 0, \theta_4, 0, 0)^\texttt{T}$ & $3$ & $0.205$ \\
$5$ & $(\theta_1, \theta_2, \theta_3, 0, 0, 0)^\texttt{T}$ & $3$ & $0.205$ \\
$6$ & $(\theta_1, \theta_2, \theta_3, \theta_4, \theta_5, 0)^\texttt{T}$ & $5$ & $0$ \\
$7$ & $(\theta_1, \theta_2, \theta_3, \theta_4, \theta_5, \theta_6)^\texttt{T}$ & $6$ & $0$ \\
\hline
\end{tabular}
\label{tab:sparse_mvn_models}
\end{table}

Models $1$ and $2$ are the most parsimonious, but differ greatly in their distance from the truth.
Models $4$ and $5$ are equally close to the truth and have the same complexity, while model $3$ is slightly farther in KL divergence but has the same complexity as $4$ and~$5$.
Models $6$ and $7$ both recover the true DGP exactly, but model $6$ is simpler.

We implement the LaD approach in \cref{alg:workflow} as follows.
We use the MLE to estimate the parameters for model $k$, which is the vector $\hat{\theta}_k\in\mathbb{R}^6$ containing the sample means in the coordinates that are free in that model, and zeros elsewhere.  That is, for coordinate $j\in\{1,\ldots,6\}$, $\hat{\theta}_{k j} = \bar{x}_j \mathds{1}(j\in J_k)$ where $\bar{x}_j = \frac{1}{n}\sum_{i=1}^n x_{i j}$ and $J_k\subseteq \{1,\ldots,6\}$ is the set where model $k$ allows $\theta_{k j}$ to be nonzero.
The bias-corrected LaD values are then $Z_{i k} = \ell_k(x_i; \hat{\theta}_k) + d_k/(2 n)$ where
$\ell_k(x_i; \hat{\theta}_k) = \frac{6}{2} \log (2\pi) + \frac{1}{2} \|x_i - \hat\theta_k \|^2$ and $d_k = |J_k|$.
We draw $T = 1000$ samples of $\mu\mid Z_{1:n}$ (\cref{eq:NIW-posterior}) and compute the SLC scores $\hat{w}_{\delta}(k)$ using \cref{alg:proc} with a temperature of $\alpha_n = n^{0.45}$. 
We consider three tolerance levels: (i) $\delta = 0.75$ is large enough to tolerate all of the models, (ii) $\delta = 0.26$ excludes models 1 and 2 but still permits models 3--7, and (iii) $\delta = 0.05$ is small enough that only models 6 and 7 are within $\delta$ of the true DGP; see \cref{tab:sparse_mvn_models}.
To interpret the value of $\delta$, we compute the rescaled tolerances $\tau$ as in \cref{alg:workflow}, using $\mathcal{N}(0,I)$ for the noise model. 
For instance, the $\delta$ values of $0.75$, $0.26$, and $0.05$ correspond to $\tau$ values of approximately $55.3\%$, $18.8\%$, and $3.7\%$, respectively.

\begin{figure}
    \centering
    \includegraphics[width=0.8\linewidth]{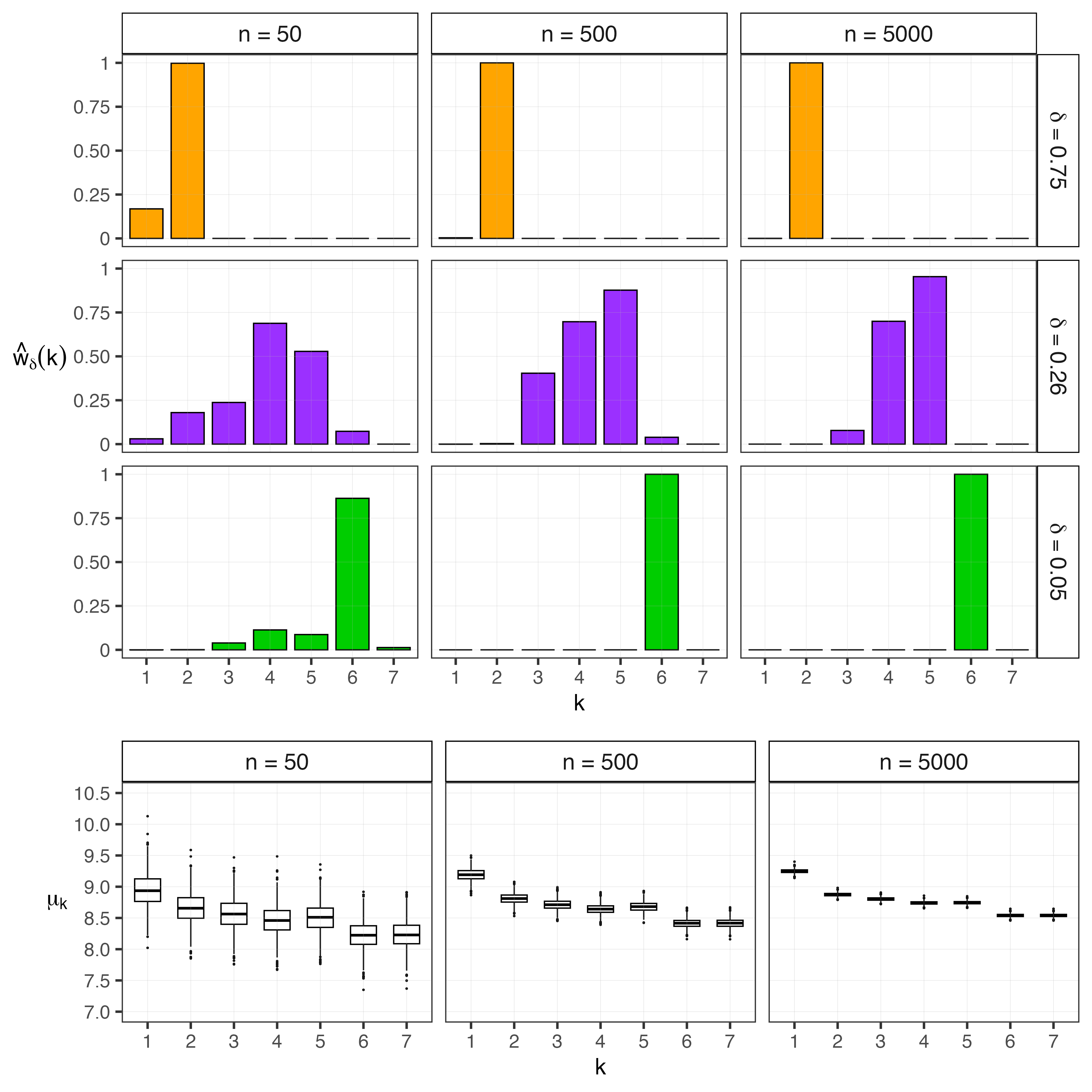}
    \captionsetup{font=small}
    \caption{Sparse multivariate normal model results. 
    \emph{(Top)} SLC scores $\hat{w}_{\delta}(k)$  for sample sizes $n \in \{50, 500, 5000\}$ (columns) and tolerances $\delta \in \{0.75, 0.26, 0.05\}$ (rows).
    \emph{(Bottom)} Boxplots of the LaD posteriors on $\mu_k$ for $k \in \{1,\ldots,7\}$.
     }
    \label{fig:mvn_result}
\end{figure}

Figure \ref{fig:mvn_result} summarizes the LaD results across the $\delta$ settings, for sample sizes $n \in \{50, 500, 5000\}$.
Here, we show results for one representative data set of each size.
In the first setting ($\delta = 0.75$), both models 1 and 2 are given a nonnegligible SLC score when $n$ is small, but the score of model 1 quickly goes to zero and the score of model 2 goes to 1 as $n$ grows.
This is because all candidate models are $\delta$-optimal at this tolerance level, but models 3--7 have greater complexity, and model 1 has worse fit than model 2.
The SLC score concentrates at model 2 since it is the best minimal-complexity $\delta$-optimal model at this tolerance; see \cref{tab:sparse_mvn_models}.

In the second setting ($\delta = 0.26$), the SLC scores are spread more evenly across models when $n$ is small, due to posterior uncertainty. 
As $n$ increases, the SLC scores concentrate on both models 4 and 5, which have equal complexity and equal KL divergence from the true DGP.  This occurs because models 1 and 2, while simpler, are not $\delta$-optimal at this tolerance level; models 6 and 7 have higher complexity; and model 3 has larger KL from the true DGP than models 4 and 5.
Notably, the SLC score is stable in the sense that it does not randomly select one of models 4 and 5, but instead gives them roughly equal score.
Thus, the method correctly treats these two models as equally good with respect to this true DGP.

In the third setting ($\delta = 0.05$), the SLC scores concentrate on model 6.  This is because models 1--5 are not $\delta$-optimal at this tolerance level, and model 7 is more complex than model 6.
Thus, the method correctly chooses model 6 at this value of $\delta$.



\subsubsection{Performance comparison with other methods}\label{sec:examples:mvn:comparison}

Up to this point, we have considered qualitative aspects of performance in our comparisons with existing methods.
In this section, we quantitatively compare model selection performance with several alternative approaches.
For these comparisons, we consider the following methods, including three versions of LaD:
\begin{itemize}
    \item[(i)] (LaD-soft) - Our proposed LaD method (see \cref{alg:workflow}).
    \item[(ii)] (LaD-hard) - A variant of LaD using hard-minimum selection (see \cref{sec:theory:thm:post_instability}(b)).
    \item[(iii)] (LaD-diag) - A variant of LaD using diagonal covariance $\Sigma$ (see \cref{supp:cov}).
    \item[(iv)] (c-posterior) - Coarsened posterior with powers $\alpha \in \{10, 100\}$ (see \cref{supp:mvn_coarsening}). 
    \item[(v)] (Bayes) - Usual Bayesian posterior probabilities on models.
    \item[(vi)] (AIC, BIC) - Standard model selection criteria that select a single model.
\end{itemize}

For each $n\in\{50, 500, 5000\}$, we generate $50$ simulated data sets from $\mathcal{N}(\theta_0, I)$ as above and, for each data set, we apply each method to select from the set of candidate models in \cref{tab:sparse_mvn_models}.
To quantify performance for each $\delta$, consider the Brier loss $S_{\mathrm{B}} := \sum_{k=1}^K (w_k - w_k^0)^2$ where $w_k$ is the weight assigned to model $k$ by a given method and $w_k^0$ is an indicator of whether model $k$ is in the true target set, that is, $w_k^0 = \mathds{1}(k \in M_\delta^*(\mu^0))$.

\begin{figure}
    \centering
    \includegraphics[width=0.99\linewidth]{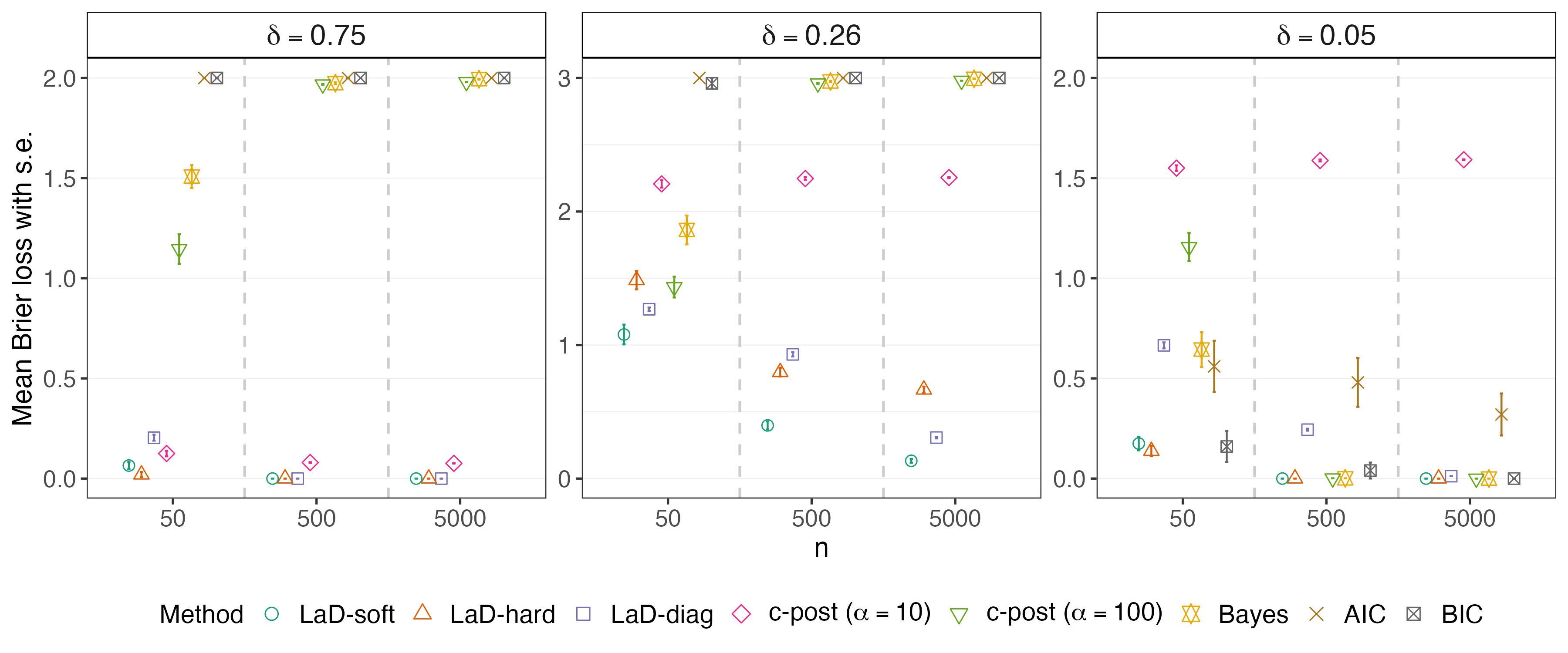}
    \captionsetup{font=small}
    \caption{Performance comparisons using Brier score. Points show the mean Brier loss, and lines show $\pm$ standard error across $50$ datasets for sample sizes $n \in \{50, 500, 5000\}$ (x-axis; vertical dashed lines separate values of $n$). Each panel corresponds to tolerance settings: left $\delta=0.75$, middle $\delta = 0.26$, and right $\delta = 0.05$. Colors indicate the competing methods. 
    }
    \label{fig:gmm_metrics}
\end{figure}

\cref{fig:gmm_metrics} displays the means and standard errors of the Brier loss $S_\mathrm{B}$ for each tolerance setting, $\delta\in\{0.75, 0.26, 0.05\}$.  LaD-soft attains the smallest Brier loss on average across all $n$ and $\delta$. LaD-hard performs comparably to LaD-soft except when there are ties, such as between models 4 and 5 in the case of $\delta = 0.26$, in which case LaD-soft performs better.
LaD-diag works reasonably well but is generally worse than LaD-soft, particularly when $\delta = 0.05$, demonstrating that modeling the full covariance matrix is important. The c-posterior with $\alpha=10$ is competitive at tolerance $\delta=0.75$, but performs poorly at smaller $\delta$.  Meanwhile, the c-posterior with $\alpha = 100$ performs well at $\delta = 0.05$, but poorly at larger $\delta$. This illustrates that the choice of $\alpha$ for the c-posterior needs to be selected in a $\delta$-dependent way, but this is challenging because $\alpha$ does not have a direct interpretation in terms of distance from the true DGP.
Finally, Bayes, AIC, and BIC are only competitive at $\delta = 0.05$ since $M_\delta^*(\mu^0)$ only contains the true DGP in this case. Overall, LaD-soft reliably provides excellent performance across all settings.


\subsection{Thermal performance curves in microbial ecology}\label{sec:examples:tpc}

Thermal performance curves (TPCs)---unimodal functions relating performance (such as growth rate) to temperature---are widely used in microbial ecology for understanding how populations respond to environmental conditions \citep{sinclair2016can}. 
For example, \textit{Prorocentrum minimum} is a mixotrophic dinoflagellate that can form harmful blooms; studying its temperature–growth relationship helps predict these events and mitigate their impact on marine ecosystems and human health \citep{grzebyk1996influences}.
There are many mechanistic and phenomenological TPC models, none of which is universally optimal across traits or taxa \citep{kontopoulos2024no, kellermann2019comparing, angilletta2006estimating}.
We analyze the \textit{P.\ minimum} data of \citet{grzebyk1996influences} consisting of $n=148$ growth-rate observations at 10 temperatures from $10^\circ$C to $31^\circ$C.

We fit $K = 9$ widely used TPC models of varying complexity: three-parameter models (Gaussian, Mitchell--Angilletta, Bri\`ere I, Eubank); four-parameter models (Weibull, Modified Gaussian); a five-parameter model (extended Bri\`ere); a six-parameter model (Poly-5); and a seven-parameter model (Sharpe--Schoolfield). 
These represent a diverse mix of TPC models that are commonly recommended for fitting microbial growth data \citep{kontopoulos2024no, grimaud2017modeling, thomas2012global}; see \cref{supp:tpc} for their functional forms. 


\begin{figure}
    \centering
    \includegraphics[width=0.95\textwidth]{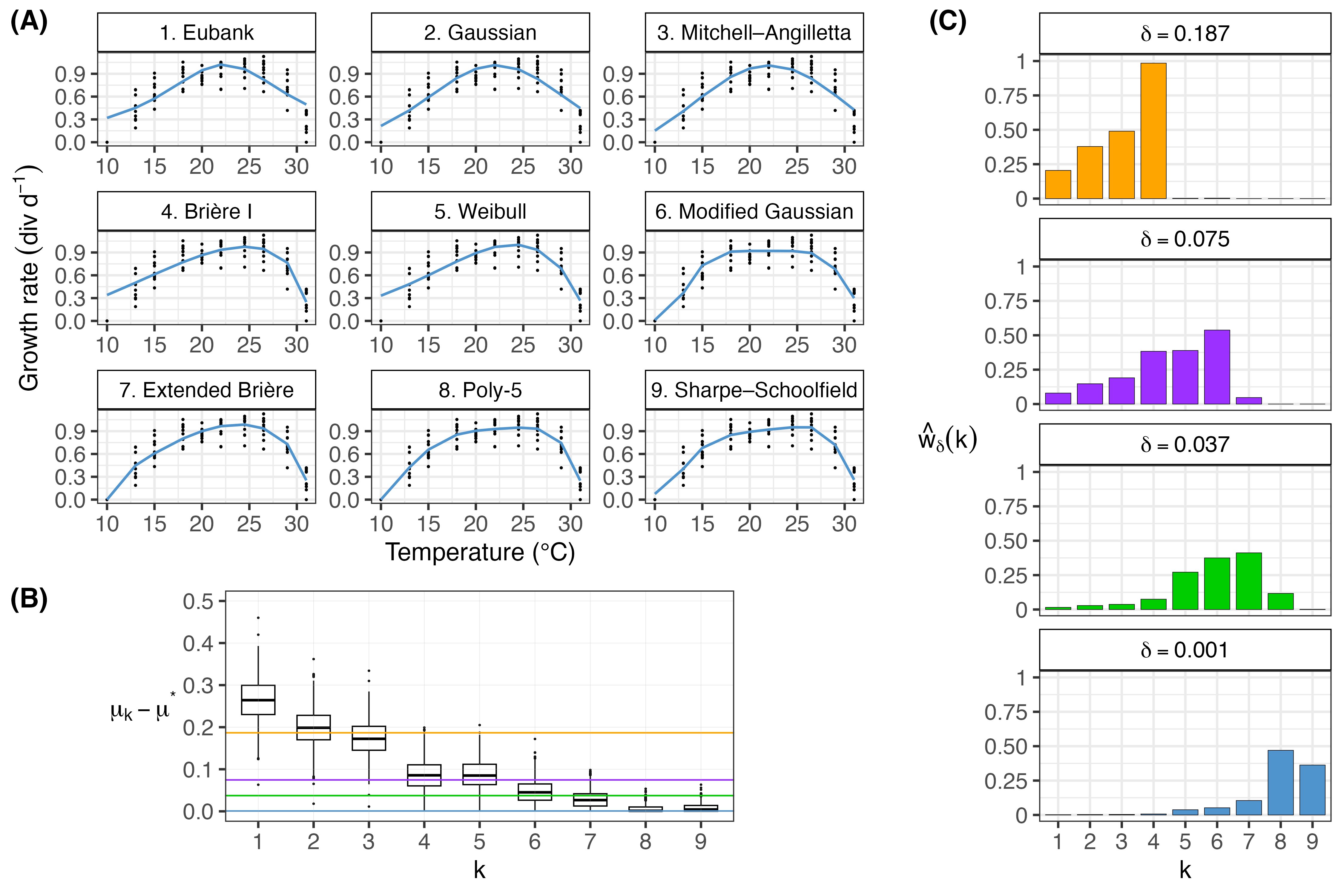}
    \captionsetup{font=small}
    \caption{(A) NLS fits (blue lines) of the $K=9$ candidate TPC models to the $n=148$ growth-rate observations (black points) across temperatures $10^{\circ}\mathrm{C}$ -- $31^{\circ}\mathrm{C}$. Growth rates are reported as divisions per day $(\mathrm{div \, d^{-1}})$, converted from per-second rates $(\mathrm{s^{-1}})$ via $\mathrm{div \, d^{-1}} = 86{,}400/\ln{2} \times \mathrm{s^{-1}}$, as in \citet{grzebyk1996influences}.
    (B) Posterior distributions of $\mu_k - \min_j \mu_j$ for each model $k$. Horizontal lines indicate selected $\delta$ thresholds (orange: $\delta = 0.187$, purple: $\delta = 0.075$, green: $\delta = 0.037$, and blue: $\delta = 0.001$).
    (C) SLC scores $\hat{w}_{\delta}(k)$ across models $k$ for each $\delta$, using the same color for each $\delta$ as in (B).
    }
    \label{fig:tpc_results}
\end{figure}


The data consist of pairs $(y_i,t_i)$ where $y_i$ is the outcome (growth rate) and $t_i$ is the temperature, for observations $i = 1,\ldots,n$.
We implement LaD as follows (\cref{alg:workflow}).
For each candidate model $k$ with parameter $\theta_k$ and mean curve $g_k(t; \theta_k)$, we fit $\hat{\theta}_k$ by nonlinear least squares (NLS), modeling $y_i\sim \mathcal{N}(g_k(t_i; \theta_k), \, \sigma_k^2)$, a standard approach for these models implemented in R \citep{padfield2021rtpc}. \cref{fig:tpc_results}(A) shows the fitted curve for each model, overlaid on a scatterplot of the data. 
The bias-corrected LaD values are $Z_{i k} = \ell_k((y_i,t_i); \hat{\theta}_k) + d_k/(2 n)$ where 
$\ell_k((y_i,t_i); \hat{\theta}_k) = \frac{1}{2} \log (2\pi \hat{\sigma}_k^2) + \big(y_i - g_k(t_i; \hat{\theta}_k) \big)^2 / (2\hat{\sigma}_k^2)$,
$\hat{\sigma}_k^2 = \frac{1}{n} \sum_i \big( y_i - m_k(T_i; \hat\theta_k) \big)^2$,
and $d_k$ is the dimension of $\theta_k$ plus 1 for the variance.
We draw $T = 1000$ samples of $\mu\mid Z_{1:n}$ (\cref{eq:NIW-posterior}) and compute the SLC scores $\hat{w}_{\delta}(k)$ using \cref{alg:proc} with complexities $c(k) = d_k$ and temperature $\alpha_n = n^{0.45}$.

To define the rescaled tolerances $\hat{\tau}$, we introduce both a noise model and a flexible nonparametric model to represent upper and lower bounds, respectively, on the KL divergences from the true DGP. The noise model is defined as $y_i \sim \mathcal{N}(\varphi, \sigma^2_{\mathrm{noise}})$, and the flexible model is $y_i \sim \mathcal{N}(\beta_{t_i},\sigma^2_{\mathrm{flex}})$ where $\beta = (\beta_t : t\in\mathcal{T})$ is defined at the set of 10 observed temperatures, $\mathcal{T}$. 
For the noise model and the flexible model, we compute the MLEs of the parameters and take the average negative log-likelihood to obtain $\hat{\mu}_{\mathrm{noise}}$ and $\hat{\mu}_{\mathrm{flex}}$, respectively; see \cref{supp:tpc} for details.
 We set $\hat{\tau} = \delta/(\muhat_{\mathrm{noise}} - \muhat_{\mathrm{flex}})$.

\cref{fig:tpc_results}(B) displays the distribution of $\mu_k-\min_j \mu_j$ for each $k$ under the posterior $\mu\mid Z_{1:n}$, and we overlay lines indicating the $\delta$ tolerance values used in panel (C).  In \cref{fig:tpc_results}(C), we show  the SLC scores $\hat{w}_{\delta}(k)$ for $\delta=0.187$ (orange), $0.075$ (purple), $0.037$ (green), and $0.001$ (blue). 
The SLC scores behave as follows: (i) small tolerance ($\delta = 0.001$, $\hat{\tau} = 0.1\%$) put higher weight on the most flexible models such as Sharpe-Schoolfield and Poly-5; (ii) moderate tolerance ($\delta \in \{0.037, \, 0.075\}$, $\hat{\tau}\in \{5\%, 10\%\}$) shifts selection to simpler 4--5 parameter models; (iii) large tolerance ($\delta = 0.187$, $\hat{\tau} = 25\%$) shifts to the 3-parameter models, with Bri\`ere I receiving the highest score.
Notably, the Eubank and Gaussian models receive relatively low weight in panel (C) across all $\delta$ values considered, indicating that these two models are not appropriate for this data set, regardless of the tolerance level.

\subsection{Population structure admixture models}
\label{sec:examples:population-structure}

Inferring latent population structure from multilocus genotype data is a central task in population genetics \citep{pritchard2000inference, rosenberg2002genetic}. In many realistic settings, individuals are admixed — that is, their genomes originate from multiple ancestral populations. Modeling this admixture requires decomposing each individual's genotype into proportions contributed by $k$ unknown source populations, while estimating allele frequencies in each population. Bayesian admixture models provide a formal probabilistic framework for this task.

A widely used software program for admixture models is STRUCTURE \citep{pritchard2000inference}, which performs posterior inference via Markov chain Monte Carlo. In the model, each individual is associated with a vector of ancestry proportions over $k$ populations, however, STRUCTURE requires the user to specify $k$ in advance. This introduces a model selection problem: The true number of populations is unknown, and must be chosen to balance fit and complexity. If $k$ is too small, genetically distinct groups may be merged, obscuring meaningful structure. If $k$ is too large, the model may chase noise and infer spurious populations. 

In practice, users often choose $k$ by comparing in-sample log-likelihoods or by applying \emph{ad hoc} heuristics such as Evanno's method \citep{evanno2005detecting}, but these procedures lack theoretical justification.
We address this problem by applying the LaD framework to perform principled, robust model selection in admixture models. The aim is to capture the explainable structure while avoiding unnecessary complexity.

We apply our methodology to the brook trout data set of \citet{erdman2022broadscale}, consisting of multilocus genotypes for $n=8454$ fish sampled from multiple streams in the Eastern United States. Each individual is genotyped at $L = 5$ loci with two allele copies per locus. Specifically, the genotype for individual $i$ is $x_i = [(x_{il1}, x_{il2})]_{l=1}^L$ where $x_{i l a} \in \{1, \ldots, V_l\}$ is the allele observed at locus $l$, allele copy $a \in \{1,2\}$. Following \citet{pritchard2000inference}, the admixture model assumes alleles are drawn from one of $k$ latent populations with allele frequencies $\phi_{jl} = (\phi_{jl1}, \ldots, \phi_{jlV_l}) \sim \mathrm{Dirichlet}(1,\ldots,1)$ for $j=1,\ldots, k$, according to ancestry proportions $q_i = (q_{i1}, \ldots, q_{ik}) \sim \mathrm{Dirichlet}(\psi)$, where $\psi = (\psi_1, \ldots, \psi_k)$ is given a hyperprior, $\psi_j\sim\mathrm{Uniform}(0,10)$ independently. Specifically,
\vspace{-1em}
\begin{align*}
    s_{i l a} &\sim \mathrm{Categorical}(q_i), \\
    x_{i l a} \mid s_{i l a} = j &\sim \mathrm{Categorical}(\phi_{jl}).
\end{align*}
The global parameters are $\theta := (\phi,\psi)$, whereas $q_i$ and $s_{i l a}$ are individual-specific.


We implement the LaD approach in \cref{alg:workflow} as follows.
For each $k = 1, \ldots, 10$, we do the following.  For a given $k$, we run STRUCTURE $20$ times and keep the run with the largest estimated log probability; see \cref{supp:admix_details}.
We take the posterior means $\hat{\phi}_{jl}$ and $\hat{\psi}$ from that run, and compute per-individual negative log-likelihoods $\ell_k(x_i; \hat\theta) = -\log p(x_i\mid \hat{\phi},\hat{\psi})$ by integrating out $q_i|\hat{\psi}$ and $s_{i l a}|q_i$ using Monte Carlo; see \cref{supp:admix_details} for details.
It is necessary to integrate out $q_i$ and $s_{i l a}$ when defining the negative log-likelihood values because these are observation-specific latent variables.
The bias-corrected LaD values are then $Z_{i k} = \ell_k(x_i; \hat\theta) + d_k/(2 n)$ where $d_k = k \sum_{l=1}^L (V_l - 1) + k$, since each $\phi_{j l}$ contributes $V_l - 1$ degrees of freedom and $\psi$ contributes $k$. We then apply \cref{alg:proc} with $\alpha_n = n^{0.45}$ and complexity $c(k) = k$ to compute the SLC scores $\hat{w}_{\delta}(k)$.
To compute the rescaled tolerances $\hat{\tau}$, we take the noise model to be 
discrete uniform at each locus: $x_{i l a} \sim \mathrm{Uniform}(\{1,\ldots, V_l\})$ independently for all $i,l,a$, which yields $\muhat_{\mathrm{noise}} = \frac{1}{n} \sum_{i=1}^n \sum_{l=1}^L N_{il} \log V_l$, where $N_{il} \in \{0,1,2\}$ is the number of observed copies at locus $l$ for individual $i$. 



\begin{figure}
    \centering
    \includegraphics[width=0.99\linewidth]{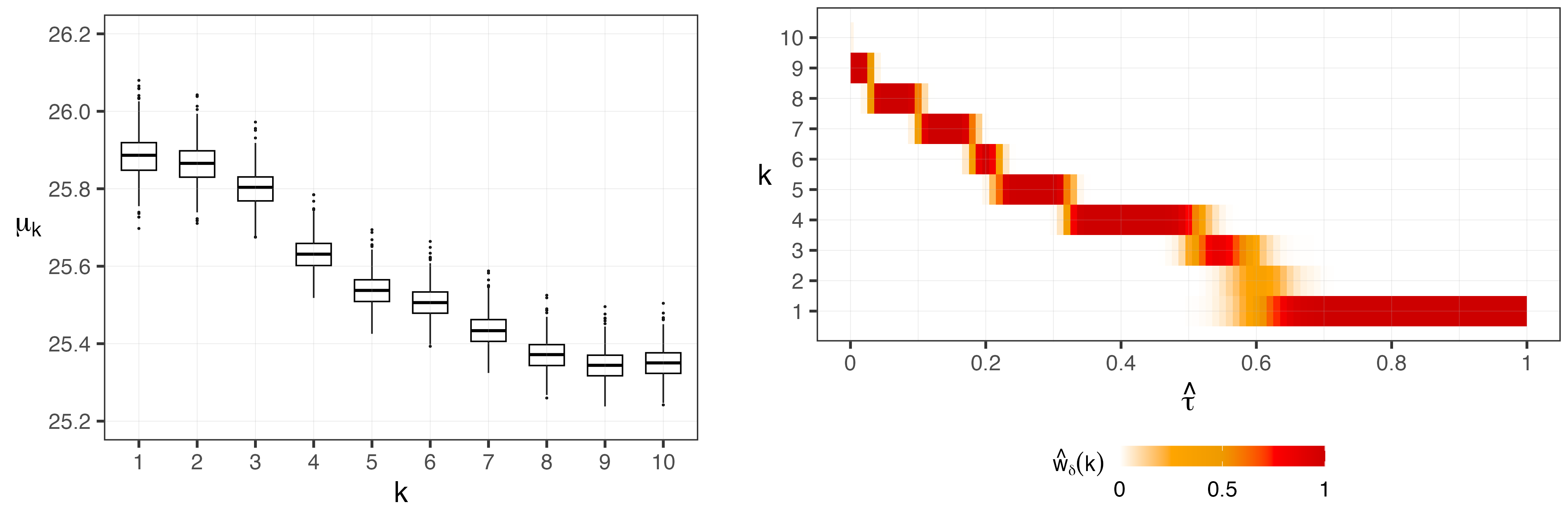}
    \captionsetup{font=small}
    \caption{\emph{(Left)} Boxplot of the posterior of $\mu_k$ for each candidate number of populations $k=1,\ldots,10$. 
    \emph{(Right)} Posterior path showing the SLC scores for $k=1,\ldots, 10$ (y-axis) versus the rescaled tolerance $\hat{\tau}\in [0,1]$ (x-axis), where $\hat{\tau} = \delta/ (\muhat_{\mathrm{noise}} - \min_k \muhat_k)$. Darker red indicates a higher SLC score.}
    \label{fig:brooktrout_lad_pp}
\end{figure}

Figure \ref{fig:brooktrout_lad_pp} shows the LaD results for the brook trout data.
The boxplots (left) show the posterior distribution of $\mu_k$ for each candidate number of populations $k$, and the heatmap (right) shows the SLC score $\hat{w}_{\delta}(k)$ as a function of the rescaled tolerance $\hat{\tau}$. As we increase $\hat{\tau}$ from 0 to 1, the selected number of populations $k$ gradually transitions from higher (more complex models) to lower (simpler models). This provides clear insights into the tradeoff between model complexity and the cost in terms of KL divergence.
For instance, the $k = 4$ population model can capture 65\% of the improvement in fit (that is, $\hat\tau = 0.35$) from the noise model to the most complex candidate model.


For comparison, \cref{fig:brooktrout_evanno} shows the results for Evanno's method \citep{evanno2005detecting}, which is widely used to select $k$ for STRUCTURE.
Evanno's method is based on finding a sharp elbow in the mean log-likelihood $L(k)$ (see \cref{supp:admix_details}), or more precisely, a large value of $\Delta k$, defined as the second-order difference of $L(k)$ divided by the standard deviation of $L(k)$ across runs.
Evanno's method strongly suggests that $k = 2$ should be selected on these data; see \cref{fig:brooktrout_evanno}.
Interestingly, however, our LaD results in \cref{fig:brooktrout_lad_pp} show that $k = 2$ is not a particularly useful choice.
Using $k = 2$ captures around 40\% of the possible improvement in fit ($\hat\tau = 0.6$), but \cref{fig:brooktrout_lad_pp} shows that using $k = 1$ is nearly as good as $k = 2$ on these data in terms of fit.  It is necessary to jump up to $k = 3$ or $4$ to obtain more substantial  gains in fit to the true DGP. 



\begin{figure}
    \centering
    \includegraphics[width=1\linewidth]{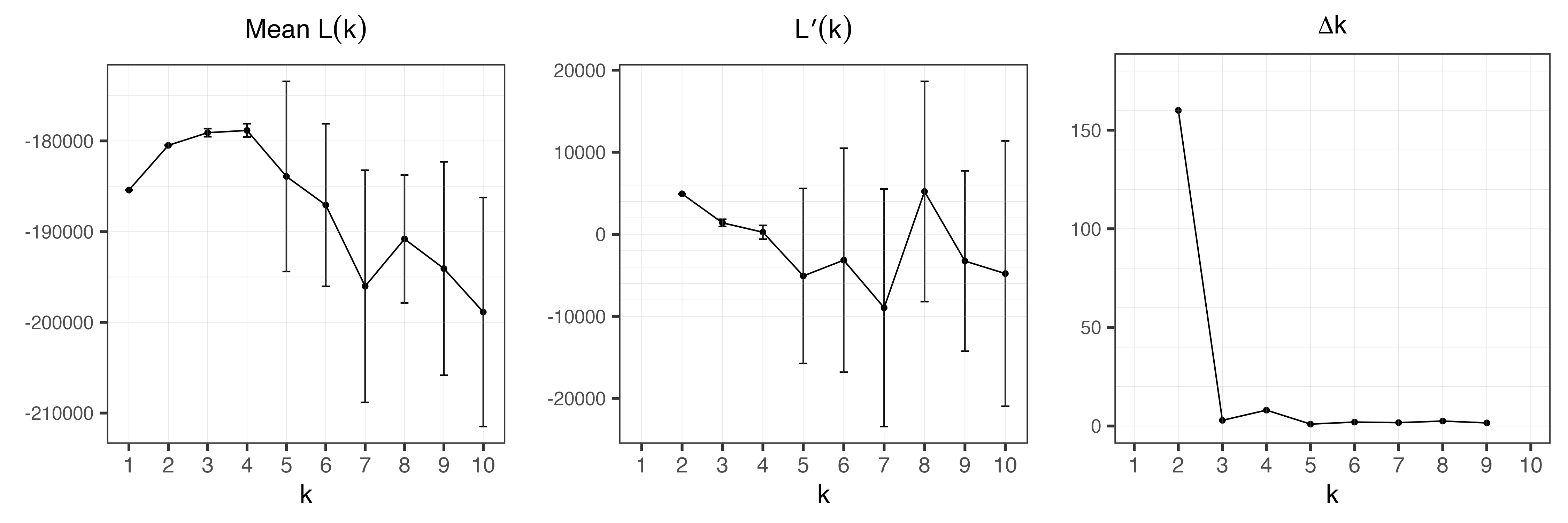}
    \captionsetup{font=small}
    \caption{Evanno's method. 
    \emph{(Left)} Mean $L(k)$ increases with $k$; error bars show standard deviation across replicates. 
    \emph{(Middle)} First difference $L'(k)$.
    \emph{(Right)} $\Delta k = \E|L''(k)| / \text{sd}\big(L(k) \big)$ peaks at $k=2$. Here, $|L''(k)|$ is the absolute second difference.}
    \label{fig:brooktrout_evanno}
\end{figure}

\section{Discussion}\label{sec:conclusion}

The Likelihood as Data (LaD) framework is a novel approach to comparing models based on treating per-observation negative log-likelihoods as data to be analyzed.
LaD quantifies uncertainty in the KL divergence between the true DGP and each candidate model, enabling robust model selection under misspecification by quantitatively characterizing the tradeoff between simplicity and fit.
The LaD approach is simple to compute, asymptotically consistent, and stable in the presence of ties. 

Although LaD is a versatile method, its performance hinges on several assumptions. These include independent observations, finite second moments of the per-observation negative log-likelihoods, and accurate point estimators for the model parameters. If these assumptions are not satisfied, LaD may be unstable or exhibit overconfidence. Furthermore, when there are a large number of candidate models, LaD may encounter difficulties due to the computational burden of exhaustively evaluating all models and estimating the full $K \times K$ covariance $\Sigma$. 

There are several interesting directions for future work.
Placing a prior on $\mu_1,\ldots,\mu_K$ that allows $\mu_k = \mu_j$ with positive probability may improve the stability of using the plug-in posterior $P(k\in M^*_\delta(\mu) \mid Z_{1:n})$ rather than our SLC score. 
Extending LaD to general loss functions beyond negative log-likelihood and developing a frequentist analogue based on sampling distributions are also promising directions.
Large $K$ could be handled by clustering the columns of $Z$ and working at the cluster level rather than over all $K$. Finally, the LaD idea can be used more broadly beyond model selection. For example, one could explore the relationships among models and among data points via dimensionality reduction of the $Z$ matrix. 

\section*{Acknowledgments}

We thank Jackson Loper, David Frazier, Colin Kremer, Zhiyi Chi, and Dipak Dey for helpful conversations.
N.A.S was supported in part by the Makuch Faculty Fund Award in Mathematical and Data Sciences at University of Connecticut.
J.W.M. was supported in part by the National Cancer Institute of the NIH under award number R01CA240299.
The content is solely the responsibility of the authors and does not necessarily represent the official views of the National Institutes of Health.




\putbib 
\end{bibunit}


\clearpage
\setcounter{page}{1}
\setcounter{section}{0}
\setcounter{table}{0}
\setcounter{figure}{0}
\setcounter{algorithm}{0}
\renewcommand{\theHsection}{SIsection.\arabic{section}}
\renewcommand{\theHtable}{SItable.\arabic{table}}
\renewcommand{\theHfigure}{SIfigure.\arabic{figure}}
\renewcommand{\theHalgorithm}{SIalgorithm.\arabic{algorithm}}
\renewcommand{\thepage}{S\arabic{page}}  
\renewcommand{\thesection}{S\arabic{section}}   
\renewcommand{\thetable}{S\arabic{table}}   
\renewcommand{\thefigure}{S\arabic{figure}}
\renewcommand{\thealgorithm}{S\arabic{algorithm}}

\begin{bibunit}
\begin{center}
{\Large Supplementary material for ``Robust model selection using likelihood as data''}
\end{center}

\thispagestyle{empty}

\section{Previous work}\label{supp:previous}

To the best of our knowledge, no prior work proposes LaD in the form we use or develops the theoretical motivation we rely on.
This section briefly summarizes the most closely related previous work.

The core idea behind LaD is related to Vuong’s test \citep{vuong1989likelihood}, which compares two parametric models using the average per-observation log-likelihood difference and tests the null that the models are equally close to the true DGP, where ``closeness'' is defined in terms of KL divergence.
Similarly, \citet{diebold1995} propose a test for equal predictive accuracy of two competing forecasts, again yielding a pairwise testing tool rather than a joint analysis of multiple models.
Several extensions of Vuong’s test and other related model selection methods have been proposed (see, for example, \citealp{shimodaira2001multiple, royall2003interpreting, royall2004likelihood, clarke2007simple, barmalzan2013model, diebold2015comparing, schennach2017simple, sayyareh2022testing, bruck2023corrected}), but none of these are close to the LaD method, which extends to more than two models and jointly quantifies uncertainty in how well each model approximates the DGP.

Several methods have been proposed for finding a parsimonious model whose fit is nearly as good as the best available model. 
Widely used heuristics such as the one-standard-error rule \citep{hastie2009elements, chen2021one} select the most parsimonious model whose cross-validated error is within one estimated standard error of the minimum.
\citet{hansen2011model} introduce the model confidence set (MCS), which compares multiple candidate models via sequential testing and elimination: clearly inferior models are removed until the remaining set contains the best model(s) with a chosen confidence level.
In Bayesian work, \citet{goutis1998model} and \citet{dupuis2003variable} quantify the loss of explanatory power incurred by KL projection from a full model to a submodel and retain only those submodels whose estimated loss stays below a pre-specified threshold; see also \citet{vehtari2004model} and \citet{Lindsay2009model}.
Our LaD method also provides a way of tackling this goal, but operates very differently: it works directly with expected log-likelihood (equivalently, KL distance), incorporates model complexity in a transparent way, and yields a Bayesian posterior over the full vector of KL distances.

Another line of work addresses misspecification by modifying the posterior update itself. 
In contrast with classical robust Bayesian analysis, which focuses on sensitivity to the prior \citep{berger1994overview}, these more recent approaches robustify likelihood-based inference by modifying the likelihood, for instance, using the coarsened posterior \citep{miller2019robust} or generalized Bayes updates based on loss functions \citep{bissiri2016general}. 
There are numerous advances in this direction for parametric inference \citep{knoblauch2022optimization, fong2023martingale, huggins2024reproducible, frazier2025synthetic, dewaskar2025robustifying}, but not primarily for model selection.

Meanwhile, for robust model selection, there have also been recent advances in Bayesian methodology.
\citet{shao2019bayesian} introduce a method that replaces the log marginal likelihood with the Hyvärinen score, a proper scoring rule to obtain Bayes-factor-like comparisons that remain well defined under diffuse or improper priors.
\citet{li2020comparing} propose D-probabilities that compare parametric models to a flexible nonparametric reference via KL divergence, yielding complexity-penalized weights that are less sensitive to prior specification than standard posterior model probabilities. 
\citet{huggins2023reproducible} consider BayesBag, a robust Bayesian model selection method that applies bagging to the posterior by averaging posterior model probabilities over bootstrap-resampled versions of the data.
More recently, \citet{shirvaikar2024general} propose a probabilistic framework for model uncertainty based on forward sampling of missing observations using one-step-ahead prediction.

There is a mature literature on predictive Bayesian approaches to model assessment, selection, and averaging---including cross-validation and predictive aggregation methods \citep{vehtari2012survey, yao2018using, sivula2025uncertainty, mclatchie2025advances}. These approaches are appealing for predictive performance, but their target is typically a predictive utility rather than uncertainty quantification for the KL divergence between the true DGP and each candidate model.


\section{Proofs}\label{supp:proofs}

In this section, we provide proofs of the results in \cref{sec:theory}. For $x\in\bbR^d$, $\|x\|$ denotes the Euclidean norm and for $A\in\bbR^{d\times d}$, $\|A\|$ denotes the Frobenius norm.
For $r > 0$ and $x_0\in\bbR^d$, define the open ball $B_r(x_0) := \{x \in \bbR^d: \, \|x - x_0\| < r\}$.

\begin{condition}\label{cond2}
    Suppose $X, X_1, X_2, \dots, X_n$ are i.i.d.\ random variables defined on a probability space $(\cX, \cF, P)$. Assume the following conditions for all models $k=1,\ldots,K$.  Suppose $\Theta_k \subseteq \mathbb{R}^{d_k}$ is open and $\ell_k: \cX \times \Theta_k \to \mathbb{R}$ is measurable. Throughout, expectations are taken with respect to $P$, and $\Theta_k$ is endowed with the Borel sigma-algebra. 
    Let $\Theta^*_k = \argmin_{\theta_k \in \Theta_k} \E\big(\ell_k(X; \theta_k)\big)$, assumed to be nonempty.
    Assume the following conditions hold for some fixed $\theta^*_k \in \Theta^*_k$ and some $r > 0$ such that $B_{r}(\theta^*_k) \subseteq \Theta_k$.
    \begin{enumerate}[label=(\alph*)]
        \item $\E \big( |\ell_k(X; \theta_k)|\big) < \infty$ for all $\theta_k \in B_{r}(\theta^*_k)$.
        
        \item For all $x \in \cX$, the map $\theta_k \mapsto \ell_k(x; \theta_k)$ is twice continuously differentiable on $B_{r}(\theta^*_k)$.
        
        \item There exists a function $H: \, \cX \to [0, \infty)$ such that $\sup_{\theta_k \in B_{r}(\theta^*_k)} \| \nabla^2_{\theta_k} \ell_k(x;\theta_k) \| \le H(x)$ for all $x\in\cX$, and $\E \big(H(X) \big)< \infty$.

        \item $\E \Big( \big\| \nabla_{\theta_k} \ell_k(X; \theta^*_k) \big\|^2 \Big) < \infty$.
     
        \item There exists an estimator $\hat{\theta}_k = \hat{\theta}_k(X_1,\ldots,X_n) \in \Theta_k$ such that $\|\hat{\theta}_k - \theta^*_k \| = O_p(1/\sqrt{n})$.
    \end{enumerate}  
\end{condition}


We use the following lemma to combine local bounds with high probability events.

\begin{lemma}\label{supp:lemma:high-prob-bound}
    Let $A_1,A_2,\ldots$ be events such that $P(A_n) \to 1$ as $n \to \infty$. Suppose $X_n$ and $Y_n$ are random variables such that $|X_n| \le Y_n$ on $A_n$, and $Y_n \ge 0$ almost surely. 
    If $Y_n = o_p(1)$, then $X_n = o_p(1)$. If $Y_n = O_p(1)$, then $X_n = O_p(1)$.
\end{lemma}
\begin{proof}[Proof of \cref{supp:lemma:high-prob-bound}]
    First, assume $Y_n = o_p(1)$. Let $\epsilon >0$. Then 
    \begin{equation}
        P\big( |X_n| > \epsilon \big) \le P( A_n^c) + P\big(|X_n|>\epsilon, \, A_n\big) \le P(A_n^c) + P(Y_n > \epsilon) \longrightarrow 0
    \end{equation}
    because $Y_n = o_p(1)$ and $P(A_n) \to 1$. Since $\epsilon > 0$ is arbitrary, $X_n = o_p(1)$.
    Next, assume $Y_n = O_p(1)$. Let $\delta >0$. Since $Y_n = O_p(1)$, there exists $M < \infty$ such that $P\big(|Y_n| > M) < \delta/2$ for all sufficiently large $n$. Since also $P(A_n^c) \to 0$, then for all sufficiently large $n$,
    \begin{equation}
        P\big( |X_n| > M \big) \le P(A_n^c) + P\big(|X_n|>M, \, A_n\big) \le P(A_n^c) + P\big(|Y_n| > M\big) < \delta.
    \end{equation}
    This shows that $X_n = O_p(1)$.
\end{proof}

\subsection{Plug-in approximation error}

\cref{supp:theorem:param} establishes that replacing the unknown optimal parameters $\theta_k^*$ with their estimates $\hat{\theta}_k$ introduces error $O_p(1/n)$ in the average log-likelihood, which is negligible compared to the empirical approximation error of order $O_p(1/\sqrt{n})$. 

\begin{theorem}\label{supp:theorem:param}
    Assume \cref{cond2} holds for each $k=1, \ldots, K$. Define
    \begin{equation*}
        \oline{Z}_{n k}(\theta_k) = \frac{1}{n}\sum_{i=1}^n \ell_k(X_i; \theta_k),
    \end{equation*}
    for $\theta_k \in \Theta_k$. Let $\oline{Z}_n(\theta) = \big( \oline{Z}_{n 1}(\theta_1),\ldots, \oline{Z}_{n K}(\theta_K) \big)^\mathtt{T}$ for $\theta = (\theta_1, \ldots, \theta_K)$. Then
    \begin{equation}
    \label{eq:R_n-error}
        \oline{Z}_n(\hat{\theta}) = \oline{Z}_n(\theta^*) + O_p(1/n).
    \end{equation}
\end{theorem}



\begin{proof}[Proof of \cref{supp:theorem:param}]
We fix $k$ and show \cref{eq:R_n-error} holds in the $k$th component. For notational clarity, we suppress the subscript $k$ in the proof, and we define the function $R_n(\theta) = \tfrac{1}{n} \sum_{i=1}^n \ell(X_i; \theta)$ to represent $\oline{Z}_{n k}(\theta_k)$.
Define the event $A_n = \{\|\hat\theta - \theta^*\| < r\}$, that is, $A_n = \{\hat{\theta}\in B_r(\theta^*)\}$. Then $P(A_n) \to 1$ as $n\to\infty$, since $\|\hat\theta - \theta^*\| = o_p(1)$ by \cref{cond2}(e).
On the event $A_n$, since $B_r(\theta^*)$ is convex, the line segment between $\theta^*$ and $\hat\theta$ is contained in $B_r(\theta^*)$. 
Furthermore, for all $x \in \cX$, the map $\theta \mapsto \ell(x; \theta)$ is twice continuously differentiable on $B_r(\theta^*)$ by \cref{cond2}(b).
Thus, on event $A_n$, Taylor's theorem implies that there exists some $\tilde{\theta}$ on the line segment between $\theta^*$ and $\hat{\theta}$ such that 
\begin{equation}\label{eq:Rn_taylor}
    R_n(\hat{\theta}) = R_n(\theta^*) + \nabla_{\theta} R_n(\theta^*)^\texttt{T} (\hat{\theta} - \theta^*) + \frac{1}{2} (\hat{\theta} - \theta^*)^\texttt{T} \nabla_{\theta}^2 R_n(\tilde{\theta}) (\hat{\theta} - \theta^*).
\end{equation}
To make $\tilde{\theta}$ well-defined on the whole sample space, we define $\tilde{\theta} = \theta^*$ on $A_n^c$. 
Then, by \cref{supp:lemma:high-prob-bound}, it suffices to show that the first- and second-order terms in \cref{eq:Rn_taylor} are $O_p(1/n)$.
Since $P(A_n) \to 1$, this will imply that $|R_n(\hat{\theta}) - R_n(\theta^*)| = O_p(1/n)$ by \cref{supp:lemma:high-prob-bound}, proving the result (\cref{eq:R_n-error}).

We claim the first-order term of \cref{eq:Rn_taylor} is $O_p(1/n)$. To establish this, we will show that $\|\nabla_{\theta} R_n(\theta^*)\| = O_p(1/\sqrt{n})$ and then multiply by $\big\| \hat\theta - \theta^* \big\| = O_p(1/\sqrt{n})$, using \cref{cond2}(e). 
Fix $x \in \cX$. By \cref{cond2}(b), for any $\theta \in B_r(\theta^*)$ we can write
\begin{equation}\label{eq:grad_int}
    \nabla_\theta \ell(x; \theta) - \nabla_\theta \ell(x; \theta^*) = \int_0^1 \nabla_\theta^2 \ell \big(x;\, \theta^* + t(\theta-\theta^*) \big) (\theta-\theta^*)\,dt
\end{equation}
by applying the fundamental theorem of calculus to each entry of the function $f(t) = \nabla_\theta \ell(x; \, \theta^* + t(\theta-\theta^*))$ for $t\in[0,1]$.
Taking norms and using \cref{cond2}(c),
\begin{align}
    \big\|\nabla_\theta \ell(x;\theta) - \nabla_\theta \ell(x;\theta^*)\big\| &\le \int_0^1 \Big\|\nabla_\theta^2 \ell\big(x;\,\theta^* + t(\theta-\theta^*)\big)\Big\|\,\|\theta-\theta^*\|\,dt  \\
    &\le \sup_{\theta' \in B_r(\theta^*)}\|\nabla_\theta^2 \ell(x; \theta')\| \, \|\theta-\theta^*\| \\ 
    &\le H(x)\, \|\theta-\theta^*\|.
\end{align}
In particular, for each fixed $x$, $\nabla_\theta \ell(x; \theta) \to \nabla_\theta \ell(x; \theta^*)$ as $\theta \to \theta^*$. Moreover, by the reverse triangle inequality, for all $\theta\in B_r(\theta^*)$,
\begin{align}\label{eq:grad_dom}
    \|\nabla_\theta \ell(x;\theta)\| &\le \|\nabla_\theta \ell(x;\theta^*)\| + H(x)\|\theta-\theta^*\| \\
    &\le \|\nabla_\theta \ell(x;\theta^*)\| + rH(x).
\end{align}

For each coordinate $j$ and each $\theta \in B_r(\theta^*)$,
\begin{equation}
     \Big| \frac{\partial}{\partial \theta_j} \ell(X; \theta) \Big| \le \big\| \nabla_\theta \ell(X; \theta) \big\| \le \big\| \nabla_\theta \ell(X; \theta^*) \big\| + r H(X),
\end{equation}
and the right-hand side is integrable by \cref{cond2}(c,d). Therefore, we may interchange differentiation and expectation componentwise \citep[Theorem 2.27(b)]{folland1999real}, which implies that $R(\theta) = \E\big(\ell(X; \theta)\big)$ is differentiable at $\theta^*$ and $\nabla_\theta R(\theta^*) = \E \big(\nabla_\theta \ell(X; \theta^*) \big)$.
Since $\Theta$ is open and $\theta^*$ is a minimizer of $R(\theta)$, it follows that $\nabla_\theta R(\theta^*) = 0$, and hence, $\E \big(\nabla_\theta \ell(X; \theta^*) \big) = 0$. Hence, the central limit theorem implies 
\begin{equation}
    \sqrt{n} \nabla_\theta R_n(\theta^*) = \frac{1}{\sqrt{n}} \sum_{i=1}^n \nabla_\theta \ell(X_i; \theta^*) \xrightarrow[n \to \infty]{\mathrm{d}} \cN\Big(0, \, \C\big( \nabla_\theta \ell(X; \theta^*)  \big)  \Big)
\end{equation}
since the covariance matrix $\C(\nabla_\theta \ell(X; \theta^*))$ exists and is finite due to \cref{cond2}(d).
This shows that $\|\nabla_\theta R_n(\theta^*)\| = O_p(1/\sqrt{n})$. Thus, by Cauchy--Schwarz,
\begin{equation}
\label{eq:first_order_term}
    |\nabla_{\theta} R_n(\theta^*)^\texttt{T} (\hat{\theta} - \theta^*)| =  O_p(1/\sqrt{n}) \times  O_p(1/\sqrt{n}) =  O_p(1/n).
\end{equation}

Next, consider the second-order term of \cref{eq:Rn_taylor}. 
By definition, we always have $\tilde{\theta}\in B_r(\theta^*)$ since $\hat{\theta}\in B_r(\theta^*)$ on $A_n$ and $\tilde{\theta} = \theta^*$ on $A_n^c$. Hence,
\begin{align}
    \big\| \nabla^2_\theta R_n(\tilde{\theta}) \big\| &= \Big\| \frac{1}{n} \sum_{i=1}^n \nabla^2_\theta \ell(X_i; \tilde{\theta}) \Big\| \\
                                              &\le \frac{1}{n} \sum_{i=1}^n \sup_{\theta \in B_r(\theta^*)} \big\| \nabla^2_\theta \ell(X_i; \theta) \big\| \\
                                              &\le \frac{1}{n} \sum_{i=1}^n H(X_i)
\end{align}
by \cref{cond2}(c). Therefore,
\begin{equation}
\label{eq:second_order_term}
    \Big| \frac{1}{2} (\hat{\theta} - \theta^*)^\texttt{T} \nabla_\theta^2 R_n(\tilde{\theta})  (\hat{\theta} - \theta^*) \Big| \le \frac{1}{2} \|\hat{\theta} - \theta^*\|^2\, \frac{1}{n} \sum_{i=1}^n H(X_i) = O_p(1/n)
\end{equation}
since  $\big\| \hat\theta - \theta^* \big\| = O_p(1/\sqrt{n})$ by \cref{cond2}(e) and $\frac{1}{n}\sum_{i=1}^n H(X_i) \xrightarrow[]{\mathrm{a.s.}} \E \big(H(X)\big) < \infty$ by the strong law of large numbers.
\end{proof}

\subsection{Bias correction is asymptotically negligible}

While the bias correction in \cref{sec:method:param:overfitting} is useful in smaller sample settings, it is simpler to work with the uncorrected LaD values in theoretical analyses. To this end, our next result shows that, asymptotically, the bias-corrected negative log-likelihood $\oline{Z}_{n}^{\mathrm{bc}}(\hat\theta)$ behaves like the uncorrected version $\oline{Z}_{n}(\hat\theta)$; compare \cref{eq:R_n_bc-error} to \cref{eq:R_n-error}.  This means that asymptotically, the validity of frequentist inference for $\mu^0$ based on $\oline{Z}_n(\hat\theta)$ (\cref{sec:method:param}) is unaffected by the bias correction. 


\begin{corollary}\label{supp:cor:bc}
    Assume \cref{cond2} holds and the loss is the negative log-likelihood, $\ell_k(x; \theta_k) = -\log f_k(x; \theta_k)$, with parameter dimension $d_k = \dim(\Theta_k)$. Define
    \begin{equation}
        \oline{Z}_n^{\mathrm{bc}}(\hat\theta) := \oline{Z}_n(\hat\theta) + \frac{d}{2 n}
    \end{equation}
    where $d = (d_1,\ldots,d_K)^\mathtt{T}\in \mathbb{R}^K$ and $\oline{Z}_n(\theta)$ is defined as in \cref{supp:theorem:param}. Then 
    \begin{equation}\label{eq:R_n_bc-error}
        \oline{Z}_{n}^{\mathrm{bc}}(\hat\theta) = \oline{Z}_{n}(\theta^*) + O_p(1/n).
    \end{equation}
    In particular, if $\sqrt{n}\big( \oline{Z}_n(\hat\theta) - \mu^0 \big)$ converges in distribution, then $\sqrt{n}\big( \oline{Z}_n^{\mathrm{bc}}(\hat\theta) - \mu^0 \big)$ converges in distribution to the same limit.
\end{corollary}

\begin{proof}[Proof of \cref{supp:cor:bc}]\label{supp:cor:bc:proof}
    By \cref{supp:theorem:param}, 
    $$
    \oline{Z}_{n}^{\mathrm{bc}}(\hat\theta) = \oline{Z}_n(\hat\theta) + \frac{d}{2 n} = \oline{Z}_n(\theta^*) + O_p(1/n) + \frac{d}{2 n} = \oline{Z}_n(\theta^*) + O_p(1/n),
    $$
    which is \cref{eq:R_n_bc-error}.
    Moreover,
    $$
    \sqrt{n}\big( \oline{Z}_n^{\mathrm{bc}}(\hat\theta) - \mu^0 \big) = \sqrt{n}\big( \oline{Z}_n(\hat\theta) - \mu^0 \big) + \frac{d}{2\sqrt{n}},
    $$
    and $d/(2\sqrt{n}) \to 0$, so the two sequences have the same weak limit by Slutsky's theorem.
\end{proof}

\subsection{Asymptotics of the data-averaged LaD posterior}
\label{sec:posterior-asymptotics}

We now return to considering all $K$ models, and focus on the Bayesian posterior in \cref{sec:method:bayesinf}.
Here, we set $Z_{i k} = \ell_k(X_i; \hat{\theta}_k)$, and as in \cref{sec:method:bayesinf}, $Z_i = (Z_{i 1},\ldots,Z_{i K})^\mathtt{T}$, $\oline{Z}_n = \frac{1}{n}\sum_{i=1}^n Z_i$, $S_n = \frac{1}{n}\sum_{i=1}^n (Z_i - \oline{Z}_n)(Z_i - \oline{Z}_n)^\mathtt{T}$, and $\mu_n = \big( \lambda_0 \mu_0 + n\oline{Z}_n \big) / \lambda_n$.
We use the notation $Z_i(\theta^*)$, $\oline{Z}_n(\theta^*)$, $S_n(\theta^*)$, and $\mu_n(\theta^*)$, to denote the corresponding values with $\theta^*_k$ in place of $\hat{\theta}_k$ for all $k$.
Also, recall the definitions of $\mu^0$ and $\Sigma^0$ in \cref{eq:mu_cov}.

\begin{theorem}\label{supp:theorem:niw_clt_a}
    Assume Conditions \ref{cond1} and \ref{cond2} hold. Then
    $$
        \sqrt{n} \big( \mu_n - \mu^0 \big) \xrightarrow[n\to\infty]{\mathrm{d}} \cN(0, \Sigma^0).
    $$
\end{theorem}

\begin{proof}[Proof of \cref{supp:theorem:niw_clt_a}]\label{supp:theorem:niw_clt_a:proof}
    By \cref{supp:theorem:param} applied coordinate-wise to the vector $\oline{Z}_n$, we have  $\oline{Z}_n = \oline{Z}_n (\theta^*) + O_p\big(1/n \big)$. It follows that 
    \begin{equation}
        \sqrt{n} \big( \mu_n - \mu_n(\theta^*) \big) = \sqrt{n} \frac{n}{\lambda_n} \big( \oline{Z}_n - \oline{Z}_n (\theta^*) \big) = \sqrt{n} \,O_p(1/n) = o_p(1),
    \end{equation}
    because $n / \lambda_n = n/(\lambda_0 + n) \to 1$ as $n \to \infty$. Thus,
    \begin{equation}\label{eq:data_decomp}
        \sqrt{n} \big( \mu_n - \mu^0 \big) = \sqrt{n} \big(\mu_n(\theta^*) - \mu^0 \big) + o_p(1).
    \end{equation}
    Using $\lambda_n = \lambda_0 + n$, we have
    \begin{equation}\label{eq:mu_n_decomp}
        \sqrt{n} \big( \mu_n(\theta^*) - \mu^0 \big) = \frac{n}{\lambda_n} \sqrt{n}\big( \oline{Z}_n(\theta^*) - \mu^0 \big) + \sqrt{n} \frac{\lambda_0}{\lambda_n} \big( \mu_0 - \mu^0 \big).
    \end{equation}
    Since $n/\lambda_n \to 1$, $\sqrt{n}\big( \oline{Z}_n(\theta^*) - \mu^0 \big) \xrightarrow[]{\mathrm{d}} \cN(0, \Sigma^0)$ by the CLT, and the second term of \cref{eq:mu_n_decomp} is $O(1/\sqrt{n})$, Slutsky's theorem gives 
    \begin{equation}\label{eq:mu_n_conv}
        \sqrt{n} \big( \mu_n(\theta^*) - \mu^0 \big)\xrightarrow[n \to \infty]{\mathrm{d}} \cN(0, \Sigma^0).
    \end{equation}
    Hence, with \cref{eq:data_decomp}, another application of Slutsky's theorem yields 
    \begin{equation}\label{eq:data_decomp_limit}
        \sqrt{n} \big( \mu_n - \mu^0 \big) \xrightarrow[n\to\infty]{\mathrm{d}} \cN(0, \Sigma^0).
    \end{equation}
    This completes the proof.
\end{proof}

\begin{theorem}\label{supp:theorem:Sn_consistency}
    Assume Conditions \ref{cond1} and \ref{cond2} hold. Then
    $$
        S_n \xrightarrow[n \to \infty]{\mathrm{p}} \Sigma^0.
    $$
\end{theorem}

\begin{proof}[Proof of \cref{supp:theorem:Sn_consistency}]\label{supp:theorem:Sn_consistency:proof}
    Under Condition \ref{cond1}, the strong law of large numbers implies that $S_n(\theta^*) \xrightarrow[n\to\infty]{\mathrm{a.s.}} \Sigma^0$. Define $Y_i := Z_i - \oline{Z}_n$ and $Y_i(\theta^*) := Z_i(\theta^*) - \oline{Z}_n(\theta^*)$. Then
    \begin{align}
    \label{eq:Sn_diff}
        S_n - S_n(\theta^*) &= \frac{1}{n} \sum_{i=1}^n \big( Y_i Y_i^\texttt{T} - Y_i(\theta^*)Y_i(\theta^*)^\texttt{T} \big) \\
        &=\frac{1}{n} \sum_{i=1}^n \Big( \big( Y_i - Y_i(\theta^*) \big) Y_i^\texttt{T} + Y_i(\theta^*) \big( Y_i - Y_i(\theta^*) \big)^\texttt{T} \Big).
    \end{align}
    
    By \cref{eq:Sn_diff} and the fact that $\|uv^\mathtt{T}\| = \|u\| \|v\|$ for the Frobenius norm,
    \begin{equation}
        \| S_n - S_n(\theta^*) \| \le \frac{1}{n} \sum_{i=1}^n \|Y_i - Y_i(\theta^*) \| \big( \|Y_i\| + \|Y_i(\theta^*)\| \big).
    \end{equation}
    By Cauchy--Schwarz,
    \begin{equation} 
    \label{eq:Sn_diff_norm}
        \| S_n - S_n(\theta^*) \| \le \Big(\frac{1}{n}\sum_{i=1}^n \|Y_i - Y_i(\theta^*)\|^2\Big)^{1/2}
        \Big(\frac{1}{n}\sum_{i=1}^n \big(\|Y_i\| + \|Y_i(\theta^*)\| \big)^2 \Big)^{1/2}.
    \end{equation}
    
    We are working with all $K$ models, so $\theta = (\theta_1,\ldots, \theta_K)$ denotes the concatenated parameter vector in $\bbR^D$ with $D = \sum_{k=1}^K d_k$. Also, $Z_i(\theta) \in \bbR^K$, so $\nabla_\theta Z_i(\theta)$ denotes the $K\times D$ Jacobian matrix of the map $\theta \mapsto Z_i(\theta)$, and $\| \nabla_\theta Z_i(\theta) \|$ denotes its Frobenius norm. 
    
    We next bound the two factors in \cref{eq:Sn_diff_norm} on the event $A_n := \big\{ \|\hat\theta - \theta^*\| < r\big\}$.
    First, consider the first factor of \cref{eq:Sn_diff_norm}. Since $Y_i - Y_i(\theta^*) = Z_i - Z_i(\theta^*) - \big( \oline{Z}_n - \oline{Z}_n (\theta^*) \big)$ and $\|u-v\|^2 \le 2\|u\|^2 + 2\|v\|^2$,
    \begin{equation}
    \label{eq:avg_y_diff}
        \frac{1}{n}\sum_{i=1}^n \|Y_i - Y_i(\theta^*)\|^2 \le \frac{2}{n}\sum_{i=1}^n \|Z_i - Z_i(\theta^*)\|^2 + 2\|\oline{Z}_n - \oline{Z}_n (\theta^*) \|^2.
    \end{equation}
    By \cref{supp:theorem:param}, $\|\oline{Z}_n - \oline{Z}_n (\theta^*) \| = O_p(1/n)$, hence the second term of \cref{eq:avg_y_diff} is $O_p(1/n^2)$.
    For the first term, on $A_n$, apply Taylor expansion of $\theta \mapsto Z_i(\theta)$ to each coordinate of $Z_i(\theta)$ around $\theta^*$: for each $i$ we have
    \begin{equation}
        Z_i - Z_i(\theta^*) = \nabla_\theta Z_i(\theta^*) (\hat{\theta} - \theta^*) + R_i, \qquad \|R_i\| \le \frac{1}{2}H(X_i) \|\hat\theta - \theta^*\|^2,
    \end{equation}
    where the remainder bound uses \cref{cond2}(c) coordinate-wise.  
    Therefore, since $\|u + v\|^2 \le 2\|u\|^2 + 2\|v\|^2$,
    \begin{equation}
        \|Z_i - Z_i(\theta^*) \|^2 \le 2\| \nabla_\theta Z_i(\theta^*)\|^2 \|\hat{\theta} - \theta^*\|^2 + \frac{1}{2}H(X_i)^2 \|\hat\theta - \theta^*\|^4.
    \end{equation}
    Averaging over $i$ yields
    \begin{equation}
    \label{eq:avg_z_diff}
        \frac{1}{n}\sum_{i=1}^n \|Z_i - Z_i(\theta^*) \|^2 \le 2\|\hat\theta - \theta^*\|^2 \frac{1}{n}\sum_{i=1}^n \| \nabla_\theta Z_i(\theta^*)\|^2 + \frac{1}{2} \|\hat\theta - \theta^*\|^4 \frac{1}{n} \sum_{i=1}^n H(X_i)^2.
    \end{equation}
    By \cref{cond2}(d), since $\| \nabla_\theta Z_i(\theta^*)\|^2 = \sum_{k=1}^K \| \nabla_{\theta_k} Z_{ik}(\theta^*_k)\|^2 $  and $K$ is fixed, we have $\frac{1}{n}\sum_{i=1}^n \| \nabla_\theta Z_i(\theta^*)\|^2 = O_p(1)$. By \cref{cond2}(e), $\|\hat\theta - \theta^*\|^2 = \sum_{k=1}^K \|\hat{\theta}_k - \theta^*_k\|^2 = O_p(1/n)$; thus, the first term of \cref{eq:avg_z_diff} is $O_p(1/n)$. Also, since $\E\big(H(X)\big) < \infty$, we have $\frac{1}{n}\sum_{i=1}^n H(X_i) = O_p(1)$ and $\max_{1\le i \le n} H(X_i) = o_p(n)$, since for any $\epsilon >0$, 
    \begin{equation}
        P\big(\max_{1 \le i \le n} H(X_i) > \epsilon n) \le n P\big(H(X) > \epsilon n\big) \le \frac{\E\big(H(X)\mathds{1}(H(X)>\epsilon n)\big)}{\epsilon}  \xrightarrow[n \to \infty]{} 0.
    \end{equation}
    Thus,
    \begin{equation}
        \frac{1}{n} \sum_{i=1}^n H(X_i)^2 \le \Big(\max_{1 \le i \le n}H(X_i)\Big) \Big(\frac{1}{n}\sum_{i=1}^n H(X_i) \Big) = o_p(n). 
    \end{equation}
    Combining with $\|\hat\theta - \theta^*\|^4 = O_p(1/n^2)$ shows the second term of \cref{eq:avg_z_diff} is $o_p(1/n)$. Hence, combining \cref{eq:avg_y_diff,eq:avg_z_diff},
    there exists a random variable $U_n$ such that (i) $U_n = O_p(1/n)$ and (ii) $\frac{1}{n}\sum_{i=1}^n \|Y_i - Y_i(\theta^*)\|^2 \leq U_n$ on $A_n$.
    
    For the second factor of \cref{eq:Sn_diff_norm}, note that
    \begin{equation}
    \label{eq:Sn_norm_sum}
        \frac{1}{n}\sum_{i=1}^n \big(\|Y_i\| + \|Y_i(\theta^*)\| \big)^2 \le \frac{2}{n} \sum_{i=1}^n \|Y_i\|^2 + \frac{2}{n}\sum_{i=1}^n \|Y_i(\theta^*)\|^2.
    \end{equation}
    The second term in \cref{eq:Sn_norm_sum} is $O_p(1)$ since $S_n(\theta^*) \xrightarrow[]{\mathrm{a.s.}} \Sigma^0$ implies $\frac{1}{n}\sum_{i=1}^n \|Y_i(\theta^*)\|^2 = \mathrm{tr}\big(S_n(\theta^*)\big) = O_p(1)$.
    Meanwhile, for the first term in \cref{eq:Sn_norm_sum},
    \begin{equation}
        \frac{1}{n}\sum_{i=1}^n \|Y_i\|^2 \le \frac{2}{n}\sum_{i=1}^n \|Y_i(\theta^*)\|^2 + \frac{2}{n}\sum_{i=1}^n \|Y_i - Y_i(\theta^*)\|^2
    \end{equation}
    and we just showed that $\frac{1}{n}\sum_{i=1}^n \|Y_i - Y_i(\theta^*)\|^2 \leq U_n$ on $A_n$, where $U_n = O_p(1/n)$.  Hence, there exists $V_n$ such that (i) $V_n = O_p(1)$ and (ii) $\frac{1}{n}\sum_{i=1}^n \big(\|Y_i\| + \|Y_i(\theta^*)\| \big)^2 \leq V_n$ on $A_n$. 

    From the bounds above, the right-hand side of \cref{eq:Sn_diff_norm} is upper bounded by $U_n^{1/2} V_n^{1/2}$ on $A_n$, and $U_n^{1/2} V_n^{1/2} = O_p(1/\sqrt{n})\, O_p(1) = o_p(1)$.
    Since $P(A_n) \to 1$ by \cref{cond2}(e), \cref{supp:lemma:high-prob-bound} implies $\|S_n - S_n(\theta^*)\| = o_p(1)$.
    Therefore, $\|S_n - S_n(\theta^*)\| \xrightarrow[]{\mathrm p} 0$. Along with $S_n(\theta^*) \xrightarrow[]{\mathrm{a.s.}} \Sigma^0$, this shows that $S_n \xrightarrow[]{\mathrm{p}} \Sigma^0$ as $n\to\infty$.

\end{proof}

\cref{supp:theorem:niw_clt_b,supp:theorem:niw_clt_c} pertain to the marginal distributions of $\mu'$ and $\Sigma'$, integrating out the data; this is sometimes referred to as the data-averaged posterior. 
In other words, we consider the marginal distributions of $\mu'$ and $\Sigma'$ when $Z_1,\ldots,Z_n$ are distributed according to the true data distribution and $(\mu', \Sigma') \sim p(\mu, \Sigma \mid Z_{1:n})$ according to the Bayesian model.
See the beginning of \cref{sec:posterior-asymptotics} for a summary of the notation.

\begin{theorem}\label{supp:theorem:niw_clt_b}
    Assume Conditions \ref{cond1} and \ref{cond2} hold, and $(\mu', \Sigma') \sim p(\mu, \Sigma \mid Z_{1:n})$ as in \cref{eq:NIW-posterior}. Then 
    $$
        \Sigma' \xrightarrow[n \to \infty]{\mathrm{p}} \Sigma^0.
    $$
\end{theorem}

\begin{proof}[Proof of \cref{supp:theorem:niw_clt_b}]\label{supp:theorem:niw_clt_b:proof}
    
    By \cref{supp:theorem:Sn_consistency}, we have $S_n \xrightarrow[]{p} \Sigma^0$ as $n \to \infty$. 
    Thus, since $\Psi_0$ is fixed, $\lambda_0 / \lambda_n \to 0$, and $\oline{Z}_n = O_p(1)$ by \cref{supp:theorem:param} and the law of large numbers,
    \begin{equation}
    \label{eq:Psi_n}
        \frac{1}{n} \Psi_n = \frac{1}{n} \Psi_0 + S_n + \frac{\lambda_0}{\lambda_n} (\oline{Z}_n - \mu_0)(\oline{Z}_n -\mu_0)^\texttt{T} \xrightarrow[n\to\infty]{\mathrm{p}} \Sigma^0,
    \end{equation}
    recalling the definitions of $\Psi_0$ and $\Psi_n$ from \cref{eq:NIW-posterior}.
    Since $\nu_n = \nu_0 + n$, for $n$ large enough that $\nu_n > K+ 1$, standard results on the moment properties of the inverse-Wishart distribution \citep{press2005applied} give that
    \begin{equation}
    \label{eq:sigma'_expectation}
        \E (\Sigma' \mid \Psi_n,\nu_n) = \frac{\Psi_n}{\nu_n -K -1} = \frac{\Psi_n}{n} \, \frac{n}{\nu_0 + n -K -1} \xrightarrow[n\to\infty]{\mathrm{p}} \Sigma^0.
    \end{equation}
    For $n$ large enough that $\nu_n > K+ 3$, the variance of the $(i,j)$-th entry is
    \begin{equation}
    \label{eq:sigma'_var}
        \V\big( (\Sigma')_{ij} \mid \Psi_n, \nu_n \big) = \frac{(\nu_n - K + 1) (\Psi_n)^2_{ij} + (\nu_n - K - 1) (\Psi_n)_{ii} (\Psi_n)_{jj}}{(\nu_n - K)(\nu_n - K - 1)^2 (\nu_n - K - 3)}.   
    \end{equation}
    Recall that $\Psi_n / n \xrightarrow[]{\mathrm p} \Sigma^0$ by \cref{eq:Psi_n}. Fix any $M > \max_{a,b} \big| (\Sigma^0)_{ab} \big|$. Then,
    \begin{equation}
        P\left( \max_{a,b} \big| (\Psi_n)_{ab} \big| \le Mn \right) \xrightarrow[n \to \infty]{} 1.
    \end{equation}
    On the event $\max_{a,b} \big| (\Psi_n)_{ab} \big| \le Mn$, the numerator of \cref{eq:sigma'_var} is $O(n^3)$ and denominator is $\Omega(n^4)$, so there exists a constant $C_{ij}(M) >0$ such that, on this event,
    \begin{equation}
        \V\big( (\Sigma')_{ij} \mid \Psi_n, \nu_n \big) \le \frac{C_{ij}(M)}{n}.
    \end{equation}
    Then, for all $\epsilon>0$, by the law of total expectation and Chebyshev's inequality,
    \begin{align}
        P\Big(& \big| (\Sigma')_{ij} - \E \big( (\Sigma')_{ij} \mid \Psi_n, \nu_n \big) \big| > \epsilon \Big) \\
        &= \E\Big(P\Big( \big| (\Sigma')_{ij} - \E \big( (\Sigma')_{ij} \mid \Psi_n, \nu_n \big) \big| > \epsilon \;\Big|\; \Psi_n, \nu_n \Big) \mathds{1}\big(\max_{a,b} \big| (\Psi_n)_{ab} \big| \le Mn \big) \Big) +  \\
        &~~~~ \E\Big(P\Big( \big| (\Sigma')_{ij} - \E \big( (\Sigma')_{ij} \mid \Psi_n, \nu_n \big) \big| > \epsilon \;\Big|\; \Psi_n, \nu_n \Big) \mathds{1}\big(\max_{a,b} \big| (\Psi_n)_{ab} \big| > Mn \big) \Big) \notag\\
        &\le \frac{C_{ij}(M)}{n\epsilon^2} + P\Big( \max_{a,b} \big| (\Psi_n)_{ab} \big| > Mn\Big) \xrightarrow[n \to \infty]{} 0.
    \end{align}
    
    Now, we control the whole matrix using the max norm and a union bound:
    \begin{equation}
         P\Big( \big\| \Sigma' - \E \big(\Sigma' \mid \Psi_n, \nu_n \big) \big\|_{\max} > \epsilon \Big) \le \sum_{i,j=1}^K P\Big( \Big| (\Sigma')_{ij} - \E \big( (\Sigma')_{ij} \mid \Psi_n, \nu_n\big) \Big| > \epsilon \Big) \xrightarrow[n \to \infty]{} 0.
    \end{equation}
    Thus, 
    \begin{equation}
        \big\| \Sigma' - \E \big(\Sigma' \mid \Psi_n, \nu_n \big) \big\|_{\max} \xrightarrow[n \to \infty]{\mathrm p} 0.
    \end{equation}
    Finally, by the triangle inequality,
    \begin{equation}
        \big\|\Sigma' - \Sigma^0 \big\|_{\max}  \le  \big\| \Sigma' - \E \big( \Sigma' \mid \Psi_n, \nu_n\big) \big\|_{\max} + \big \| \E \big( \Sigma' \mid \Psi_n, \nu_n\big) - \Sigma^0 \big\|_{\max} \xrightarrow[n \to \infty]{\mathrm{p}} 0,
    \end{equation}
    since the second term is $o_p(1)$ by \cref{eq:sigma'_expectation}. Therefore,  $\Sigma' \xrightarrow[]{\mathrm{p}} \Sigma^0$, proving the result.
\end{proof}

\cref  {supp:theorem:niw_clt_c} establishes the limiting behavior of $\mu'$ under the data-averaged LaD posterior, that is, integrating out the data. The factor of $2$ in the covariance $2\Sigma^0$ reflects the combined contribution of (i) the sampling variability of the data and (ii) the posterior uncertainty given the data.

\begin{theorem}\label{supp:theorem:niw_clt_c}
    Assume Conditions \ref{cond1} and \ref{cond2} hold, and $(\mu', \Sigma') \sim p(\mu, \Sigma \mid Z_{1:n})$ as in \cref{eq:NIW-posterior}. Then
    $$
        \displaystyle\sqrt{n} (\mu' - \mu^0) \xrightarrow[n \to \infty]{\mathrm{d}} \cN(0, 2\Sigma^0).
    $$
\end{theorem}

\begin{proof}[Proof of \cref{supp:theorem:niw_clt_c}]\label{supp:theorem:niw_clt_c:proof}
    Write
    \begin{equation}\label{eq:post_data_decomp}
       \sqrt{n} (\mu' - \mu^0) = \sqrt{n} (\mu' - \mu_n) + \sqrt{n} (\mu_n  - \mu^0). 
    \end{equation}
    By \cref{supp:theorem:niw_clt_b}, we have $\Sigma' \xrightarrow[]{\mathrm{p}} \Sigma^0$. Since the matrix square root and inverse are continuous on the space of positive definite matrices, the continuous mapping theorem yields $(\Sigma')^{-1/2} \xrightarrow[]{\mathrm{p}} (\Sigma^0)^{-1/2}$.
    Define the standardized quantities:
    \[
        V_n := \sqrt{n} (\Sigma')^{-1/2} \big( \mu' - \mu_n  \big), 
        \qquad
        W_n := \sqrt{n} (\Sigma')^{-1/2} \big( \mu_n - \mu^0 \big).
    \]
    We characterize the limiting distributions of $V_n$ and $W_n$. 
    For $V_n$, since $\mu' \mid \Sigma',Z_{1:n} \sim \cN \big(\mu_n , \Sigma' / \lambda_n \big)$, we have 
    \begin{equation}\label{eq:post_decomp_cond}
        V_n \mid \Sigma',Z_{1:n} \sim \cN(0, nI/\lambda_n),
    \end{equation}
    Thus the conditional distribution $V_n$ does not depend on $\Sigma'$ or $Z_{1:n}$, and hence $V_n \sim \cN(0, nI/\lambda_n)$ marginally for all $n$. Since $n / \lambda_n \to 1$, it follows from Slutsky's theorem that $V_n \xrightarrow[]{\mathrm{d}} \cN(0, I)$. 
    For $W_n$, from \cref{supp:theorem:niw_clt_a} and $(\Sigma')^{-1/2} \xrightarrow[]{\mathrm{p}} (\Sigma^0)^{-1/2}$, Slutsky's theorem yields $W_n \xrightarrow[]{\mathrm{d}} \cN(0, I)$. 
    Moreover, since we can represent $\mu' = \mu_n + \Sigma'^{1/2}\xi /\sqrt{\lambda_n}$ with $\xi\sim\mathcal N(0,I)$ independent of $(Z_{1:n},\Sigma')$, it follows that $V_n = \xi\sqrt{n/\lambda_n}$ depends only on $\xi$, whereas $W_n$ depends only on $(Z_{1:n},\Sigma')$. Thus, $V_n$ and $W_n$ are independent. Therefore, $V_n + W_n \xrightarrow[]{\mathrm{d}} \cN(0, 2 I)$. Finally, Slutsky’s theorem yields
    \begin{equation}
        \sqrt{n} \big( \mu' - \mu^0 \big) = (\Sigma')^{1/2} (V_n + W_n) \xrightarrow[n \to \infty]{\mathrm{d}} \cN(0, 2\Sigma^0),
    \end{equation}
    which proves the result.
\end{proof}


\subsection{Posterior consistency of hard min for the minimal KL}

Recall from \cref{eq:mu_cov} that $\mu^0 = (\mu^0_1, \ldots, \mu^0_K)^\texttt{T}$ where $\mu^0_k = \E\big(\ell_k(X; \theta^*_k)\big)$.


\begin{corollary}
\label{supp:corollary:En}
    Assume Conditions \ref{cond1} and \ref{cond2} hold, and $(\mu',\Sigma') \sim p(\mu, \Sigma \mid Z_{1:n})$ as in \cref{eq:NIW-posterior}. Then
    $$ \mu' \xrightarrow[n\to\infty]{\mathrm{p}} \mu^0.$$ 
\end{corollary}


\begin{proof}[Proof of \cref{supp:corollary:En}]
    By Lemma \ref{supp:theorem:niw_clt_c}, we have $\sqrt{n}(\mu' - \mu^0) \xrightarrow[]{\mathrm{d}} \cN(0, 2\Sigma^0)$ as $n \to \infty$. Consequently, $\mu' - \mu^0 = O_p(1/\sqrt{n})$, so $\mu' \xrightarrow[]{\mathrm{p}} \mu^0$ in $\bbR^K$. 
\end{proof}

Recall from \cref{eq:M} that we define $M_\delta(\mu) = \big\{ k : \mu_k  \leq \min_j \mu_j + \delta \big\}$ for $\delta \geq 0$. 

\begin{theorem}\label{supp:theorem:min_correct}
     Assume Conditions \ref{cond1} and \ref{cond2} hold, and $(\mu',\Sigma') \sim p(\mu, \Sigma \mid Z_{1:n})$ as in \cref{eq:NIW-posterior}. 
     Then
    \[
        P \Big( \argmin_{1 \le k \le K} \, \mu'_k \in M_0(\mu^0) \Big) \xrightarrow[n\to\infty]{} 1.
    \]
\end{theorem}

\begin{proof}[Proof of \cref{supp:theorem:min_correct}]
    If $M_0(\mu^0) = \{1, \ldots, K\}$, then the result is immediate. Suppose $M_0(\mu^0) \neq \{1, \ldots, K\}$.
    Define $m := \min\big\{ \mu^0_k - \mu^0_{\mathrm{min}}: k \notin M_0(\mu^0) \big\}$ where $\mu^0_{\mathrm{min}} = \min_k \mu^0_k$, and note that $m > 0$. Choose any $\eta \in (0, m/2)$. By \cref{supp:corollary:En}, $P\big( E_n(\eta) \big) \xrightarrow[]{} 1$ where $E_n(\eta) = \big\{ \max_{1\le k\le K} |\mu'_k - \mu^0_k| < \eta \big\}$. On the event $E_n(\eta)$, for all $k \in M_0(\mu^0)$, we have 
    \begin{equation}
        \mu'_k < \mu_k^0 + \eta = \mu^0_{\mathrm{min}} + \eta.
    \end{equation}
    Meanwhile, for all $k \notin M_0(\mu^0)$, we have $\mu_k^0 - \mu^0_{\mathrm{min}} \ge m >0$, and so since $\eta < m/2$,
    \begin{equation}
        \mu'_k > \mu^0_k - \eta \ge \mu^0_{\mathrm{min}} + m - \eta > \mu^0_{\mathrm{min}} + \eta.
    \end{equation}
    Consequently, on event $E_n(\eta)$, we have $\max_{k \in M_0(\mu^0)} \,\mu'_k < \min_{k \notin M_0(\mu^0)}\, \mu'_k$, and hence, 
    $\argmin_{1 \le k \le K} \mu'_k \in M_0(\mu^0)$.  Therefore,
    \begin{equation}
        P \Big( \argmin_{1 \le k \le K} \mu'_k \in M_0(\mu^0) \Big) \ge P\big(E_n(\eta) \big) \longrightarrow 1.
    \end{equation}
\end{proof}

\subsection{Posterior consistency for the complexity class}

Recall from \cref{eq:min_c} that $c_\delta^*(\mu) = \min\big\{ c(k): k \in M_\delta(\mu) \big\}$.


\begin{lemma}\label{supp:lemma:class_correct}
   The function $\mu\mapsto c_\delta^*(\mu)$ is continuous at any point $\mu\in\bbR^K$ such that $\delta \neq \mu_k - \min_j\mu_j$ for all $k$. 
\end{lemma}



\begin{proof}[Proof of \cref{supp:lemma:class_correct}]
Since $c_\delta^*(\mu)$ is integer valued, the claim is that $\mu\mapsto c_\delta^*(\mu)$ is continuous at any point $\mu\in\bbR^K$ such that $\delta \neq \mu_k - \min_j\mu_j$ for all $k$. 
To show this, define $g:\bbR^K \to \bbR^K$ by $g(\mu) = \mu - \min_k \mu_k$
and define $f:\bbR^K \to 2^{\{1,\ldots,K\}}$ by $f(x) = \{k : x_k \leq \delta\}$.  
We give $\bbR^K$ the usual Euclidean topology and give $2^{\{1,\ldots,K\}}$ the discrete topology.
Observe that $f(g(\mu)) = M_\delta(\mu)$, so if $f\circ g$ is continuous at some point $\mu$, then $c_\delta^*(\mu)$ is also continuous at $\mu$, because it depends on $\mu$ only through the discrete quantity $M_\delta(\mu)$.

Note that $g$ is continuous everywhere, since $\mu \mapsto \min_k \mu_k$ is continuous.
Next, note that $f$ is continuous at any point $x$ such that $x_k \neq \delta$ for all $k$.
Now, let $\mu\in\bbR^K$ such that $\delta \neq \mu_k - \min_j\mu_j$ for all $k$. 
Then, letting $x = g(\mu)$, we have $x_k\neq \delta$ for all $k$, so $f$ is continuous at $x$. 
Therefore, $f\circ g$ is continuous at $\mu$.
It follows that $c_\delta^*(\mu)$ is continuous at $\mu$, proving the claim.
\end{proof}

\cref{supp:theorem:class_consistent} shows that the data-averaged posterior is consistent for the minimal $\delta$-optimal complexity class.
More precisely, with probability tending to $1$, a sample of $c_\delta^*(\mu')$ from the data-averaged posterior will be equal the true value $c_\delta^*(\mu^0)$.

\begin{theorem}\label{supp:theorem:class_consistent}
    Assume Conditions \ref{cond1} and \ref{cond2}, and let $(\mu', \Sigma') \sim p(\mu, \Sigma \mid Z_{1:n})$ as in \cref{eq:NIW-posterior}. Suppose $\delta \neq \mu^0_k - \min_j \mu^0_j$ for all $ k = 1,\ldots,K$.  Then 
    \[
        P\big( c_\delta^*(\mu') = c_\delta^*(\mu^0) \big) \xrightarrow[n\to\infty]{} 1.
    \]
\end{theorem}

\begin{proof}[Proof of \cref{supp:theorem:class_consistent}]
By \cref{supp:lemma:class_correct}, $\mu\mapsto c_\delta^*(\mu)$ is continuous at $\mu^0$. By \cref{supp:corollary:En}, $\mu' \stackrel{\mathrm{p}}{\to} \mu^0$. 
Therefore, by the continuous mapping theorem, $c_\delta^*(\mu') \stackrel{\mathrm{p}}{\to} c_\delta^*(\mu^0)$.  Since $c_\delta^*(\mu)$ is integer-valued, this implies $P\big( c_\delta^*(\mu') = c_\delta^*(\mu^0) \big) \longrightarrow 1$.
\end{proof}

\cref{supp:theorem:class_consistent_cond} shows that given $Z_{1:n}$, the posterior is consistent for the minimal $\delta$-optimal complexity class.
This establishes consistency of the posterior, whereas \cref{supp:theorem:class_consistent} shows consistency of the data-averaged posterior.


\begin{theorem}\label{supp:theorem:class_consistent_cond}
    Under the same assumptions as \cref{supp:theorem:class_consistent},
    $$
        P\big( c_\delta^*(\mu') = c^*_\delta(\mu^0) \, \big\vert \, Z_{1:n} \big) \xrightarrow[n \to \infty]{\mathrm{p}} 1.
    $$
\end{theorem}



\begin{proof}[Proof of \cref{supp:theorem:class_consistent_cond}]
Define the event $A_n = \{c_\delta^*(\mu') \neq c_\delta^*(\mu^0)\}$.
By \cref{supp:theorem:class_consistent}, $P(A_n) \to 0$ as $n \to \infty$.
Consider the random variable $P(A_n \mid Z_{1:n})$, which is the posterior probability of $A_n$ given the data. By Markov's inequality, for all $\epsilon > 0$,
\begin{equation}
    P\big( P(A_n \mid Z_{1:n}) > \epsilon \big) \le \frac{\E \big( P(A_n \mid Z_{1:n}) \big)}{\epsilon} = \frac{P(A_n)}{\epsilon} \longrightarrow 0.
\end{equation}
Therefore, $P(A_n \mid Z_{1:n}) \xrightarrow[]{\mathrm{p}} 0$, and hence,
$$ P\big( c_\delta^*(\mu') = c^*_\delta(\mu^0) \, \big\vert \, Z_{1:n} \big) = 1 - P(A_n \mid Z_{1:n}) \xrightarrow[n\to\infty]{\mathrm{p}} 1. $$
\end{proof}

\subsection{Asymptotics of the soft selection score}


\cref{supp:lemma:alpha_choice} establishes the asymptotics of the soft selection score in \cref{eq:softmin}.  First, we show a simple lemma that is used in the proof.

\begin{lemma}
\label{supp:lemma:alpha_choice_helper}
If $X_n = O_p(1)$ and $Y_n \stackrel{\mathrm{p}}{\to} \infty$ then $X_n + Y_n  \stackrel{\mathrm{p}}{\to} \infty$ and $\min\{X_n,Y_n\} = O_p(1)$.
\end{lemma}
\begin{proof}
Let $M \in \bbR$ and $\epsilon > 0$.  Choose $B>0$ such that $P(|X_n| > B) < \epsilon$ for all $n$ sufficiently large.  Since $P(X_n + Y_n < M) \leq P(|X_n| > B) + P(Y_n < M + B)$, we have
$$ \limsup_n P(X_n + Y_n < M) \leq \limsup_n P(|X_n| > B) + \limsup_n P(Y_n < M + B) \leq \epsilon $$
because $Y_n\stackrel{\mathrm{p}}{\to} \infty$. Since $M\in\bbR$ and $\epsilon>0$ are arbitrary, this shows that $X_n + Y_n \stackrel{\mathrm{p}}{\to} \infty$. 
Furthermore, since $P(|\min\{X_n,Y_n\}| > B) \leq P(|X_n| > B) + P(Y_n < -B)$,
$$ \limsup_n P(|\min\{X_n,Y_n\}| > B) \leq \limsup_n P(|X_n| > B) + \limsup_n P(Y_n < -B) \leq \epsilon. $$
This shows that $\min\{X_n,Y_n\} = O_p(1)$.
\end{proof}



\begin{lemma}\label{supp:lemma:alpha_choice}
Suppose $Y_1,Y_2,\ldots\in\bbR^K$ are random vectors such that $\sqrt{n}(Y_n - \mu) \xrightarrow[n \to \infty]{\mathrm{d}} \cN(0, \Sigma)$ for some $\mu \in \mathbb{R}^K$ and a positive definite matrix $\Sigma \in \mathbb{R}^{K \times K}$. 
Suppose $\alpha_n\to\infty$ such that $\alpha_n/\sqrt{n} \to 0$ as $n\to\infty$. 
Then
\begin{align}
\label{eq:theorem:alpha_choice}
\exp\big(-\alpha_n(Y_{n k} - Y_n^*)\big) \xrightarrow[n\to\infty]{\mathrm{p}} \mathds{1}(\mu_k = \mu^*)
\end{align}
for all $k = 1,\ldots,K$, where $Y_n^* = \min_k Y_{n k}$ and $\mu^* = \min_k \mu_k$.
\end{lemma}


\begin{proof}[Proof of \cref{supp:lemma:alpha_choice}]
For each $k$, decompose
\begin{align}
\label{eq:alpha_choice_decomp}
\alpha_n(Y_{n k} - Y_n^*) = \alpha_n(Y_{n k} - \mu_k) - \alpha_n(Y_n^* - \mu^*) + \alpha_n(\mu_k - \mu^*).
\end{align}
Let $S = \{k : \mu_k = \mu^*\}$.
For the third term in \cref{eq:alpha_choice_decomp}, since $\alpha_n\to\infty$, we have 
\begin{align}
\label{eq:alpha_choice_decomp3}
\alpha_n(\mu_k - \mu^*) \longrightarrow \branch{0}{k \in S}{\infty}{k \not\in S}
\end{align}
as $n\to\infty$. For the first term in \cref{eq:alpha_choice_decomp}, we have
\begin{align}
\label{eq:alpha_choice_decomp1}
\alpha_n(Y_{n k} - \mu_k) = \frac{\alpha_n}{\sqrt{n}} \sqrt{n}(Y_{n k} - \mu_k) \xrightarrow[n\to\infty]{\mathrm{p}} 0
\end{align}
since $\alpha_n/\sqrt{n}\to 0$ and $\sqrt{n}(Y_{n k} - \mu_k) = O_p(1)$, because $\sqrt{n}(Y_n - \mu)$ converges in distribution by assumption.
To study the second term, we write 
\begin{align}
\label{eq:alpha_choice_min}
    \sqrt{n}(Y_n^* - \mu^*) = \min\Big\{\min_{k\in S} \sqrt{n}(Y_{n k} - \mu^*), \; \min_{k\in S^c} \sqrt{n}(Y_{n k} -\mu^*)\Big\}.
\end{align}
Letting $W\sim\cN(0,\Sigma)$, the first argument of \cref{eq:alpha_choice_min} is
\begin{align}
\label{eq:alpha_choice_limit1}
\min_{k\in S} \sqrt{n}(Y_{n k} - \mu^*) \stackrel{\mathrm{(a)}}{=} \min_{k\in S} \sqrt{n}(Y_{n k} - \mu_k) \xrightarrow[n\to\infty]{\mathrm{d}} \min_{k\in S} W_k = O_p(1)
\end{align}
where (a) holds since $\mu_k = \mu^*$ for $k\in S$, and the limit holds by the continuous mapping theorem since $x \mapsto \min_{k\in S} x_k$ is a continuous function and $\sqrt{n}(Y_n - \mu) \stackrel{\mathrm{d}}{\to} W$ by assumption.
Meanwhile, the second argument of \cref{eq:alpha_choice_min} is
\begin{align*}
\min_{k\in S^c} \sqrt{n}(Y_{n k} - \mu^*) &= \min_{k\in S^c} \Big(\sqrt{n}(Y_{n k} - \mu_k) + \sqrt{n}(\mu_k - \mu^*)\Big) \\
&\geq \min_{k\in S^c} \sqrt{n}(Y_{n k} - \mu_k) + \min_{k\in S^c}\sqrt{n}(\mu_k - \mu^*).
\end{align*}
As in \cref{eq:alpha_choice_limit1}, $\min_{k\in S^c} \sqrt{n}(Y_{n k} - \mu_k) \stackrel{\mathrm{d}}{\to} \min_{k\in S^c} W_k$.
Meanwhile, the second term satisfies $\min_{k\in S^c}\sqrt{n}(\mu_k - \mu^*) \to \infty$ since $\min_{k\in S^c} \mu_k - \mu^* > 0$ by the definition of $S$ and $\mu^*$. It follows that 
\begin{align}
\label{eq:alpha_choice_limit2}
\min_{k\in S^c} \sqrt{n}(Y_{n k} - \mu^*) \xrightarrow[n\to\infty]{\mathrm{p}} \infty,
\end{align}
by \cref{supp:lemma:alpha_choice_helper}.
Combining \cref{eq:alpha_choice_limit1,eq:alpha_choice_limit2} with \cref{eq:alpha_choice_min}, we have that
$\sqrt{n}(Y_n^* - \mu^*) = O_p(1)$ by \cref{supp:lemma:alpha_choice_helper}.  Therefore,
\begin{align}
\label{eq:alpha_choice_decomp2}
\alpha_n(Y_n^* - \mu^*) = \frac{\alpha_n}{\sqrt{n}}\sqrt{n}(Y_n^* - \mu^*) \xrightarrow[n\to\infty]{\mathrm{p}} 0,
\end{align}
since $\alpha_n/\sqrt{n}\to 0$.
Plugging \cref{eq:alpha_choice_decomp3,eq:alpha_choice_decomp1,eq:alpha_choice_decomp2} into \cref{eq:alpha_choice_decomp}, we have
\begin{align}
\label{eq:alpha_choice_overall}
\alpha_n(Y_{n k} - Y_n^*) \xrightarrow[n\to\infty]{\mathrm{p}} \branch{0}{k \in S}{\infty}{k \not\in S.}
\end{align}
For $k\in S$, we obtain \cref{eq:theorem:alpha_choice} by the continuous mapping theorem since $x\mapsto \exp(-x)$ is continuous.  For $k\in S^c$, \cref{eq:alpha_choice_overall} implies \cref{eq:theorem:alpha_choice} since for all $\epsilon > 0$,
$$ P\Big(\exp\big(-\alpha_n(Y_{n k} - Y_n^*)\big) > \epsilon\Big) =  P\Big(\alpha_n(Y_{n k} - Y_n^*) < -\log\epsilon\Big) \xrightarrow[n\to\infty]{} 0. $$
This completes the proof.
\end{proof}

\subsection{Posterior consistency of the SLC score}

Before proving our main consistency result, \cref{sec:theory:thm:consistency_score}, we establish the following lemma.  
\cref{supp:lemma:soft_consistent} shows that, under the data-averaged posterior, the within-class soft-selection score (\cref{eq:softmin}) concentrates on models with the smallest KL among the set of minimal-complexity $\delta$-optimal models.
Recall the definitions of $c_\delta^*(\mu)$, $r_{n k} (\mu)$, and $\mu^0$ from \cref{eq:min_c,eq:softmin,eq:mu_cov}, respectively. Also recall that $\mu_{\min,c(k)} = \min_{j \,:\, c(j) = c(k)} \mu_j$.

\begin{lemma}\label{supp:lemma:soft_consistent}
    Assume Conditions \ref{cond1} and \ref{cond2}, and $(\mu', \Sigma') \sim p(\mu, \Sigma \mid Z_{1:n})$ as in \cref{eq:NIW-posterior}. Suppose $\alpha_n \to \infty$ with $\alpha_n = o(\sqrt{n})$.  Then, for all $k$ such that $c(k) = c_\delta^*(\mu^0)$,
    \begin{align}
    \label{eq:lemma:soft_consistent}
        r_{n k}(\mu') = \exp\big( - \alpha_n (\mu'_k - \mu'_{\min,c(k)}) \big) \xrightarrow[n \to \infty]{\mathrm{p}} \mathds{1}\big(\mu^0_k = \mu^0_{\min,c(k)}\big).
    \end{align}
\end{lemma}


This result pertains to the data-averaged posterior, that is, the marginal distribution of $\mu'$ integrating out data $Z_{1:n}$ from the true distribution.


\begin{proof}[Proof of \cref{supp:lemma:soft_consistent}]
    Let $C := \big\{k: \, c(k) = c_\delta^*(\mu^0) \big\}$.  By \cref{supp:theorem:niw_clt_c}, we have $\sqrt{n}(\mu' - \mu^0) \xrightarrow[]{d} \cN(0, 2\Sigma^0)$ as $n \to \infty$.  By the continuous mapping theorem, this convergence holds when restricted to the coordinates in $C$, that is, $\sqrt{n}(\mu'_C - \mu^0_C) \xrightarrow[]{d} \cN(0, 2\Sigma^0_{C,C})$ where  $\mu'_C = (\mu'_j : j\in C)$, $\mu^0_C = (\mu^0_j : j\in C)$, and $\Sigma^0_{C,C} = (\Sigma^0_{j k} : j,k\in C)$.
    We apply \cref{supp:lemma:alpha_choice} with $\mu'_C$, $\mu^0_C$, and $2 \Sigma^0_{C,C}$ playing the roles of $Y_n$, $\mu$, and $\Sigma$, respectively.
    This shows that for all $k\in C$,
    $$\exp\Big(-\alpha_n(\mu'_k - \min_{j\in C}\mu'_j)\Big) \xrightarrow[n\to\infty]{\mathrm{p}} \mathds{1}\Big(\mu^0_k = \min_{j\in C} \mu^0_j\Big).$$
The result in \cref{eq:lemma:soft_consistent} follows since $k\in C$ implies $c(k) = c_\delta^*(\mu^0)$, and therefore, $\mu'_{\min,c(k)} = \min_{j\in C} \mu'_j$ and $\mu^0_{\min,c(k)} = \min_{j\in C} \mu^0_j$.
\end{proof}


\begin{lemma}
\label{lemma:conditional_convergence}
If $X_n$ and $Y_n$ are random variables such that $X_n$ is bounded and $X_n \stackrel{\mathrm{p}}{\to} c$ for some constant $c\in\bbR$, then $\E(X_n \mid Y_n) \stackrel{\mathrm{p}}{\to} c$.
\end{lemma}
\begin{proof}
These assumptions imply $X_n\to c$ in $L^1$, that is, $\E|X_n - c| \to 0$ \citep[Theorem 17.4]{jacod2004probability}. Thus, 
$$ \big|\E(X_n \mid Y_n) - c\big| = \big|\E(X_n - c \mid Y_n)\big| \leq \E\big(|X_n - c| \,\big\vert\, Y_n\big) $$
almost surely. Taking expectations, this yields 
$$ \E\Big(\big|\E(X_n \mid Y_n) - c\big|\Big) \leq \E|X_n - c| \to 0 $$
by the law of total expectation. Hence, $\E(X_n \mid Y_n)\to c$ in $L^1$, and thus $\E(X_n \mid Y_n)\stackrel{\mathrm{p}}{\to} c$ \citep[Theorem 17.2]{jacod2004probability}.
\end{proof}


\begin{proof}[\bf Proof of \cref{sec:theory:thm:consistency_score}]
    Recall from \cref{eq:pkn2} that we define the SLC score as
    \begin{align}
    \label{eq:thm:consistency_score:w}
        w_{\delta}(k \mid Z_{1:n}) = P\big(c_\delta^*(\mu') = c(k) \mid Z_{1:n} \big) \, \E\big( r_{nk}(\mu') \mid Z_{1:n} \big)
    \end{align}
    where $r_{n k}(\mu) = \exp(-\alpha_n (\mu_k - \mu_{\min,c(k)}))$ and $\mu_{\min,c(k)} = \min_{j: c(j)=c(k)} \mu_j$ for $\mu\in\bbR^K$.
    By assumption, $\alpha_n \to \infty$ with $\alpha_n = o(\sqrt{n})$, and $(\mu',\Sigma') \sim p(\mu, \Sigma \mid Z_{1:n})$ as in \cref{eq:NIW-posterior}. Here, we write $\mu'$ and $\Sigma'$ to emphasize that these are random variables distributed according to the posterior, rather than arbitrary values of $\mu$ and $\Sigma$.
    

    First, consider the first factor in \cref{eq:thm:consistency_score:w}. By \cref{supp:theorem:class_consistent_cond}, for all $k$,
    \begin{equation}
    \label{eq:thm:consistency_score:between}
        P\big(c_\delta^*(\mu') = c(k) \mid Z_{1:n} \big) \xrightarrow[n \to \infty]{\mathrm{p}} \mathds{1}\big( c_\delta^*(\mu^0) = c(k) \big).
    \end{equation}
    
    Now, consider the second factor in \cref{eq:thm:consistency_score:w}. For all $k$ such that $c(k) =  c_\delta^*(\mu^0)$, 
    we have $r_{n k}(\mu') \stackrel{\mathrm{p}}{\to} \mathds{1}\big(\mu^0_k = \mu^0_{\min,c(k)}\big)$ by \cref{supp:lemma:soft_consistent}, and therefore,
    \begin{align}
    \label{eq:thm:consistency_score:within}
    \E\big(r_{n k}(\mu') \,\big\vert\, Z_{1:n}\big) \xrightarrow[n\to\infty]{\mathrm{p}}  \mathds{1}\big(\mu^0_k = \mu^0_{\min,c(k)}\big)
    \end{align}
    by \cref{lemma:conditional_convergence}, since $r_{n k}(\mu') \in [0,1]$ is bounded.

    We combine the two factors as follows. 
    If $c(k) \neq c_\delta^*(\mu^0)$, then by \cref{eq:thm:consistency_score:w,eq:thm:consistency_score:between},
    \begin{align}
    \label{eq:thm:consistency_score:neq}
    0 \leq w_\delta(k\mid Z_{1:n}) \leq P\big(c_\delta^*(\mu') = c(k) \mid Z_{1:n} \big) \xrightarrow[n\to\infty]{\mathrm{p}} 0,
    \end{align}
    and hence, $w_\delta(k\mid Z_{1:n}) \stackrel{\mathrm{p}}{\to} 0$.
    Meanwhile, if $c(k) = c_\delta^*(\mu^0)$, then applying Slutsky's theorem to \cref{eq:thm:consistency_score:w} along with \cref{eq:thm:consistency_score:between,eq:thm:consistency_score:within}, we have
    \begin{align}
    \label{eq:thm:consistency_score:eq}
    w_\delta(k\mid Z_{1:n}) \xrightarrow[n\to\infty]{\mathrm{p}}  \mathds{1}\big(\mu^0_k = \mu^0_{\min,c(k)}\big).
    \end{align}
    Hence, for all $k$,   
    \begin{align}
    \label{eq:thm:consistency_score:combined}
    w_\delta(k\mid Z_{1:n}) \xrightarrow[n\to\infty]{\mathrm{p}}  \mathds{1}\big(c_\delta^*(\mu^0) = c(k) \big) \, \mathds{1}\big(\mu^0_k = \mu^0_{\min,c(k)}\big),
    \end{align}
    by combining  \cref{eq:thm:consistency_score:neq,eq:thm:consistency_score:eq}.
    Recalling the definition of $M^*_\delta(\mu^0)$, observe that    
    \begin{align}
    \label{eq:thm:consistency_score:M}
    M^*_\delta(\mu^0) = \argmin_{k \,:\, c(k) = c_{\delta}^*(\mu^0)} \mu^0_k = \Big\{k \,:\, c_\delta^*(\mu^0) = c(k), \; \mu^0_k = \mu^0_{\min,c(k)} \Big\}.
    \end{align}    
    Therefore, by \cref{eq:thm:consistency_score:combined,eq:thm:consistency_score:M}, $w_\delta(k\mid Z_{1:n}) \stackrel{\mathrm{p}}{\to} \mathds{1}(k \in M^*_\delta(\mu^0))$ as $n\to\infty$.

\end{proof}

\subsection{Posterior consistency of the approximate SLC score}


\begin{proof}[\bf Proof of \cref{sec:theory:thm:consistency}]
    Recall \cref{alg:proc} defines $\hat{w}_{\delta}(k)$ as the product of the between-class factor $\hat{p}_{\delta}(k)$ and the within-class factor $\hat{r}(k)$, where 
    \[
        \hat{p}_{\delta}(k) := \frac{1}{T} \sum_{t=1}^T \mathds{1} \big(c_\delta^*(\mu^{(t)}) = c(k) \big) \quad \text{and} \quad \hat{r}(k) := \frac{1}{T} \sum_{t=1}^T r_{n k}(\mu^{(t)}).
    \]
    
    First, consider the between-class factor $\hat{p}_{\delta}(k)$. Since $\mu' \sim p(\mu, \Sigma \mid Z_{1:n})$ as in \cref{eq:NIW-posterior}, by \cref{supp:theorem:class_consistent}, 
    \begin{equation}
        P\big( c_\delta^*(\mu') = c(k) \big) \xrightarrow[n\to\infty]{} \mathds{1}\big( c_\delta^*(\mu^0)  = c(k) \big)
    \end{equation}
    for all $k$.
    Each $\mu^{(t)}$ has the same distribution as $\mu'$, so for all $t = 1,\ldots, T$,
    \begin{equation}
        \mathds{1}\big( c_\delta^*(\mu^{(t)})  = c(k) \big) \xrightarrow[n \to \infty]{\mathrm{p}} \mathds{1}\big( c_\delta^*(\mu^0)  = c(k) \big).
    \end{equation}
    Since $T$ is fixed, averaging preserves the limit in probability, hence
    \begin{equation}
    \label{eq:approx_consistency:p}
        \hat{p}_{\delta}(k) \xrightarrow[n \to \infty]{\mathrm{p}} \mathds{1}\big( c_\delta^*(\mu^0)  = c(k) \big).
    \end{equation}
    
    Next, consider the within-class factor $\hat{r}(k)$. Fix any $k$ such that $c(k) =  c_\delta^*(\mu^0)$. For a single posterior draw $\mu'$, we have $r_{n k}(\mu') \stackrel{\mathrm{p}}{\to} \mathds{1}\big(\mu^0_k = \mu^0_{\min,c(k)}\big)$
    by \cref{supp:lemma:soft_consistent}.
    Again, each $\mu^{(t)}$ has the same distribution as $\mu'$, so for all $ t = 1,\ldots, T$,
    \begin{equation}
        r_{n k} (\mu^{(t)}) \xrightarrow[n \to \infty]{\mathrm{p}} \mathds{1}\big(\mu^0_k = \mu^0_{\min,c(k)}\big).
    \end{equation}
    Since $T$ is fixed, averaging preserves the limit in probability, so 
    \begin{equation}
    \label{eq:approx_consistency:r}
        \hat{r}(k) \xrightarrow[n \to \infty]{\mathrm{p}} \mathds{1}\big(\mu^0_k = \mu^0_{\min,c(k)}\big). 
    \end{equation}
    By Slutsky's theorem, combining \cref{eq:approx_consistency:p,eq:approx_consistency:r} yields
    $$
    \hat{w}_\delta(k) = \hat{p}_{\delta}(k)\, \hat{r}(k) \xrightarrow[n\to\infty]{\mathrm{p}}  \mathds{1}\big(c_\delta^*(\mu^0) = c(k) \big) \, \mathds{1}\big(\mu^0_k = \mu^0_{\min,c(k)}\big) = \mathds{1}(k\in M_\delta^*(\mu^0))$$
    where the last equality holds by the same argument as in the proof of \cref{sec:theory:thm:consistency_score}.

    

\end{proof}

\subsection{Consistency of the estimated target set of models}

\begin{proof}[\bf Proof of \cref{sec:theory:thm:consistency_M}]
    By \cref{sec:theory:thm:consistency_score}, for all $k$,
    \begin{equation}
        w_{\delta}(k \mid Z_{1:n}) \xrightarrow[n \to \infty]{\mathrm{p}} \mathds{1} \big(k \in M^*_\delta(\mu^0) \big).
    \end{equation} 
    Fix $\omega \in (0,1)$. The function $x \mapsto \mathds{1}(x > \omega)$ is continuous at $0$ and $1$, so by the continuous mapping theorem,
    \begin{equation}
        \mathds{1}\big(w_{\delta}(k \mid Z_{1:n}) > \omega \big) \xrightarrow[n \to \infty]{\mathrm{p}} \mathds{1} \big(k \in M^*_\delta(\mu^0)\big).
    \end{equation}
    for all $k$. Equivalently, $P(A_{n k}) \to 0$ as $n\to\infty$ where 
    \begin{equation}
        A_{nk} := \Big\{ \mathds{1}\big(w_{\delta}(k \mid Z_{1:n}) > \omega \big) \neq \mathds{1} \big(k \in M^*_\delta(\mu^0) \big) \Big\}.
    \end{equation} 
    Since $\hat{M}^*_\delta = \{k : w_{\delta}(k \mid Z_{1:n}) > \omega\}$,
    \begin{equation}
        P \big( \hat{M}^*_\delta \neq M^*_\delta (\mu^0) \big)  = P \bigg( \bigcup_{k=1}^K A_{nk} \bigg) \le \sum_{k=1}^K P(A_{nk} )\xrightarrow[n\to\infty]{} 0
    \end{equation}
    since $P(A_{n k})\to 0$ and $K$ is finite.
    Therefore, $P \big( \hat{M}^*_\delta = M^*_\delta (\mu^0) \big) \to 1$ as $n\to\infty$.
\end{proof}

\subsection{Instability of alternative approaches at ties}


\begin{proof}[\bf Proof of \cref{sec:theory:thm:post_instability}]
    Part (a). For brevity, let us denote $M := M^*_\delta(\mu^0)$. 
    Define $C = \{k: c(k) = c_\delta^* (\mu^0) \}$, and note that $M \subseteq C$. 
    Let $(\mu',\Sigma') \sim p(\mu, \Sigma \mid Z_{1:n})$ as in \cref{eq:NIW-posterior}; we use $\mu'$ and $\Sigma'$ to emphasize that these are random variables distributed according to the posterior.
    By the definition of $M^*_\delta(\mu)$, for all $k$,
    \begin{equation}
        \mathds{1}\big(k \in M_\delta^*(\mu')\big) = \mathds{1}\big(c_\delta^* (\mu') = c(k),\; \mu'_k = \mu'_{\min,c(k)} \big).
    \end{equation}
    Taking posterior expectations given $Z_{1:n}$ and factoring, 
    \begin{equation}\label{eq:cond_decomp}
        P(k\in M^*_\delta(\mu') \mid Z_{1:n}) = P(c_\delta^* (\mu') = c(k) \mid Z_{1:n}) \, P(\mu'_k = \mu'_{\min,c(k)} \mid Z_{1:n}, \,c_\delta^* (\mu') = c(k)).
    \end{equation}
    We analyze four cases: (1) $k\not\in C$, (2) $k\in C$, (3) $k\in C\setminus M$, and (4) $k\in M$.

    (Case 1: $k\not\in C$.)
    By \cref{supp:theorem:class_consistent_cond}, for all $k$, 
    \begin{equation}
    \label{eq:class_consistency}
        P(c_\delta^* (\mu') = c(k) \mid Z_{1:n}) \xrightarrow[n \to \infty]{\mathrm{p}} \mathds{1}\big( c_\delta^* (\mu^0) = c(k) \big) = \mathds{1}(k\in C).
    \end{equation}
    Hence, for all $k \notin C$, 
    $P(k\in M^*_\delta(\mu') \mid Z_{1:n}) \xrightarrow[n\to\infty]{\mathrm{p}} 0$ by \cref{eq:cond_decomp,eq:class_consistency}.

    (Case 2: $k\in C$.)
    Fix $k\in C$.
    Define the events $A := \{\mu_k' = \mu'_{\min,c(k)}\}$ and $B := \{c^*_\delta(\mu') = c(k)\}$. 
    Then \cref{eq:cond_decomp} can be written more succinctly as
    \begin{equation}\label{eq:plugin_post}
        P(k\in M^*_\delta(\mu') \mid Z_{1:n}) = P(B \mid Z_{1:n}) \, P(A \mid Z_{1:n}, B).
    \end{equation}
    By the law of total probability, 
    \begin{equation}
        P(A \mid Z_{1:n}) = P(A \mid Z_{1:n}, B) \,P(B \mid Z_{1:n}) + P(A \mid Z_{1:n}, B^c) \,P(B^c \mid Z_{1:n}), 
    \end{equation}
    so
    \begin{equation}
        P(A \mid Z_{1:n}) - P(A \mid Z_{1:n}, B) =  P(B^c \mid Z_{1:n}) \,\Big( P(A \mid Z_{1:n}, B^c) - P(A \mid Z_{1:n}, B) \Big),
    \end{equation}
    hence 
    \begin{equation}
        \Big| P(A \mid Z_{1:n}) - P(A \mid Z_{1:n}, B) \Big| \le P(B^c \mid Z_{1:n}).
    \end{equation}
    By \cref{supp:theorem:class_consistent_cond}, $P(B^c \mid Z_{1:n}) \xrightarrow[]{\mathrm{p}} 0$ since $k\in C$, and thus,
    \begin{equation}\label{eq:bound_diff}
        P(A \mid Z_{1:n}, B) = P(A \mid Z_{1:n}) + o_p(1).
    \end{equation}
    Combining \cref{eq:plugin_post,eq:bound_diff} and the fact that $P(B \mid Z_{1:n}) \xrightarrow[]{\mathrm p} 1$, we obtain
    \begin{equation}
        P(k\in M^*_\delta(\mu') \mid Z_{1:n}) = P(B\mid Z_{1:n}) P(A\mid Z_{1:n}) + o_p(1) = P(A\mid Z_{1:n}) + o_p(1),
    \end{equation}
    that is,
    \begin{equation}\label{eq:reduce-C*}
        P(k\in M^*_\delta(\mu') \mid Z_{1:n}) = P\big(\mu_k' = \mu'_{\min,c(k)} \mid Z_{1:n} \big) + o_p(1).
    \end{equation}

    Note that $\mu'_{\min,c(k)} = \min_{j \in C} \mu_j'$ since $k\in C$.
    We reexpress \cref{eq:reduce-C*} in terms of an argmin over $M$, rather than the min over $C$.
    First, we show that $\min_{j \in C} \mu_j' = \min_{j \in M} \mu_j'$.
    If $C = M$, then this is trivial; otherwise, suppose $C\setminus M \neq \varnothing$.
    Define the within-class gap
    \[
        m := \min_{j \in C \setminus M} \big(\mu_j^0 - \min_{i \in M} \mu_i^0 \big) > 0,
    \]
    which is strictly positive since $M\subseteq C$ and $C\setminus M \neq \varnothing$. For $\eta > 0$, define the event $E_n(\eta) := \big\{ \max_{1\le k\le K} |\mu'_k - \mu^0_k| < \eta \big\}$. By \cref{supp:corollary:En}, for any fixed $\eta > 0$, we have $P\big( E_n(\eta) \big) \to 1$ as $n\to\infty$. Choose $\eta \in (0, m/2)$. On $E_n(\eta)$, for all $j \in M$,
    \begin{equation}
        \mu_j' \le \mu_j^0 + \eta = \min_{i \in M} \mu_i^0 + \eta,
    \end{equation}
    while for all $j \in C\setminus M$,
    \begin{equation}
        \mu_j' \ge \mu_j^0 - \eta \ge \min_{i \in M} \mu_i^0 + m - \eta > \min_{i \in M} \mu_i^0 + \eta.
    \end{equation}
    Hence, $\min_{j\in C\setminus M} \mu'_j > \min_{j\in M} \mu'_j$, and thus, on event $E_n(\eta)$,
    \begin{align}
    \label{eq:min_equivalence}
    \min_{j \in C} \mu_j' = \min_{j \in M} \mu_j'.
    \end{align}
    
    Under the NIW posterior (\cref{sec:method:bayesinf}), the conditional distribution $\mu' \mid Z_{1:n}, \Sigma'$ is multivariate normal with positive definite covariance matrix $\Sigma'/\lambda_n$. Consequently, the posterior of $\mu' \mid Z_{1:n}$ has a density with respect to Lebesgue measure on $\bbR^K$, so for all $i\neq j$, we have $P(\mu_i' = \mu_j' \mid Z_{1:n}) =0$. 
    Hence, the minimum $\min_{j\in M}\mu'_j$ is attained at only one index, almost surely, under the posterior. Along with \cref{eq:min_equivalence}, this implies
    $$ P\big(\{\mu_k' = \min_{j \in C} \mu_j'\}\cap E_n(\eta)\,\big\vert\,Z_{1:n}) = P\big(\{\{k\} = \argmin_{j \in M} \mu_j'\}\cap E_n(\eta)\,\big\vert\,Z_{1:n}\big). $$
    For any events $A,B,E$, if $P(A\cap E) = P(B\cap E)$, then $|P(A)-P(B)|\leq P(E^c)$. Thus,
    \begin{align}
    \label{eq:bound_min_difference}
        \Big| P\big(\mu_k' = \min_{j \in C} \mu_j' \mid Z_{1:n} \big) - P\big(\{k\} = \argmin_{j\in M}\mu_j' \mid Z_{1:n} \big) \Big| 
        \leq P(E_n(\eta)^c \mid Z_{1:n}).
    \end{align}
    Since $0 \le P(E_n(\eta)^c \mid Z_{1:n}) \le 1$, Markov's inequality yields, for all $\epsilon > 0$,
    \begin{equation}
        P\Bigl( P(E_n(\eta)^c \mid Z_{1:n}) > \epsilon\Bigr) \le \frac{\E\bigl(P(E_n(\eta)^c \mid Z_{1:n})\bigr)}{\epsilon} \xrightarrow[n\to\infty]{} 0,
    \end{equation}
    so $P(E_n(\eta)^c \mid Z_{1:n}) \xrightarrow{\mathrm p} 0$.
    Combining this with \cref{eq:bound_min_difference,eq:reduce-C*}, along with the fact that $\mu'_{\min,c(k)} = \min_{j\in C}\mu'_j$ since $k\in C$,
    \begin{equation}\label{eq:prob-argmin-M*}
        P(k\in M^*_\delta(\mu') \mid Z_{1:n}) = P\Big(\{k\} = \argmin_{j \in M} \mu_j' \;\Big\vert\; Z_{1:n} \Big) + o_p(1).
    \end{equation}
    
    (Case 3: $k\in C\setminus M$.)
    If $k \in C\setminus M$, then $\{\{k\} = \arg\min_{j\in M} \mu_j'\} = \varnothing$, so in this case, \cref{eq:prob-argmin-M*} implies    
    $P(k\in M^*_\delta(\mu') \mid Z_{1:n}) \stackrel{\mathrm{p}}{\longrightarrow} 0$ as $n\to\infty$.

    (Case 4: $k\in M$.)
    Finally, we analyze the asymptotics of the argmin over $M$ in \cref{eq:prob-argmin-M*}.
    Define 
    \[
        V_n := \sqrt{n}\big(\mu' - \mu_n\big), \quad W_n := \sqrt{n}\big(\mu_n - \mu^0\big),
    \]
    where $\mu_n$ is the posterior mean defined in \cref{sec:method:bayesinf}, namely, $\mu_n = \big( \lambda_0 \mu_0 + n\oline{Z}_n \big) / \lambda_n$, where the $Z$ values are evaluated at $\hat{\theta}_k$ as in \cref{sec:posterior-asymptotics}.  
    By the form of the NIW posterior in \cref{sec:method:bayesinf}, we can represent $\mu'\mid Z_{1:n},\Sigma'$ as 
    \begin{equation}
        \mu' = \mu_n + (\Sigma')^{1/2}\xi/\sqrt{\lambda_n},
    \end{equation}
    with $\xi \sim \cN(0, I)$ independent of $(Z_{1:n}, \Sigma')$. Thus, $V_n \mid Z_{1:n}, \Sigma' \sim \cN(0, n\Sigma' / \lambda_n)$, and $W_n\mid Z_{1:n}$ is deterministic. In addition, we have $n/\lambda_n \to 1$, $\Sigma' \xrightarrow[]{\mathrm{p}} \Sigma^0$ (\cref{supp:theorem:niw_clt_b}), and $W_n \xrightarrow[]{d} \cN(0, \Sigma^0)$ (\cref{supp:theorem:niw_clt_a}). 
    Let $k \in M$ and define    
    \begin{align}
    \label{eq:Q}
    Q_{n k} := P\Big(\{k\} = \argmin_{j \in M} \mu_j' \;\Big\vert\; Z_{1:n} \Big).
    \end{align}
    For all $j \in M$, since $\mu_j^0 = \mu_k^0$,
    \begin{equation}
        \sqrt{n}(\mu_j' - \mu_k') = (V_{n j} - V_{n k}) + (W_{n j} - W_{n k}).
    \end{equation}
    Note that $\{k\} = \argmin_{j\in M} \mu'_j$ if and only if $\mu_j' - \mu_k' > 0$ for all $j\in M\setminus\{k\}$. Thus,
    \begin{equation}
        Q_{n k} = P\big(L_k V_n + L_k W_n > 0 \mid Z_{1:n} \big),
    \end{equation}
    where we define the linear difference map $L_k : \bbR^K \to \bbR^{|M|-1}$ by $L_k x := (x_j - x_k)_{j \in M \setminus \{k\}}$, and $>$ indicates componentwise inequality.
    We have
    \[
      L_k V_n \mid Z_{1:n}, \Sigma' \sim \cN\big(0,\Gamma_{n k}(\Sigma')\big),
      \qquad
      \Gamma_{n k}(\Sigma') := L_k \big(n\Sigma'/\lambda_n\big) L_k^\mathtt{T}.
    \]
    Since $n/\lambda_n \to 1$ and $\Sigma' \xrightarrow{\mathrm p} \Sigma^0$,
    \begin{equation}
        \Gamma_{n k}(\Sigma') \xrightarrow[n \to \infty]{\mathrm p} \Gamma_k := L_k \Sigma^0 L_k^\mathtt{T}.
    \end{equation}
    Let $\Phi_\Gamma(\cdot)$ denote the CDF of $\cN(0, \Gamma)$. By the symmetry of Gaussian distributions with mean zero, along with the fact that $W_n \mid Z_{1:n}$ is deterministic, 
    \begin{align}
        Q_{n k} &= P\big( -L_k V_n < L_k W_n \mid Z_{1:n} \big) \notag\\
                &= \E\big(P\big( -L_k V_n < L_k W_n \mid Z_{1:n}, \Sigma' \big)\;\big\vert\; Z_{1:n}\big) \\
                & = \E\big( \Phi_{\Gamma_{n k}(\Sigma')}(L_k W_n) \, \big| \, Z_{1:n} \big). \notag
    \end{align}
    
    Observe that since $k\in M$, the function $\pi_k(x)$ defined in \cref{eq:pi_k} can equivalently be written as 
    \begin{align*}
        \pi_k(x) = \Phi_{\Gamma_k}(L_k x)
    \end{align*}
    for $x \in \bbR^K$.
    The map $(\Gamma,x) \mapsto \Phi_\Gamma(L_k x)$ is continuous on the space of positive definite covariance matrices and vectors $x \in \bbR^{|M|-1}$, and is bounded in $[0,1]$.
    Let $W\sim\mathcal{N}(0,\Sigma^0)$.
    Since $\Gamma_{n k}(\Sigma') \xrightarrow{\mathrm p} \Gamma_k$ and $W_n \xrightarrow[]{\mathrm d} W$, we have joint convergence
    \begin{equation}
        \big( \Gamma_{n k}(\Sigma'), W_n \big) \xrightarrow{\mathrm d} \big( \Gamma_k, W \big).
    \end{equation}
    By the continuous mapping theorem,
    \begin{equation}
        D_n := \Phi_{\Gamma_{n k}(\Sigma')}\big( L_k W_n\big) - \Phi_{\Gamma_k}\big( L_k W_n\big) \xrightarrow[n \to \infty]{\mathrm d} 0.
    \end{equation}
    Since the limit is constant, $D_n \xrightarrow[]{\mathrm p} 0$, and $|D_n| \le 1$ for all $n$. Thus, $\E|D_n| \to 0$ \citep[Theorem 17.4]{jacod2004probability}.
    By Markov's inequality, for all $\epsilon >0$,
    \begin{equation}
        P\Big( \E\big(|D_n| \big| Z_{1:n}\big) > \epsilon \Big) \le \frac{\E|D_n|}{\epsilon} \xrightarrow[n\to\infty]{} 0.
    \end{equation}
    Therefore, $\big| \E(D_n \big| Z_{1:n}) \big| \le \E\big(|D_n| \big| Z_{1:n} \big) = o_p(1)$. Thus, for $k\in M$,
    \begin{align}
    \label{eq:Q_pi}
        Q_{n k} &= \E \Big( \Phi_{\Gamma_{n k}(\Sigma')}\big( L_k W_n\big) \, \Big| \, Z_{1:n} \Big) \notag\\
                &=  \Phi_{\Gamma_k}\bigl(L_k W_n\bigr) + o_p(1) \\
                &= \pi_k(W_n) + o_p(1). \notag
    \end{align}

    We now combine the cases to obtain the result.
    Since \cref{eq:prob-argmin-M*,eq:Q,eq:Q_pi} hold coordinate-wise for each $k\in M$, they also hold jointly for $k\in M$, because coordinate-wise convergence in probability implies joint convergence in probability.  Hence,
    \begin{align}
    \label{eq:PM_pi}
    \big(P(k\in M^*_\delta(\mu') \mid Z_{1:n})\big)_{k\in M} = \big(\pi_k(W_n)\big)_{k\in M} + o_p(1).
    \end{align}
    Meanwhile, by cases (1) and (3) above, for $k \notin M$ we know $P(k\in M^*_\delta(\mu') \mid Z_{1:n}) \xrightarrow[]{\mathrm{p}} 0$, 
    and by definition $\pi_k(x) = 0$ for all $x$ when $k\not\in M$ (\cref{eq:pi_k}).
    Thus, in fact, \cref{eq:PM_pi} holds jointly over all $k\in\{1,\ldots,K\}$, not just $k\in M$.
    Since $W_n \xrightarrow[]{\mathrm{d}} W \sim \cN(0, \Sigma^0)$ and the map $x \mapsto \big(\pi_k(x)\big)_{k=1}^K$ is continuous, the continuous mapping theorem implies $\big(\pi_k(W_n)\big)_{k=1}^K \xrightarrow[]{\mathrm{d}} \big(\pi_k(W)\big)_{k=1}^K$. Therefore, by Slutsky's theorem,
    \begin{equation}
            \big(P(k\in M^*_\delta(\mu') \mid Z_{1:n})\big)_{k=1}^K \xrightarrow[n \to \infty]{\mathrm{d}} \big(\pi_k(W)\big)_{k=1}^K.
    \end{equation}
    
    
    Part (b). By the definition $\tilde{r}_{n k}(\mu) = \mathds{1}(\mu_k = \mu_{\min,c(k)})$, we have
    \begin{equation}\label{eq:hardmin_rep}
        \tilde w_\delta (k \mid Z_{1:n}) = P\big(c_\delta^*(\mu') = c(k) \mid Z_{1:n} \big) P\big( \mu_k' = \mu'_{\min,c(k)} \mid Z_{1:n} \big) = P(B\mid Z_{1:n}) P(A \mid Z_{1:n}),
    \end{equation}
    where $A = \{\mu_k' = \mu'_{\min,c(k)}\}$ and $B = \{c^*_\delta(\mu') = c(k)\}$ as before. For $k\in C$, subtracting \cref{eq:plugin_post} from \cref{eq:hardmin_rep} and using \cref{eq:bound_diff}, we obtain
    \begin{equation}
        \tilde w_\delta (k \mid Z_{1:n}) - P\big(k\in M_\delta^*(\mu') \mid Z_{1:n}\big) = P(B\mid Z_{1:n}) \big( P(A\mid Z_{1:n}) - P(A\mid Z_{1:n},B) \big) = o_p(1),
    \end{equation}
    since $|P(B\mid Z_{1:n})| \leq 1$.
    Meanwhile, for $k\notin C$, \cref{eq:class_consistency} implies that both $\tilde w_\delta (k \mid Z_{1:n})$ and $P(k\in M_\delta^*(\mu') \mid Z_{1:n})$ converge to $0$ in probability, as in case (1); thus, their difference also converges to $0$ in probability. Therefore, for all $k$, regardless of whether $k\in C$,
    \begin{equation}
    \label{eq:w_tilde_diff}
        \tilde w_\delta (k \mid Z_{1:n}) - P\big(k\in M_\delta^*(\mu') \mid Z_{1:n}\big) \xrightarrow[n \to \infty]{\mathrm p} 0.
    \end{equation}
    Therefore, 
    \begin{equation}
    \label{eq:w_tilde_PM}
    \big(\tilde w_\delta (k \mid Z_{1:n})\big)_{k=1}^K = \big(P\big(k\in M_\delta^*(\mu') \mid Z_{1:n}\big)\big)_{k=1}^K + o_p(1).
    \end{equation}
    Along with the result of part (a), this implies that
    $\big(\tilde w_\delta(k\mid Z_{1:n})\big)_{k=1}^K \stackrel{\mathrm{d}}{\to} \big(\pi_k(W)\big)_{k=1}^K$ by Slutsky's theorem.
\end{proof}



\begin{proof}[\bf Proof of \cref{sec:theory:cor:instability_two2}]
    By \cref{sec:theory:thm:post_instability}(b), it suffices to show that 
    \begin{equation}
        \big(\pi_{k_1}(W), \pi_{k_2}(W) \big) \stackrel{d}{=} (U, 1- U)
    \end{equation}
    where $W \sim \cN(0, \Sigma^0)$ and $U\sim\mathrm{Uniform}(0,1)$.
    Since $M_\delta^*(\mu^0) = \{k_1, k_2\}$, the definition of $\pi_k(x)$ reduces to
    \begin{equation}
        \pi_{k_1}(x) = P (W_{k_2} - W_{k_1} \le x_{k_2} - x_{k_1}), \quad
        \pi_{k_2}(x) = P (W_{k_1} - W_{k_2} \le x_{k_1} - x_{k_2})
    \end{equation}
    for all $x \in \bbR^K$. Let $F$ be the cumulative distribution function of the Gaussian random variable $W_{k_2} - W_{k_1}$. Since $\Sigma^0$ is positive definite, $\V(W_{k_2} - W_{k_1}) >0$, so $F$ is continuous and strictly increasing. Then, $\pi_{k_1}(x) = F(x_{k_2} - x_{k_1})$ and
    \begin{equation}
        \pi_{k_2}(x) = P\big( -(W_{k_2} - W_{k_1}) \le - (x_{k_2}-x_{k_1}) \big) = 1 - F(x_{k_2} - x_{k_1}).
    \end{equation}
    Hence, for all $x$, 
    \begin{equation}\label{eq:sum_pi}
        \pi_{k_1}(x) + \pi_{k_2}(x) = 1.
    \end{equation}

    Now, we consider evaluating these functions at $x=W$.
    By the probability integral transform,
    \begin{equation}\label{eq:pi1_unif}
        \pi_{k_1}(W) = F(W_{k_2} - W_{k_1}) \sim \mathrm{Uniform}(0,1).
    \end{equation}
    By \cref{eq:sum_pi}, $\pi_{k_2}(W) = 1 - \pi_{k_1}(W)$. Combined with \cref{eq:pi1_unif}, this implies
    \begin{equation}
        \big(\pi_{k_1}(W), \pi_{k_2}(W) \big) \overset{\mathrm d}{=} (U, 1-U), \quad U \sim \mathrm{Uniform}(0,1).
    \end{equation}
    This completes the proof. 
\end{proof}

\section{Additional details on the examples}

\subsection{Gaussian mixture example details}\label{supp:mix_details}

This section provides additional details on the Gaussian mixture example in \cref{sec:examples:gmm}.

\paragraph*{Data source.} The Shapley galaxy radial velocity data are from \citet{drinkwater2004large}, available at \url{https://sites.psu.edu/astrostatistics/datasets-shapley-galaxy-dataset/}.

For a given number of mixture components $k$, we model the observations $x_1, \ldots, x_n \in \mathbb{R}$ by a $k$-component Gaussian mixture
\[
    f_k(x) = \sum_{j=1}^k \phi_{j} \,\cN \big( x \,\big\vert\,m_{j}, s_{j}^2 \big)
\]
with parameters $\theta = \big\{ ( \phi_j, m_j, s_j^2 ) \big\}_{j=1}^k$, where $\phi_1,\ldots,\phi_k \ge 0$ are mixture weights such that $\sum_{j=1}^k \phi_j = 1$, while $m_j\in\bbR$ and $s_j^2 > 0$ are the mean and variance, respectively, of the $j$th component.

\paragraph{Prior distributions.}
Following \citet{fraley2007bayesian}, we use a normal-inverse-gamma prior, which is a conjugate prior for the means and variances:
\[
    m_j \mid s_j^2 \sim \cN(m_0, s_j^2/\kappa_0), \quad s_j^2 \sim \mathrm{InverseGamma}(\nu_0/2, s_0^2/2),
\]
and we place a uniform prior on the vector of mixing proportions $\phi = (\phi_1,\ldots,\phi_k)$, that is, $\phi \sim \mathrm{Dirichlet(1,\ldots,1)}$, which is also conjugate.
In each run (for a given $n$ and $k$), we set $m_0 = \bar x$, $\kappa_0 = 0.01$, and $s_0^2 = \widehat\V(x) /k^2$, following the default prior choices in \citet{fraley2007bayesian}, and we set the degrees of freedom to $\nu_0=10$.

\paragraph{Expectation-maximization (EM) algorithm.}
We estimate $\theta = \big\{ ( \phi_j, m_j, s_j^2 ) \big\}_{j=1}^k$ via the EM algorithm, using \emph{maximum a posteriori} (MAP) updates to prevent likelihood singularities and improve stability \citep{bishop2006pattern}. 
The E-step consists of computing the responsibilities
\[
    r_{i j} = \frac{\phi_j \,\cN( x_i \mid m_{j}, s_{j}^2 )}{\sum_{l=1}^k \phi_l \,\cN(x_i \mid m_l, s_l^2)}.
\]
For the M-step, define $n_j := \sum_{i=1}^n r_{i j}$ and define the weighted mean $\bar x_j := \frac{1}{n_j} \sum_{i=1}^n r_{i j} x_i$.  The MAP M-step updates are then
\[
    \phi_j^{\mathrm{new}} = \frac{n_j}{n}, \qquad m_j^{\mathrm{new}} = \frac{\kappa_0 m_0 + n_j \bar x_j}{\kappa_0 + n_j},
\]
\[
    s_j^{2,\mathrm{new}} = \frac{\displaystyle s_0^2 + \frac{\kappa_0 n_j}{\kappa_0 + n_j} (\bar x_j - m_0)^2 + \sum_{i=1}^n r_{i j} (x_i - \bar x_j)^2 }{\nu_0 + n_j + 3}.
\]
We run this EM algorithm using the \texttt{mclust} package in R \citep{scrucca2016mclust, r2016r} with random initialization.
For each $k$, we run the algorithm on $50$ independent restarts and keep the run with the highest observed log-likelihood.

\paragraph{Bayesian posterior on the number of components.}

In this example, the model index $k$ represents the number of mixture components. 
As a comparison method, we obtain a posterior over the (unknown) number of components $k$ by placing a prior on $k$ and computing its posterior $\pi(k \mid x_{1:n})$ under a mixture of finite mixtures model (MFM; \citealp{miller2018mixture}). Specifically, we take $K \sim \mathrm{Uniform}\{1,2,\ldots,30\}$; so in particular, $P(K>30)=0$. For the mixture weights, we take $(\phi_1, \ldots, \phi_k) \mid K=k \sim \mathrm{Dirichlet(1, \ldots, 1)}$. For the component parameters, we use independent conditionally conjugate priors on the mean and variance; specifically, given $K = k$, we take
$$ m_j \sim \cN(m_0, \sigma_0^2), \quad s_j^2 \sim \mathrm{InverseGamma}(2,b),$$
independently for $j = 1,\ldots,k$, where $m_0 = (\min_i x_i + \max_i x_i)/2$ and $\sigma_0^2 =  (\max_i x_i - \min_i x_i)^2$, and we further place a $b \sim \mathrm{Gamma}\big(0.2, \, 10/\sigma_0^2 \big)$ hyperprior on $b$.

We run an MCMC algorithm implemented in the Julia package BayesianMixtures.jl (\url{https://github.com/jwmi/BayesianMixtures.jl}) targeting the MFM posterior for $200{,}000$ iterations, discard the first $100{,}000$ as burn-in, and use the remaining $100{,}000$ draws for posterior inference. 
The MCMC algorithm samples from the joint posterior over assignments of observations to components, component parameters, and the hyperparameter $b$. 
We obtain posterior draws of $k$ using the technique of \citet[Equation 3.7]{miller2018mixture}, specifically, we sample from $k|t$ for each MCMC sample of $t$, where $t$ denotes the number of occupied clusters in a posterior draw.



\subsection{Sparse MVN example details}\label{supp:mvn_details}

In this section, we provide additional details on the sparse multivariate normal (MVN) model example in \cref{sec:examples:mvn}.

\paragraph{Kullback--Leibler (KL) divergence calculation.}\label{supp:mvn_kl}
The KL divergence between two multivariate normal distributions $\cN(\mu_1, \Sigma_1)$ and $\cN(\mu_2, \Sigma_2)$ on $\bbR^k$ is
\[
    D\big(\cN(\mu_1, \Sigma_1) \,\|\, \cN(\mu_2, \Sigma_2)\big) = \frac{1}{2} \bigg( \mathrm{tr}(\Sigma_2^{-1}\Sigma_1) + (\mu_1 - \mu_2)^\texttt{T} \Sigma_2^{-1} (\mu_1 - \mu_2) - k + \log  \frac{|\det \Sigma_2|}{|\det \Sigma_1|} \bigg).
\]
In our example, $\Sigma_1 = \Sigma_2 = I$, so the KL divergence simplifies to $D(\cN(\mu_1, \Sigma_1) \,\|\, \cN(\mu_2, \Sigma_2)) = \frac{1}{2} \|\mu_1 - \mu_2\|^2$.
Thus, for a candidate MVN model $k$ whose mean is constrained to a linear subspace $\Theta_k$, the minimal KL divergence from the true model $\cN(\theta_0, I)$ is
\[
    \min_{\theta_k\in\Theta_k} D\big(\cN(\theta_0, I) \,\|\, \cN(\theta_k, I)\big) = \min_{\theta_k\in\Theta_k} \frac{1}{2} \|\theta_0 - \theta_k \|^2 = \frac{1}{2} \|\theta_0 - \theta^*_k \|^2,
\]
where $\theta^*_k := \argmin_{\theta \in \Theta_k} \|\theta_0 - \theta\|^2$, that is, $\theta^*_k$ is the Euclidean orthogonal projection of $\theta_0$ onto the parameter subspace for model $k$.



\paragraph{Coarsened posterior (c-posterior).}\label{supp:mvn_coarsening}
In \cref{sec:examples:mvn:comparison}, we compare with the c-posterior of \citet{miller2019robust}; here we provide details on this method.
For robustness to misspecification, the c-posterior is obtained by conditioning on the event that the empirical distribution generated by the model is ``close'' to that of the observed data. When closeness is measured by KL divergence (also known as relative entropy), the c-posterior can be approximated by a power posterior, obtained by raising the likelihood to a fractional power $\zeta_n$ \citep{miller2019robust}. Concretely, if $\pi_k(\theta)$ is the prior under model $k$, then the power posterior for model $k$ is
\[
    \pi_{\alpha, k} (\theta ; x_{1:n}) \propto \pi_k(\theta) \prod_{i=1}^n f_k(x_i; \theta)^{\zeta_n}, 
\]
where $\zeta_n := \alpha / (\alpha + n) \in (0,1]$ and $\alpha > 0$ is a setting that control the precision of the power posterior. In particular, the power posterior recovers the standard posterior as a special case in the limit as $\alpha \to \infty$ ($\zeta_n \to 1$).   The marginal power likelihood is defined as
\[
    m_{\alpha, k}(x_{1:n}) := \int \pi_k(\theta) \prod_{i=1}^n f_k(x_i; \theta)^{\zeta_n} d\theta,
\]
and can be used to construct a posterior on models $k$ that provides robust model inference.
For exponential families, the power likelihood has the attractive property that it preserves conjugacy.  Thus, when using conjugate priors with exponential families, the power posterior and power marginal likelihood take closed form expressions.

We specialize the c-posterior framework to the sparse MVN mean model in \cref{sec:examples:mvn}. Recall that we observe $x_1, \ldots x_n \overset{\mathrm{iid}}{\sim} \cN(\theta_0, I)$ where $\theta_0\in\bbR^6$ is the true mean, which is treated as unknown. Consider the $K=7$ models in \cref{tab:sparse_mvn_models}, with a uniform prior over models, $\pi(k)=1/K$. For each model $k = 1,\ldots,K$, let $J_k \subseteq \{1, \ldots, 6\}$ be the set of free coordinates and define $d_k = |J_k|$; coordinates in $J_k^c$ are fixed at $0$. Denote the sample mean by $\bar{x} = (\bar{x}_1, \ldots, \bar{x}_6)^\texttt{T}$. 
Since the covariance matrix is $I$, the power likelihood for each model is
\begin{align}
\label{eq:power-lik1}
     \prod_{i=1}^n f_k(x_i; \theta)^{\zeta_n} \propto \prod_{i=1}^n \exp \Big( - \frac{\zeta_n}{2} \|x_i - \theta\|^2 \Big) \propto \exp \Big( - \frac{\zeta_nn}{2} \| \bar x -  \theta\|^2 \Big).
\end{align}
Denote $\theta_{J_k} = (\theta_j : j\in J_k)$ and likewise for $\theta_{J_k^c}$, as well as $\bar{x}_{J_k} = (\bar{x}_j : j \in J_k)$.
Under model $k$, we have $\theta_{J_k^c} = 0$,   and hence, $\|\bar x - \theta \|^2 = \|\bar x_{J_k} - \theta_{J_k} \|^2 + \|\bar x_{J_k^c} \|^2$. Therefore,
\begin{align}
\label{eq:power-lik2}
    \prod_{i=1}^n f_k(x_i; \theta)^{\zeta_n} \propto \exp\Big(-\frac{\zeta_n n}{2} \| \bar x_{J_k} - \theta_{J_k} \|^2 \Big) \exp\Big(-\frac{\zeta_n n}{2} \| \bar x_{J_k^c} \|^2\Big).
\end{align}

For each model $k$, we place a Gaussian prior on the free coordinates and fix the rest to zero:
\[
    \theta_{J_k} \sim \cN\big( \theta_{0J_k}, \kappa_0^{-1} I_{d_k}\big), \qquad \theta_{J_k^c} = 0,
\]
where $\theta_{0J_k} \in \mathbb{R}^{d_k}$ is a fixed prior mean and $\kappa_0 >0$ is a fixed prior precision. Then, by conjugacy, the power posterior on $\theta_{J_k}$ under model $k$ is 
\[
    \theta_{J_k} \;\big|\; x_{1:n}, \alpha \sim \cN \bigg( \frac{\kappa_0 \theta_{0J_k} + \zeta_nn\bar{x}_{J_k}}{\kappa_0 + \zeta_nn}, \, (\kappa_0 + \zeta_nn)^{-1} I_{d_k} \bigg).
\]

Observe that the constant of proportionality does not depend on $k$ in \cref{eq:power-lik1}, and hence also in \cref{eq:power-lik2}.  Thus, writing the first factor of \cref{eq:power-lik2} in the form of a Gaussian density, we have
\begin{align}
\prod_{i=1}^n f_k(x_i; \theta)^{\zeta_n} = \frac{C}{(\zeta_n n)^{d_k/2}}\mathcal{N}(\bar x_{J k}\mid \theta_{J k}, (\zeta_n n)^{-1} I_{d_k}) \exp\Big(-\frac{\zeta_n n}{2} \| \bar x_{J_k^c} \|^2\Big)
\end{align}
where $C$ does not depend on $\theta$ or $k$.
Then, the marginal power likelihood is
\begin{align*}
m_{\alpha,k}(x_{1:n}) &= \int \pi(\theta_{J_k})\prod_{i=1}^n f_k(x_i; \theta)^{\zeta_n} d \theta_{J_k} \\
&= \frac{C\exp(-\frac{\zeta_n n}{2} \| \bar x_{J_k^c} \|^2)}{(\zeta_n n)^{d_k/2}} \int \mathcal{N}(\bar x_{J_k}\mid \theta_{J k}, (\zeta_n n)^{-1} I_{d_k})\, \mathcal{N}(\theta_{J_k} \mid \theta_{0 J_k}, \kappa_0^{-1} I_{d_k}) d\theta_{J_k} \\
&= \frac{C\exp(-\frac{\zeta_n n}{2} \| \bar x_{J_k^c} \|^2)}{(\zeta_n n)^{d_k/2}}\mathcal{N}(\bar x_{J_k} \mid \theta_{0 J_k}, ((\zeta_n n)^{-1} + \kappa_0^{-1})I_{d_k}).
\end{align*}

Therefore, the log marginal power likelihood is
\begin{align}
\label{eq:log-power-marlik}
    \log m_{\alpha, k} (x_{1:n}) = c + \frac{d_k}{2} \log\Big(\frac{\kappa_0}{\kappa_0 + \zeta_n n}\Big) - \frac{\kappa_0\zeta_nn}{2(\kappa_0 + \zeta_nn)} \|\bar x_{J_k} - \theta_{0J_k}\|^2 - \frac{\zeta_nn}{2}\|\bar x_{J_k^c}\|^2
\end{align}
where $c$ is a constant that does not depend on $k$.
Since we assume a uniform prior over models, $\pi(k) = 1/K$, the coarsened posterior on models is 
\begin{align}
\label{eq:cpost-k}
\pi_{\alpha}(k; x_{1:n}) = \frac{m_{\alpha, k} (x_{1:n})}{\sum_{j=1}^K m_{\alpha, j} (x_{1:n})}.
\end{align}

\cref{eq:log-power-marlik,eq:cpost-k} yield a  closed-form expression that enables exact computation of the c-posterior probabilities for the $K$ models, which we use in \cref{sec:examples:mvn:comparison}.

\subsection{Thermal performance curve example details}\label{supp:tpc}

This section provides additional details on the example from \cref{sec:examples:tpc}.

\paragraph*{Data source.} The \emph{Prorocentrum minimum} growth-rate data are from \citet{kontopoulos2024dat}, available at \url{https://doi.org/10.6084/m9.figshare.24106161.v3}.

\paragraph{MLEs for noise model and flexible model.} 
The MLE for the noise model is $\hat{\varphi} = \frac{1}{n}\sum_{i=1}^n y_i$ and $\hat{\sigma}_{\mathrm{noise}}^2 = \frac{1}{n} \sum_{i=1}^n (y_i-\hat{\varphi})^2$.
The MLE for the flexible model is 
$\hat{\beta}_t = \frac{1}{n_t}\sum_{i=1}^n y_i\mathds{1}(t_i = t)$ where $n_t = \sum_{i=1}^n \mathds{1}(t_i = t)$, and $\hat{\sigma}_{\mathrm{flex}}^2 = \frac{1}{n} \sum_{i=1}^n (y_i-\hat{\beta}_{t_i})^2$.

\paragraph{Functional forms of the TPC candidates.} 
For completeness, we state the functional forms of the TPC candidate models here. 
All formulas are taken from the supplementary material of \citet{kontopoulos2024no}. Temperature parameters are real-valued, and in nonlinear least squares fitting we impose the constraints $t_{pk} \in [0,150]$ and $t_{\mathrm{min}} \le t_{pk} \le t_{\mathrm{max}}$.
\begin{enumerate}
    \item Eubank (3 parameters)
    \[
        g_1(t; \, a,b, t_{pk}) = \frac{a}{(t - t_{pk})^2 + b}, \qquad a>0, \; b>0.
    \]

    \item Gaussian (3 parameters)
    \[
        g_2(t; \, a, b_{pk}, t_{pk}) = b_{pk} \exp \Big( -\frac{1}{2 a^2} (t - t_{pk})^2 \Big), \qquad a>0, \; b_{pk}>0.
    \]

    \item Mitchell--Angilletta (3 parameters)
    \[
        g_3(t; \, a, b, t_{pk}) = \frac{a}{2b}  \bigg( 1 + \cos \Big( \frac{t - t_{pk}}{b} \cdot \pi \Big) \bigg), \qquad a>0, \; b>0.
    \]

    \item Bri\`ere I (3 parameters)
    \[
        g_4(t; \, a, t_{\mathrm{min}}, t_{\mathrm{max}}) = a\, t\, (t-t_{\mathrm{min}}) \sqrt{t_{\mathrm{max}} - t}, \qquad a>0, \; t_{\mathrm{min}} < t < t_{\mathrm{max}}.
    \]

    \item Weibull (4 parameters)
    \begin{align*}
        g_5(t; \, b,c, b_{pk}, t_{pk}) &= b_{pk}  \Big(\frac{c-1}{c}\Big)^{(1-c)/c}  \bigg( \frac{t - t_{pk}}{b} + \Big( \frac{c-1}{c} \Big)^{1/c} \bigg)^{c-1}  \\
        & \quad \times \exp \bigg( - \left(\frac{t - t_{pk}}{b} + \Big(\frac{c-1}{c} \Big)^{1/c} \right)^c + \frac{c-1}{c} \bigg),
    \end{align*}
    $b>0, \; c>0, \; b_{pk}>0$.
    
    \item Modified Gaussian (4 parameters)
    \[
        g_6(t; \, a, b, b_{pk}, t_{pk}) = b_{pk}  \exp \Big( -\frac{1}{2} \Big(\frac{|t - t_{pk}|}{a} \Big)^b \Big), \qquad a>0, \; b>0, \; b_{pk} >0.
    \]

    \item Extended Bri\`ere (5 parameters)
    \[
        g_7(t; \, a, b, c, t_{\mathrm{min}}, t_{\mathrm{max}}) = a \, t \, (t-t_{\mathrm{min}})^b \, (t_{\mathrm{max}} - t)^c, 
    \]
    $a>0, \; b>0, \; c>0, \; t_{\mathrm{min}} < t < t_{\mathrm{max}}$.
    
    \item Fifth-order polynomial (Poly-5; 6 parameters)
    \[
        g_8(t; \, a, b, c, d, f, g) = a + b t + c t^2 + d t^3 + f t^4 + g t^5.
    \]

    \item (Extended) Sharpe--Schoolfield (7 parameters)
    \begin{align*}
        & g_9(T; \, a, b_0, t_{L50}, t_{H50}, \Delta H_L, \Delta H_A^{\ddagger}, \Delta H_{H}) \\ &\qquad = \frac{\displaystyle b_0 \frac{T}{T_{\mathrm{ref}}} \exp \bigg( \frac{\Delta H_A^{\ddagger} }{R} \Big( \frac{1}{T_{\mathrm{ref}}} - \frac{1}{T} \Big) \bigg) }{\displaystyle a + \exp \bigg( \frac{\Delta H_L}{R} \Big(\frac{1}{T_{L50}} - \frac{1}{T} \Big) \bigg) + \exp \bigg( \frac{\Delta H_H}{R} \Big(\frac{1}{T_{H50}} - \frac{1}{T} \Big) \bigg) }, 
    \end{align*}
    $a>0, \; b_0 >0, \; \Delta H_L >0, \; \Delta H_A^{\ddagger} >0, \; \Delta H_{H} >0$. Here, $T_{\mathrm{ref}} = 273.5$ (temperatures in Kelvin) and $R = 1.987$.
\end{enumerate}

\subsection{Population structure admixture example details}\label{supp:admix_details}

Here, we provide additional details on the admixture model example in \cref{sec:examples:population-structure}.

\paragraph*{Data source.} The brook trout microsatellite genotype data are from \citet{erdman2022broadscale}, USGS data release available at \url{https://www.sciencebase.gov/catalog/item/611d264cd34e40dd9c01284e}.

\paragraph{STRUCTURE output: Estimated Ln Prob of Data}
STRUCTURE performs posterior inference for the admixture model using MCMC, and reports an ``Estimated Ln Prob of Data'' (ELP), which is intended as a rough approximation to the log marginal likelihood $\log p(x \mid k)$. For LaD, we use this quantity as a fit diagnostic to filter out failed MCMC runs; specifically, we run the STRUCTURE algorithm 20 times with random restarts, and keep the run with the highest ELP. This quantity is also used for implementing Evanno's method, in which the ELP is used to compute $L(k)$.

As described in \citet{pritchard2000inference}, the ELP is constructed from the Bayesian deviance
\[
    \mathrm{dev}(\phi, q, s) := -2 \log p(x \mid \phi, q, s)
\]
evaluated on the MCMC output. 
Specifically, \citet{pritchard2000inference} use a normal approximation to show that
\[
    -2 \log p(x \mid k) \approx \E \big(\mathrm{dev}(\phi, q, s) \mid x,k \big) + \frac{1}{4} \V \big(\mathrm{dev}(\phi, q, s) \mid x,k \big).
\]
In practice, $\E \big(\mathrm{dev}(\phi, q, s) \mid x,k \big)$ and $\V \big(\mathrm{dev}(\phi, q, s) \mid x,k \big)$ are approximated by Monte Carlo averages over the post–burn-in MCMC draws. The resulting approximation of $\log p(x \mid k)$ is what STRUCTURE returns as ``Estimated Ln Prob of Data'' (ELP).

\paragraph{Per-individual negative log-likelihoods.}

For the LaD technique, we need to compute the per-individual negative log-likelihood values $\ell_k(x_i; \hat\phi,\hat\psi) = -\log p(x_i\mid \hat\phi,\hat\psi)$, integrating out $q_i$ and $s_{i l a}$ for all $l,a$, given the point estimate $(\hat\phi$, $\hat\psi)$ for model $k$.  We compute these using Monte Carlo approximations, as follows.
Fix a given model $k$, and let $\hat\phi$ and $\hat\psi$ denote the posterior means obtained from MCMC. First, observe that 
\[
    p(x_{i l a}=v \mid \phi, q_i) = \sum_{j=1}^k p(s_{i l a} =j \mid q_i) \, p(x_{i l a}=v \mid s_{i l a}=j,\, \phi) = \sum_{j=1}^k q_{ij} \phi_{jl}(v).
\]
We write $\phi_{jl}(v)$ instead of $\phi_{jlv}$ for visual clarity of the notation in what follows.
Thus, for the whole genotype of individual $i$,
\[
    p(x_i \mid \phi, \psi) = \int_{\Delta^{k-1}} \bigg( \prod_{l=1}^L \prod_{a=1}^2 \sum_{j=1}^k q_{ij} \phi_{jl}(x_{i l a}) \bigg) \mathrm{Dirichlet}(q_i\mid \psi) \, dq_i, 
\]
where the integral is over the simplex $\Delta^{k-1} = \{q \in \mathbb{R}_+^k: \sum_{j=1}^k q_j=1\}$ and $\mathrm{Dirichlet}(q \mid \psi) \propto \prod_{j=1}^k q_j^{\psi_j-1}$. 
For each individual $i$, draw $q_i^{(b)} \sim \mathrm{Dirichlet}(\hat\psi)$ for $b=1,\ldots, B$. Then, for each sample $b$, compute 
\[
    \log p(x_i \mid \hat\phi, q_i^{(b)}) = \sum_{l=1}^L \sum_{a=1}^2 \log \bigg( \sum_{j=1}^k q_{ij}^{(b)} \hat\phi_{jl} (x_{i l a}) \bigg).
\]
We would like to form a Monte Carlo approximation $\tilde{p}(x_i \mid \hat\phi,\hat\psi) = \frac{1}{B} \sum_{b=1}^B p(x_i \mid \hat\phi, q_i^{(b)})$, but we need to do so in a numerically stable way.  To this end, we use the log-sum-exp trick; that is, we compute $C := \max_b  \log  p(x_i \mid \hat\phi, q_i^{(b)})$ and then 
\[
    \log \tilde{p}(x_i \mid \hat\phi,\hat\psi) = \log \bigg(\frac{1}{B} \sum_{b=1}^B \exp\Big( \log  p(x_i \mid \hat\phi, q_i^{(b)}) - C\Big) \bigg) + C.
\]
This yields Monte Carlo approximations of the per-individual negative log-likelihood values via $\tilde\ell_k(x_i; \hat\phi,\hat\psi) = -\log \tilde{p}(x_i \mid \hat\phi,\hat\psi)$.
The additional computational burden of this procedure is not significant relative to the time required for posterior sampling with the MCMC algorithm.

\section{Alternative LaD model with diagonal covariance}\label{supp:cov} 

Our proposed LaD model (\cref{sec:method:bayesinf}) places a NIW prior on $(\mu, \Sigma)$, which uses a full covariance matrix for $\Sigma$.
In \cref{sec:examples:mvn:comparison}, we compare with an alternative LaD model that employs a diagonal covariance matrix.
We provide the details of this diagonal covariance approach here.

Assume $\Sigma = \diag(\sigma^2_1, \ldots, \sigma^2_K)$, and model each coordinate $k$ independently with a normal-inverse-gamma prior.  Specifically, for each $k$, consider modeling
\[
    Z_{ik} \mid \mu_k, \sigma^2_k \overset{\mathrm{iid}}{\sim} \cN(\mu_k, \sigma^2_k), \quad    \sigma_k^2 \sim \mathrm{InverseGamma}(a_{0k}, b_{0k}), \quad \mu_k \mid \sigma_k^2 \sim \cN\big(\mu_{0k}, \sigma_k^2/\lambda_{0k}\big),
\]
where $\mathrm{InverseGamma}(a, b)$ denotes an inverse-gamma distribution with shape $a$ and scale $b$. 
Let $\oline{Z}_{.k} = \frac{1}{n} \sum_{i=1}^n Z_{ik}$. Then the posterior is, independently for $k = 1,\ldots,K$,
\[
    \sigma^2_k \mid Z_{1:n} \sim \mathrm{InverseGamma}(a_{nk}, b_{nk}), \quad \mu_k \mid \sigma^2_k, Z_{1:n} \sim \cN \big(\mu_{nk}, \sigma^2_k/\lambda_{nk} \big),
\]
where
\[
    \lambda_{nk} = \lambda_{0k} +n, \quad \mu_{nk} = \frac{\lambda_{0k}\mu_{0k} + n \oline{Z}_{.k}}{\lambda_{nk}}, \quad a_{nk} = a_{0k} + \frac{n}{2},
\]
\[
    b_{nk} = b_{0k} + \frac{1}{2} \sum_{i=1}^n (Z_{ik} - \oline{Z}_{.k})^2 + \frac{1}{2} \frac{\lambda_{0k}n}{\lambda_{nk}} (\oline{Z}_{.k} - \mu_{0k})^2.
\]
We set the hyperparameters to match the NIW prior $\Sigma \sim \mathrm{InverseWishart} (\Psi_0, \nu_0)$ and $\mu \mid \Sigma \sim \cN(\mu_0, \Sigma/\lambda_0)$, by choosing
\[
    a_{0k} = \frac{\nu_0 - K + 1}{2}, \quad b_{0k} = \frac{\Psi_{0, kk}}{2}, \quad \mu_{0k} = (\mu_0)_k, \quad \lambda_{0k} = \lambda_0,
\]
for $k = 1,\ldots,K$.
With this choice, the NIW marginal for $\Sigma_{kk}$ match $\mathrm{IG}(a_{nk}, b_{nk})$ and yields $\mu_k \mid \sigma^2_k \sim \cN(\mu_{0k}, \sigma^2_k / \lambda_{0k})$ under independence.


\putbib 
\end{bibunit}

\end{document}